\documentclass[UKenglish,thm-restate,notab,nolineno]{socg-lipics-v2021}

\usepackage[sort, nospace, noadjust]{cite} %
\newcommand*{\lipicsPatchAmsMathEnvironmentForLineno}[1]{%
  \renewenvironment{#1}%
  {\linenomath\postdisplaypenalty=0\csname old#1\endcsname}%
  {\csname oldend#1\endcsname\endlinenomath}}%
\newcommand*\lipicsPatchBothAmsMathEnvironmentsForLineno[1]{%
  \lipicsPatchAmsMathEnvironmentForLineno{#1}%
  \lipicsPatchAmsMathEnvironmentForLineno{#1*}}%
\AtBeginDocument{%
  \lipicsPatchBothAmsMathEnvironmentsForLineno{align}%
  \lipicsPatchBothAmsMathEnvironmentsForLineno{flalign}%
  \lipicsPatchBothAmsMathEnvironmentsForLineno{alignat}%
  \lipicsPatchBothAmsMathEnvironmentsForLineno{gather}%
  \lipicsPatchBothAmsMathEnvironmentsForLineno{multline}}%

\usepackage{amsmath,amssymb,amsfonts,mathtools}
\usepackage{mathrsfs}
\usepackage[c2]{optidef}
\usepackage{datetime}
\usepackage{stmaryrd}
\usepackage{tcolorbox}
\usepackage{wasysym, fancybox}
\usepackage{dsfont}
\usepackage{eqparbox}
\usepackage{comment}
\usepackage{algorithmic,algorithm}
\usepackage{bbold}
\usepackage{verbatim}
\usepackage{xspace}
\usepackage{subcaption} %
\usepackage[normalem]{ulem}
\usepackage[section]{placeins}
\usepackage{caption}

\usepackage{etoc} 
\makeatletter
\def\etocarticlestyle{%
    \etocsettocstyle
    {\section *{\contentsname
         }
         }
    {}}
\makeatother

\usepackage{thm-restate} %
	\newtheorem{property}[theorem]{Property}
\usepackage{hyperref}
\usepackage{xcolor}
	\definecolor{darkred}{RGB}{220,50,0}
	\definecolor{lightblue}{rgb}{.80,.85,1}
	\definecolor{darkgreen}{RGB}{0,100,0}
	\definecolor{firebrick}{RGB}{178,34,34}
	\definecolor{salmon}{RGB}{250,128,114}
	\definecolor{turquoise}{RGB}{0,128,114}
	\definecolor{turquoise2}{RGB}{0,180,140}
	\definecolor{turquoise3}{RGB}{0,160,220}
	\definecolor{turquoise3}{RGB}{0,160,220}
	\definecolor{bordeaux}{RGB}{144, 12, 63}
	\definecolor{darkorchid}{rgb}{0.60,0.20,0.80}
	\definecolor{lightorange}{RGB}{234,219,173}
	\definecolor{bianca}{RGB}{0,92,153}
	\definecolor{lila}{RGB}{193,128,255}

\graphicspath{{Figures/}}

\theoremstyle{plain}

\newcommand{\ignore}[1]{}

\newcommand{\DomiEmph}[1]{{\color{darkorchid}\sethlcolor{yellow} \hl{#1}}}
\newcommand{\Domi}[1]{{\color{darkorchid} #1}}

\newcommand{\BiancasProof}[1]{}

\newcommand{\domi}[1]{\marginpar{\tiny\color{darkorchid} D: #1}}

\newcommand{\bianca}[1]{\sethlcolor{lila}\hl{#1}} 
	\soulregister\cite7
	\soulregister\xspace7
	\soulregister\M7
	\soulregister\ref7
\newcommand{\biancacmt}[1]{{\color{white}\sethlcolor{turquoise}\hl{#1}}}
\newcommand{\sidebianca}[1]{\marginpar{\tiny B: \biancacmt{#1}}}

\newcommand{\itemref}[1]{\textcolor{darkgray}{\sffamily\bfseries\upshape\mathversion{bold}\ref{#1}}\xspace}
\newcommand{\styleitem}[1]{\textcolor{darkgray}{\sffamily\bfseries\upshape\mathversion{bold}#1}\xspace}
\newcommand{\Step}[1]{\medskip{\underline{\bf Step #1:}}\hspace{1mm}}

\newcommand{\Rspace}{\ensuremath{\mathbb{R}}\xspace}
\newcommand{\Sphere}{\mathbb{S}}
\newcommand{\Conv}[1]{\operatorname{conv}(#1)}
\newcommand{\Aff}[1]{\operatorname{Aff}(#1)}
\newcommand{\Bis}[1]{\operatorname{Bis}(#1)}
\newcommand{\card}[1]{\operatorname{card}#1}
\newcommand{\Link}[2]{\operatorname{Lk}(#1,#2)}
\newcommand{\Star}[2]{\operatorname{St}(#1,#2)}
\newcommand{\UpperC}[1]{\operatorname{UpperComplex}_{\M}(#1)}
\newcommand{\LowerC}[1]{\operatorname{LowerComplex}_{\M}(#1)}
\newcommand{\Above}[1]{\operatorname{UpperFacets}_{\M}(#1)}
\newcommand{\Below}[1]{\operatorname{LowerFacets}_{\M}(#1)}
\newcommand{\NaiveVertSimp}{\hyperlink{naivealgo}{\operatorname{\tt NaiveVerticalSimplification}}}
\newcommand{\PracticalVertSimp}{\hyperlink{practicalalgo}{\operatorname{\tt PracticalVerticalSimplification}}}
\newcommand{\PracticalSquash}{\operatorname{\tt PracticalSquash}}
\newcommand{\NaiveSquash}{\operatorname{\tt NaiveSquash}}
\newcommand{\NonCrossingVertSimp}{\operatorname{\tt NonCrossingVertSimp}}
\newcommand{\NonCrossingSquash}{\operatorname{\tt NonCrossingSquash}}
\newcommand{\UpperFacets}[2]{\operatorname{UpperFacets}_{#2}(#1)}
\newcommand{\LowerFacets}[2]{\operatorname{LowerFacets}_{#2}(#1)}
\newcommand{\Bold}[1]{\mbox{\bf #1}}

\newcommand{\TanSpace}{\operatorname{Tan}}

\newcommand{\Offset}[2]{#1^{\oplus #2}}

\newcommand{\M}{\ensuremath{\mathcal{M}}\xspace}
\newcommand{\N}{\ensuremath{\pi_\M(\US{K})}\xspace}
\newcommand{\X}{\ensuremath{\mathcal{X}}\xspace}
\newcommand{\HH}{\ensuremath{\mathcal{H}}\xspace}

\newcommand{\up}[2]{\operatorname{up}_{#1}(#2)}
\newcommand{\low}[2]{\operatorname{low}_{#1}(#2)}
\newcommand{\UpperSkin}[1]{\operatorname{UpperSkin}_{\M}(#1)}
\newcommand{\LowerSkin}[1]{\operatorname{LowerSkin}_{\M}(#1)}

\newcommand{\MA}[1]{\operatorname{axis}(#1)}

\newcommand{\lfs}[1]{\operatorname{lfs}(#1)}
\newcommand{\below}{\prec_{\M}}
\renewcommand{\above}{\succ_{\M}}
\newcommand{\AboveSet}[1]{\operatorname{Above}(#1)}
\newcommand{\BelowSet}[1]{\operatorname{Below}(#1)}
\newcommand{\Side}[1]{\operatorname{Side}(#1)}
\newcommand{\alt}[1]{\operatorname{alt}_{\M}(#1)}

\newcommand{\Nerve}[1]{\operatorname{Nerve}(#1)}
\newcommand{\Vertexset}[1]  {\operatorname{Vert}{#1}}
\newcommand{\Del}[1]{\operatorname{Del}(#1)}
\newcommand{\DelR}[1]{\operatorname{Del}_{\M}(#1)}
\newcommand{\DelC}[1]{\operatorname{CoreDel}_{\M}(#1)}

\newcommand{\US}[1]{\boldsymbol{\lvert}#1\boldsymbol{\rvert}}
\newcommand{\Interior}[1]{\operatorname{relint}(#1)}
\newcommand{\Closure}[1]{\operatorname{closure}(#1)}

\newcommand{\Cl}[1]{\operatorname{Cl}#1}

\newcommand{\Reach}[1]{\operatorname{Reach}(#1)}

\newcommand{\Tangent}[2]{\mathbf{T}_{#1}#2}

\newcommand{\Normal}[2]{\mathbf{N}_{#1}#2}
\newcommand{\normal}[1]{\mathbf{n}({#1})}

\newcommand{\spx}{\gamma}

\newcommand{\reach}{\ensuremath{\mathcal{R}}\xspace}
\newcommand{\rtube}{\ensuremath{r}\xspace}
\newcommand{\rcirc}[1]{\rho({#1})}

\newcommand{\Nintegers}{\mathbb{N}}

\title{
  When alpha-complexes collapse onto codimension-1 submanifolds 
}

\author{Dominique Attali}{Universit{\'e} Grenoble Alpes, CNRS, Grenoble INP, GIPSA-lab, Grenoble, France}{Dominique.Attali@grenoble-inp.fr}{https://orcid.org/0000-0003-4808-6301}{}

\author{Matt{\'e}o Cl{\'e}mot}{Universite Claude Bernard Lyon 1, CNRS, INSA Lyon, LIRIS, Villeurbanne, France}{matteo.clemot@ens-lyon.fr}{https://orcid.org/0009-0000-2524-0244}{}

\author{Bianca B. Dornelas}{Institute of Geometry, TU Graz, Austria \and Institute for Medical Informatics, Statistics and Documentation, MedUni Graz, Austria}{contact@bdornelas.com}{https://orcid.org/0000-0002-4827-4663}{Funded by the Austrian Science Fund (FWF), grant W1230.}

\author{Andr{\'e} Lieutier}{No affiliation, Aix-en-Provence, France}{andre.lieutier@gmail.com}{}{}

\authorrunning{D. Attali, M. Cl\'emot, B. Dornelas, and A. Lieutier} %

\Copyright{Dominique Attali, Matt\'eo Cl\'emot, Bianca Dornelas, and Andr\'e Lieutier} %

\ccsdesc{Theory of computation $\rightarrow$ Computational geometry}%

\keywords{Submanifold reconstruction, triangulation, abstract simplicial complexes, collapses, convexity}%

\nolinenumbers %

\hideLIPIcs

\begin{document}
\maketitle

\begin{abstract}
    Given a finite set of points $P$ sampling an unknown smooth
    surface $\M \subseteq \Rspace^3$, our goal is to triangulate $\M$
    based solely on $P$. Assuming $\M$ is a smooth orientable
    submanifold of codimension $1$ in $\Rspace^d$, we introduce a
    simple algorithm, \emph{Naive Squash}, which simplifies the
    $\alpha$-complex of $P$ by repeatedly applying a new type of
    collapse called \emph{vertical} relative to $\M$. Naive Squash
    also has a practical version that does not require knowledge of
    $\M$. We establish conditions under which both the naive and
    practical Squash algorithms output a triangulation of $\M$. We
    provide a bound on the angle formed by triangles in the
    $\alpha$-complex with $\M$, yielding sampling conditions on $P$
    that are competitive with existing literature for smooth surfaces
    embedded in $\Rspace^3$, while offering a more compartmentalized
    proof. As a by-product, we obtain that the restricted Delaunay
    complex of $P$ triangulates \M when \M is a smooth surface in
    $\Rspace^3$ under weaker conditions than existing ones.
\end{abstract}

\etocdepthtag.toc{mtchapter}

\clearpage
\section{Introduction}
\label{section:introduction}

Given a finite set of points $P$ that sample an unknown smooth surface
$\M \subseteq \Rspace^3$ (example in
Figure~\ref{fig:croissants-left}), we aim to approximate $\M$ based
solely on $P$. This problem, known as {\em surface reconstruction},
has been widely
studied~\cite{ohrhallinger2021surveycurves,bjerkevik2022reconstructEpsilonSample,hoppe1992surface,bernardini1999ball,alexa2001point,carr2001reconstruction,kazhdan2006poisson,giraudot2013noise}. Several
algorithms based on computational geometry have been developed, such
as Crust~\cite{amenta1999surface},
PowerCrust~\cite{Amenta_Choi_Kolluri:2001:powerCrust},
Cocone~\cite{Amenta_Choi_Dey_Leekha:2022:simple_algo},
Wrap~\cite{edels:2003:wrap} and variants based on flow
complexes~\cite{giesen2002surface,giesen2003flow,dey2008critical,bauer2024wrapping}. These
algorithms rely on the Delaunay complex of $P$ and offer theoretical
guarantees, summarized in~\cite{Dey:2006:reconstruction_book}.

The most desirable guarantee is that the reconstruction outputs a
triangulation of \M, that is, a simplicial complex whose support is
{\em homeomorphic} to $\M$, in which case we call the algorithm {\em
  topologically correct}.  That has been established for many of the
aforementioned algorithms, assuming that $P$ is noiseless ($P
\subseteq \M$) and sufficiently dense. Specifically, let $\reach > 0$
be a lower bound on the reach of $\M$, and $\varepsilon \geq 0$ an
upper bound on the distance between any point of $\M$ and its nearest
point in $P$.  Both Crust and Cocone are topologically correct under
the condition $\frac{\varepsilon}{\reach} \leq
0.06$~\cite{dey2017curve}, which, to our knowledge,
{is the weakest such constraint}
guaranteeing topological correctness
  for surface reconstruction algorithms in $\Rspace^3$.
	
Surface reconstruction generalizes to approximating an unknown smooth
submanifold $\M \subseteq \Rspace^d$ from a finite sample $P$. One
approach in that case, similar to the Wrap algorithm in
$\mathbb{R}^3$, involves \emph{collapses}, which are typically applied
to complexes like the {\em $\alpha$-complex} \cite{bauer2017morse}. The $\alpha$-complex of
$P$~\cite{edelsbrunner1994three, edelsbrunner2011alpha,
  Edels-Harer_Comput-Topo} includes simplices whose circumspheres have
radius $\leq \alpha$ and enclose no other points of
$P$~\cite{Edels_Kirkpatrick_Seidel:1983:introducce_alpha_shape}. For
well-chosen $\alpha$, the $\alpha$-complex has the same homotopy type
as
$\M$~\cite{niyogi08:_findin_homol_of_subman_with,chazal2009sampling,chazal08:_smoot_manif_recon_from_noisy,socg24-NSW},
provided that $P$ is sufficiently dense and has low noise relative to the
reach of $\M$. However, it may still fail to capture the topology of
$\M$, as illustrated in Figure~\ref{fig:croissants-middle}: for $\M
\subseteq \Rspace^3$, the $\alpha$-complex of $P$ includes {\em
  slivers}, tetrahedra that have one dimension more than $\M$,
preventing the existence of a homeomorphism. Slivers complicate
reconstructing $k$-dimensional submanifolds in $\Rspace^d$ for $k \geq
2$ for all Delaunay-based reconstruction attempts.

\subparagraph{Contributions.}

We introduce a simple algorithm, $\NaiveVertSimp$, which takes
  as input a simplicial complex $K$ and simplifies it by applying
  collapses guided by the knowledge of \M. We call it naive because
  this knowledge is non-realistic in practice. We find conditions
  under which the algorithm is topologically correct for smooth
orientable submanifolds $\M$ of $\Rspace^d$ with codimension one. Its
variant, $\PracticalVertSimp$, does not rely on $\M$ and remains
topologically correct, though it requires stricter
conditions. When applying both algorithms to the $\alpha$-complex of
$P$ and returning the result, we obtain two reconstruction algorithms
which we refer to as $\NaiveSquash$ and $\PracticalSquash$,
respectively.  We determine conditions on the inputs $P$ and $\alpha$
that guarantee the topological correctness of these squash
algorithms. Moreover, for $d=3$, we show that $\PracticalSquash$ is
correct under the sampling condition $\frac{\varepsilon}{\reach} \leq
0.178$ (see Figure~\ref{fig:croissants-right} for an example output),
while $\NaiveSquash$ is correct for $\frac{\varepsilon}{\reach} \leq
0.225$, assuming suitable choice of $\alpha$.  We also show that the
{\em restricted Delaunay
  complex}~\cite{boissonnat_Oudot_dyer_ghosh-restricted_del} is
generically homeomorphic to $\M$ when $\frac{\varepsilon}{\reach} \leq
0.225$.

In addition, while proving these results, we derive an upper bound for
when triangles with vertices on a smooth submanifold
$\M\subseteq \Rspace^d$ form a small angle with $\M$: for a triangle
$abc$ with $a,b,c \in \M$, longest edge $bc$, and circumradius $\rho$,
we show that the angle between the affine space spanned by $abc$ and
the tangent space to $\M$ at $a$ satisfies:
\begin{align}
  \label{eq:contrib-angle}
  \sin \angle \Aff{abc}, \Tangent a \M \, &\leq \, \frac{\sqrt{3}\, \rho}{\reach}.
\end{align}

\begin{figure}
  \centering
  \begin{subfigure}[b]{0.5em}
	\subcaption{}\label{fig:croissants-left}
  \end{subfigure}%
  \raisebox{-0.1cm}{
	\includegraphics[width=.28\linewidth]{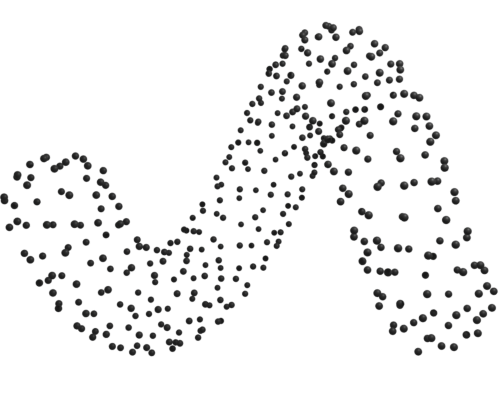}
  }
  \hspace{0.1cm}
  \begin{subfigure}[b]{0.5em}
	\subcaption{}\label{fig:croissants-middle}
  \end{subfigure}%
  \raisebox{-0.1cm}{
	\includegraphics[width=.28\linewidth]{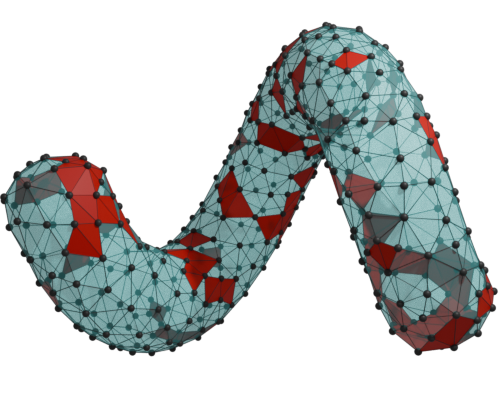} 
  }
  \hspace{0.1cm}
  \begin{subfigure}[b]{0.5em}
	\subcaption{}\label{fig:croissants-right}
  \end{subfigure}%
  \raisebox{-0.1cm}{
	\includegraphics[width=.28\linewidth]{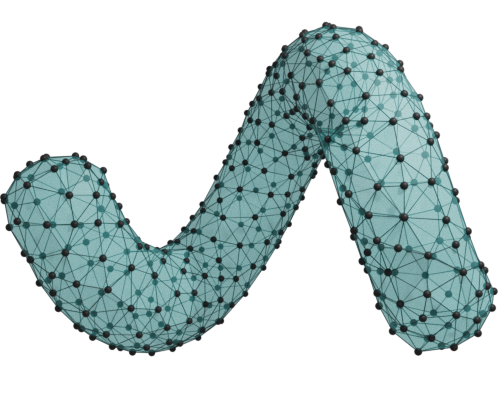}
  }
  \caption{Points sampling a surface in $\mathbb{R}^3$ (a) with the
    corresponding $\alpha$-complex, where tetrahedra are highlighted
    (b). Applying Practical Squash with parameter $\alpha$ outputs
    (c). }
  \label{fig:collapsing_croissants}
\end{figure}
	
\subparagraph{Techniques.}
Our proof of correctness for the squash algorithms is more
compartmentalized than the ones present in the literature: we first
consider a smooth orientable submanifold $\M \subseteq \Rspace^d$ with
codimension one and a general simplicial complex $K$ embedded in
$\Rspace^d$ and contained within a small tubular neighborhood of
\M. We introduce {\em vertical} collapses (relative to $\M$) in $K$
which remove $d$-simplices of $K$ that either have no $d$-simplices of
$K$ above them in directions normal to $\M$ or no $d$-simplices of $K$
below them in directions normal to $\M$.  $\NaiveVertSimp$ iteratively
applies vertical collapses relative to $\M$.  $\PracticalVertSimp$
does not depend on knowledge of $\M$ and applies vertical collapses
relative to a hyperplane, constructed dynamically based on the simplex
currently considered for collapse.
	
We examine conditions for the correctness of these algorithms.  Apart
from the requirement that $K$ has no vertical $i$-simplices relative
to $\M$ for $0 < i < d$ and that its support projects onto $\M$ and
fully covers it, we require the {\em vertical convexity} of $K$
relative to $\M$. This means that each normal line to $\M$ at a point
$m$ (restricted to a small ball around $m$) intersects the support of
$K$ in a convex set. For $\PracticalVertSimp$, an additional
requirement is that the $(d-1)$-simplices of $K$ must form an angle of
at most $\frac{\pi}{4}$ with $\M$.
	
Afterwards, we present $\PracticalSquash$ and $\NaiveSquash$, which
initialize the previous algorithms with $K$ as the $\alpha$-complex of
a point set $P \subseteq \mathbb{R}^d$ that samples $\M$. We show that
correctness is guaranteed when the $i$-simplices in the $\alpha$-complex form small angles
with $\M$ for $0 < i < d$. We provide explicit upper bounds for these
angles, expressed in terms of $\varepsilon$, $\delta$, and $\alpha$,
where $\varepsilon$ and $\delta$ control the sample density and noise
in $P$.
	
We analyze the case $d=3$ and provide numerical bounds on the ratios
$\frac{\varepsilon}{\reach}$ and $\frac{\alpha}{\reach}$ that ensure
the correctness of both squash algorithms. Instrumental to this step,
we derive \eqref{eq:contrib-angle} which enables us to upper bound the
angles of triangles in the $\alpha$-complex relative to the manifold
and is of independent interest.

\subparagraph{Related work.}
The Squash algorithms (both practical and naive) are similar to Wrap
\cite{edelsbrunner1994three,bauer2017morse} in that they compute a subcomplex $K$ of the Delaunay complex
of $P$ and then perform a sequence of collapses. However, the
selection of $K$ and the nature of the collapses differ between the
two methods: in the naive squash, definitions are relative to $\M$,
unlike Wrap, which uses flow lines derived from $P$ to guide the
collapsing sequence. This distinction allows us to address the general
case first and then focus on the specific case of the
$\alpha$-complex.  Moreover, while we guarantee correctness for a
larger interval of the ratio $\frac{\varepsilon}{\reach}$ compared to
previous literature, most existing work addresses non-uniform sampling
cases, whereas our work focuses on uniform sampling.

The vertical convexity assumption, crucial for the correctness of our
algorithms, has been employed in various forms to establish
collapsibility of certain classes of simplicial complexes
\cite{adiprasito2020barycentric,chen2019neural,attali2019convexity}.
Similarly, bounding the angle between the manifold and the
  simplices used for reconstruction has been essential in prior work
\cite{Dey:2006:reconstruction_book,cheng2013delaunay,cheng2005manifold,AttaliLieutierFlatDelaunay2022,Attali_Lieutier:22:Del_Like_Triang}.
Our bound \eqref{eq:contrib-angle} remains true when replacing
$\reach$ with the local feature size of $a$, as explained in
Appendix~\ref{appendix:angles}. The thus modified bound improves upon
the known bound~\cite[Lemma 3.5]{Dey:2006:reconstruction_book}.

At last, for $d=3$, we show in
Appendix~\ref{appendix:restricted-Delaunay-complex}
that the restricted Delaunay complex is generically homeomorphic to $\M$ for
$\frac{\varepsilon}{\reach} \leq 0.225$. In contrast, it is proven to
be homeomorphic to $\M$ only if $\frac{\varepsilon}{\reach} \leq
0.09$~\cite[Theorem 13.16]{cheng2013delaunay}, a result based on the
Topological Ball Theorem~\cite[Theorem 13.1]{cheng2013delaunay}. Our
proof bypasses this requirement, relying instead on $\NaiveSquash$.

\subparagraph{Outline.}
After the preliminaries in Section~\ref{section:preliminaries},
Section~\ref{section:vertically-convex-simplicial-complexes} defines
vertically convex simplicial complexes. We introduce the concepts of
upper and lower skins for these complexes and prove that both are
homeomorphic to their orthogonal projection onto
$\M$. Section~\ref{section:vertical-collapses} presents general
conditions under which a simplicial complex $K$ can be transformed
into a triangulation of $\M$ through either Naive or Practical
vertical simplification. Section~\ref{section:alpha-complexes}
provides conditions ensuring the topological correctness of both the
naive and practical Squash algorithms and the restricted Delaunay
complex.  All missing proofs can be found in the appendices.

\section{Preliminaries}
\label{section:preliminaries}
\paragraph*{Subsets and submanifold.}

Given a subset
$X \subseteq \Rspace^d$, we define several important geometric concepts. The convex hull of $X$ is denoted as $\Conv{X}$ and the affine space spanned by $X$ as $\Aff{X}$. The interior of $X$ is denoted as $X^\circ$. The \emph{relative interior} of $X$, denoted as $\Interior{X}$, represents the interior of $X$ within $\Aff{X}$. 
For any point $x$ and radius $r$, we denote the closed ball with center $x$ and radius $r$ as $B(x,r)$. The $r$-offset of $X$, denoted as $\Offset{X}{r}$, is the union of closed balls centered at each point in $X$ with radius $r$: $\Offset X r = \bigcup_{x \in X} B(x,r)$.
The \emph{medial axis} of $X$, denoted as $\MA{X}$, is the set of points in $\Rspace^d$ that have at least two nearest points in $X$.  The \emph{reach} of $X$, denoted as $\Reach{X}$, is the infimum of distances between $X$ and $\MA{X}$. 
Furthermore, we define the projection map $\pi_X: \Rspace^d \setminus \MA{X} \to X$, which associates each point $x$ with its unique closest point in $X$. This projection map is well-defined on every subset of $\Rspace^d$ that does not intersect $\MA{X}$, particularly on every $r$-offset of $X$ with $r < \Reach{X}$.

\begin{tcolorbox}
Throughout the paper, we designate $\M$ as a compact $C^2$ submanifold
of $\Rspace^d$ of codimension one, and, therefore, orientable (see
e.g. \cite{samelson1969orientability}).
\end{tcolorbox}
Given $m \in \M$, we denote the affine tangent space to $\M$ at $m$ as
$\Tangent m \M$ and the affine normal space as $\Normal m \M$. As $\M$
has codimension one, $\Tangent m \M$ is a hyperplane and $\Normal m
\M$ is a line. Additionally, since $\M$ is $C^2$, it has a positive
reach~\cite{scholtes2013hypersurfaces}.  For all real numbers $\rtube$
such that $0 < \rtube < \Reach \M$, the $\rtube$-offset of $\M$ can be
partitioned into the set of normal segments $\{\Normal m \M \cap
B(m,\rtube)\}_{m \in \M}$~\cite{doC76}, that is,
\[
\Offset \M \rtube = \dot{\bigcup}_{m \in \M} \Normal m \M \cap B(m,\rtube).
\]

\begin{tcolorbox}
  We define $\mathbf{n}: \M \to \Rspace^d$ as a
  differentiable field of unit normal vectors of $\M$~\cite{doC76}. We
  let $\reach$ be a finite arbitrary number such that $0 < \reach \leq
  \Reach\M$, fixed throughout.
\end{tcolorbox}

\paragraph*{Abstract simplicial complexes and collapses.}

We recall some classical definitions of algebraic topology
\cite{munkres1993elements, Edels-Harer_Comput-Topo}. An {\em abstract
  simplicial complex} is a collection $K$ of finite non-empty sets
with the property that if $\sigma$ belongs to $K$, so does every
non-empty subset of $\sigma$. Each element $\sigma$ of $K$ is called
an {\em abstract simplex} and its {\em dimension} is one less than its
cardinality: $\dim \sigma = \card \sigma - 1$. A simplex of dimension
$i$ is called an $i$-simplex and the set of $i$-simplices of $K$ is
denoted as $K^{[i]}$.  If $\tau$ and $\sigma$ are two simplices such
that $\tau \subseteq \sigma$, then $\tau$ is called a \emph{face} of
$\sigma$, and $\sigma$ is called a \emph{coface} of
$\tau$.
The $(d-1)$-dimensional faces
of $\sigma$ are the {\em facets} of
$\sigma$. The {\em vertex set} of $K$ is
$\Vertexset K = \bigcup_{\sigma \in K} \sigma$.  A {\em subcomplex}
$L$ of $K$ is a simplicial complex whose elements belong to $K$. The
     {\em link} of $\sigma$ in $K$, denoted $\Link \sigma K$, is the
     set of simplices $\tau$ in $K$ such that $\tau \cup \sigma \in K$
     and $\tau \cap \sigma = \emptyset$. It is a subcomplex of
     $K$. The {\em star} of $\sigma$ in $K$, denoted as $\Star \sigma
     K$, is the set of cofaces of $\sigma$.  The simplicial complex
     formed by all the faces of $\sigma$ is the \emph{closure} of
     $\sigma$, $\Cl{\sigma}$.

Consider next an abstract simplex $\sigma \subseteq \Rspace^d$.  One
can associate it to the geometric simplex
$\Conv \sigma \subseteq \Rspace^d$, called the {\em support} of $\sigma$. In
general, $\dim(\Aff\sigma) \leq \dim(\sigma)$ and we say that $\sigma$
is {\em non-degenerate} whenever $\dim(\Aff \sigma) = \dim\sigma$.
Given a simplicial complex $K$ with vertices in $\Rspace^d$, we say
that $K$ is {\em canonically embedded} if the following two
conditions are satisfied:
\begin{enumerate}
\item $\dim \sigma = \dim (\Aff \sigma)$ for all $\sigma \in K$;
\item $\Conv {\alpha \cap \beta} = \Conv \alpha \cap \Conv \beta$ for all $\alpha,\beta \in K$.
\end{enumerate}

\begin{tcolorbox}
In this paper we consider exclusively abstract simplicial complexes $K$ with vertex sets in $\Rspace^d$ and which are canonically embedded.
\end{tcolorbox}

Given such a simplicial complex, its \emph{underlying space} (or \emph{support}) is the point set $\US{K} = \bigcup_{\sigma \in K} \Conv \sigma$.
\hypertarget{complexhere}{}
 If $\US{K}$ is homeomorphic to \M, then $K$ is called a {\em triangulation} of \M or is said to triangulate \M. 
Since $K$ is canonically embedded, the link of every $i$-simplex
of $K$ falls into one of the following two categories: (1) it is a
triangulation of the sphere of dimension $d-i-1$ or (2) it is a
proper\footnote{A {\em proper} subset $A$ of $B$ is such that $A \neq
B$. } subcomplex of such a triangulation. The {\em boundary complex}
of a simplicial complex $K$ is the subset of simplices in the second
category, denoted $\partial K$, and it holds that $\US{\partial K} =
\partial \US{K}$. Simplices in $\partial K$ are referred to as {\em
  boundary simplices} of $K$. Given a set of abstract simplices
$\Sigma$, if $\sigma \in \Sigma$ has no coface in $\Sigma$ besides
itself, then $\sigma$ is said to be {\em inclusion-maximal} in
$\Sigma$.

Suppose that $\tau \in K$ is a simplex whose star in $K$ has a unique
inclusion-maximal element $\sigma \neq \tau$. Then $\tau$ is said to
be {\em free} in $K$. Equivalently, $\tau$ is free in $K$ if and only
if the link of $\tau$ in $K$ is the closure of a simplex. 
Consequently,
free simplices of $K$ are always
boundary simplices of $K$. However, not all boundary simplices of $K$ are necessary free. 
There are instances where none of them are free, such as the famous example 
when $K$ triangulates the 2-dimensional subspace of $\Rspace^3$, known as the ``house with two rooms''.  A {\em collapse}
in $K$ is the operation that removes from $K$ a free simplex $\tau$
along with all its cofaces. This operation is known to preserve the
homotopy-type of $\US{K}$.

\paragraph*{Delaunay complexes, $\alpha$-complexes, and $\alpha$-shapes.}

Consider a finite collection of points $P \subseteq \Rspace^d$. The
Voronoi region of $q\in P$ is the collection of points $x
\in \Rspace^d$ that are closer to $q$ than to any other points of $P$:
\[
V(q,P) = \{ x \in \Rspace^d \mid \| x - q \| \leq \| x - p \|, \text{ for all $p \in P$} \}.
\]
Given a subset $\sigma \subseteq P$, 
let $V(\sigma,P) = \bigcap_{q
  \in \sigma} V(q,P)$. The {\em Delaunay complex} is defined as
\[
\Del P = \{  \sigma \subseteq P \mid  \sigma \neq \emptyset \text{ and } V(\sigma,P) \neq \emptyset \}.
\]
A simplex $\sigma \in \Del{P}$ is called a {\em Delaunay simplex} of $P$ and it
is {\em dual} to its corresponding {\em Voronoi cell}
$V(\sigma,P)$. 
Henceforth, we assume that the set of points $P$ is in
{\em general position}. This means that no $d+2$ points of $P$ lie on the same
$d$-dimensional sphere and no $k+2$ points of $P$ lie on the same
$k$-dimensional flat for $k<d$. In that case, $\Del P$ is canonically embedded
\cite{Fortune_Del-triang}.  For $\alpha \geq 0$, the $\alpha$-complex
of $P$ is the subcomplex of $\Del{P}$ defined by:
\[
\Del{P,\alpha} = \{ \sigma \subseteq P \mid  \sigma \neq \emptyset \text{ and }  V(\sigma,P) \cap \Offset P \alpha \neq \emptyset \}.
\]
Its underlying space $\US{\Del{P,\alpha}} = \bigcup_{\sigma \in
  \Del{P,\alpha}} \Conv\sigma$ is called the {\em $\alpha$-shape} of
$P$. It
has the properties: (i) $\US{\Del{P,\alpha}} \subseteq
\Offset P \alpha$ and (ii) $\US{\Del{P,\alpha}}$ is homotopy
equivalent to $\Offset P \alpha$; see \cite{edelsbrunner2011alpha} for
more details.

\begin{figure}[htb]
  \begin{center}
    \includegraphics{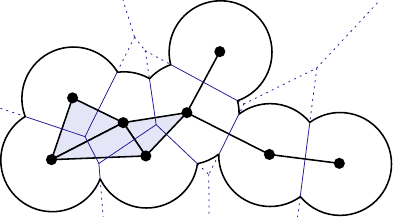}\hfill
    \includegraphics{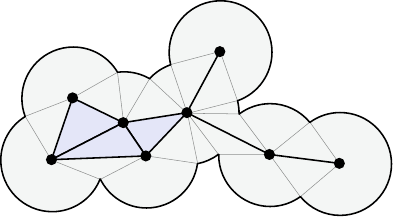}
  \end{center}
  \caption{Left: $P$ is such that neither $\Offset P \alpha$ nor
    $\Del{P,\alpha}$ are vertically convex relative to a horizontal
    line.  Right: Decomposition of $\Offset P \alpha \setminus
    \US{\Del{P,\alpha}}^\circ$ in joins as described in \cite{edelsbrunner2011alpha}.\label{figure:alpha-complex-simple} }
\end{figure}

\section{Vertically convex simplicial complexes}
\label{section:vertically-convex-simplicial-complexes}

In this section, we define the concept of vertical convexity relative
to \M\ for both a set and a simplicial complex. We then study the
boundary of a vertically convex simplicial complex $K$.
Specifically, we divide the boundary of its underlying space into an
upper and a lower skins, enabling us to identify two boundary
subcomplexes: an upper and a lower ones. Furthermore, we show that
each of these subcomplexes triangulates the orthogonal projection of
$\US{K}$ onto \M
(Lemma~\ref{lemma:upper-and-lower-complex-equal-boundary}). We also
extend the definitions for a single $d$-simplex.

\begin{definition}[Vertical convexity]
  \label{definition:vertical-convexity}
  A set $X \subseteq \Rspace^d$ is  \emph{vertically convex} relative to $\M$ if $\exists\, \rtube \in [0, \Reach{\M})$ such that
  \begin{enumerate}
  \item $X \subseteq \Offset \M \rtube$ and
  	\item $\forall\, m \in \M$, $\Normal m \M \cap B(m,\rtube) \cap X$ is convex. 
  \end{enumerate}
  In other words, for any $m \in \M$, the set $\Normal m \M \cap
  B(m,\rtube) \cap X$ is either empty or a line segment (possibly of
  zero-length). A simplicial complex $K$ is {\em vertically convex}
  relative to \M if its underlying space $\US{K}$ is.
\end{definition}

\begin{figure}[htb]
  \def\svgwidth{.9\linewidth}
  \centering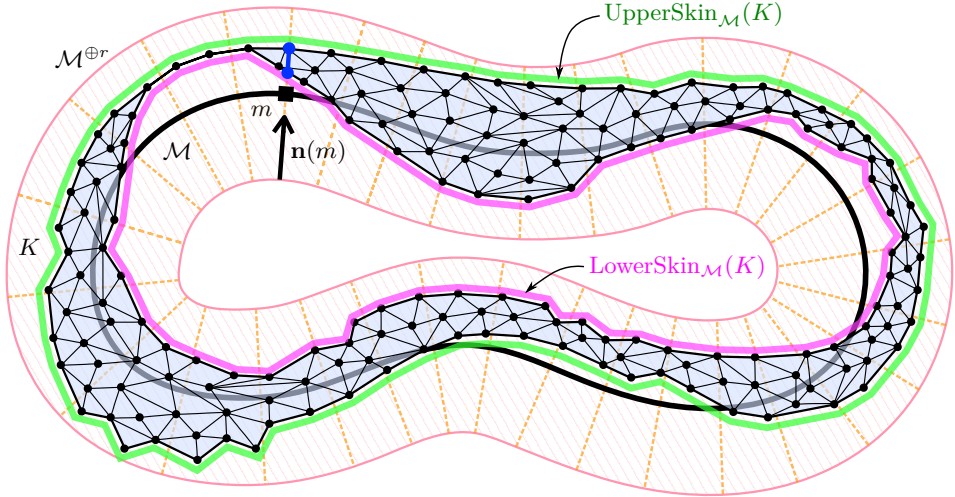
  \caption{
    A simplicial complex $K$ vertically convex relative to the curve \M. Each segment $\Normal m \M \cap B(m,r)$ (in dashed orange) intersects $\US{K}$ in a line segment, as highlighted (blue) for the point $m$ (represented by a black square). Lemma~\ref{lemma:upper-and-lower-skins-homeomorphic-to-manifold} shows that each of the two skins of $\US{K}$, depicted in green and pink according to the labeling arrows, is homeomorphic to \M.} \label{figure:upper-lower-skins}
\end{figure}

Examples of a non-vertically convex and a vertically convex simplicial complexes
are provided in Figures~\ref{figure:alpha-complex-simple} and
\ref{figure:upper-lower-skins}, respectively.

\subsection{Upper and lower skins}
\label{section:upper-and-lower-skins}

Assume that $X \subseteq \Rspace^d$ is vertically convex relative to \M and let $m \in
\pi_\M(X)$.  The endpoints (possibly equal) of the segment $\Normal m
\M \cap B(m,\rtube) \cap X$ are denoted by $\low X m$ and $\up X m$,
with $\up X m$ being above $\low X m$ along the direction of the unit
normal vector $\normal m$. With this notation, $X$ can be expressed as
a union of disjoint normal segments:
\[
X  = \dot{\bigcup}_{m \in \pi_\M(X)} [\low X m, \up X m].
\]
The {\em upper skin} and {\em lower
  skin} of $X$ are, respectively:
\begin{align*}
  \UpperSkin X &= \{ \up X m \mid m \in \pi_\M(X) \},\\
  \LowerSkin X &= \{ \low X m \mid m \in \pi_\M(X) \}.
\end{align*}
Figure~\ref{figure:upper-lower-skins} displays an example.
Our goal is to study the skins of $\US{K}$, for which we need two extra definitions.

\begin{definition}[Vertical simplex]
  A simplex $\sigma \subseteq \Rspace^d$ such that $\Conv\sigma
    \subseteq \Rspace^d \setminus \MA{\M}$ is {\em vertical} relative
  to \M if there exists a pair of distinct points in $\Conv \sigma$
  sharing the same projection onto \M.
\end{definition}

\begin{definition}[Non-vertical skeleton]
  Assume that $\US{K}  \subseteq \Rspace^d \setminus \MA{\M}$.  We say that $K$ has a {\em non-vertical skeleton} relative to \M if $K$ contains no vertical $i$-simplices relative to \M for all integers $0 < i <d$.
\end{definition}

The next lemma is a key property of vertically convex
simplicial complexes:

\begin{restatable}{lemmma}{SkinsHomeomorphicToManifold}
  \label{lemma:upper-and-lower-skins-homeomorphic-to-manifold}
  Suppose that $K$ is vertically convex and has a non-vertical
  skeleton relative to \M. Then, the upper and lower skins of
  $\US{K}$ are closed sets, each homeomorphic to $\pi_\M(\US{K})$.
  The homeomorphism is realized in both cases by $\pi_\M$. In addition,
  \begin{equation}
    \label{eq:boundary-equals-upper-and-lower}
  \partial \US{K} = \UpperSkin {\US{K}} \cup \LowerSkin {\US{K}}.
  \end{equation}
\end{restatable}

A simple consequence follows:

\begin{restatable}{lemma}{UpperEqualsLowerConsequence}
  \label{lemma:upper-equals-lower-consequence}
  Let $K$ be vertically convex and with non-vertical skeleton
  relative to $\M$. If $\UpperSkin {\US{K}} = \LowerSkin {\US{K}}$,
  then $K = \partial K$.
\end{restatable}

Aiming for a simplicial version of Equation
\eqref{eq:boundary-equals-upper-and-lower}, we define the {\em upper
  complex} of $K$ and the {\em lower complex} of $K$ relative to \M
as follows:
\begin{align*}
  \UpperC K &= \{ \nu \subseteq \partial K \mid  \Conv{\nu} \subseteq \UpperSkin{\US{K}} \}, \\
  \LowerC K &= \{ \nu \subseteq \partial K \mid  \Conv{\nu} \subseteq \LowerSkin{\US{K}} \}.
\end{align*}
By construction, both are subcomplexes of $\partial K$.  A
combinatorial equivalent of
Lemma~\ref{lemma:upper-and-lower-skins-homeomorphic-to-manifold} is:

\begin{restatable}{lemmma}{uppLowComplex}
  \label{lemma:upper-and-lower-complex-equal-boundary}
  Let $K$ be vertically convex with non-vertical skeleton relative to $\M$. Then,
  \begin{align*}
    &\US{\UpperC K} = \UpperSkin{\US{K}},\\
    &\US{\LowerC K} = \LowerSkin{\US{K}} \,\,\, \text{ and } \\
    &\partial K = \UpperC K \, \cup \, \LowerC K.
  \end{align*}
  Moreover, if $\pi_\M(\US{K}) = \M$, both $\UpperC K$ and $\LowerC K$ are triangulations of $\M$.
\end{restatable}

\subsection{Upper and lower facets of a $d$-simplex}
\label{section:upper-and-lower-facets}

Let $\sigma$ be a non-degenerate $d$-simplex of $\Rspace^d$ such that $\Conv \sigma \subseteq \Offset \M \rtube$ for some $\rtube < \Reach{\M}$. In that case,  $\Cl \sigma$ is embedded and
vertically convex relative to $\M$. The facets of $\sigma$ can be partitioned into \emph{upper facets} and \emph{lower facets} of $\sigma$ relative to $\M$ as follows:
\begin{align*}
  \Above \sigma &= \{ \nu \text{ facet of } \sigma \mid \nu \in \UpperC{\Cl \sigma} \} \\
  \Below \sigma &= \{ \nu \text{ facet of } \sigma \mid \nu \in \LowerC{\Cl \sigma} \}.
\end{align*}
An example can be seen in Figure~\ref{figure:simplex}, where one can also observe the following property:

\begin{figure}[htb]
  \def\svgwidth{.6\linewidth}
  \centering%% Creator: Inkscape 1.3.2 (1:1.3.2+202311252150+091e20ef0f), www.inkscape.org
%% Accompanies image file 'simplex-modified.pdf' (pdf, eps, ps)
\begingroup%
  \makeatletter%
  \providecommand\color[2][]{%
    \errmessage{(Inkscape) Color is used for the text in Inkscape, but the package 'color.sty' is not loaded}%
    \renewcommand\color[2][]{}%
  }%
  \providecommand\transparent[1]{%
    \errmessage{(Inkscape) Transparency is used (non-zero) for the text in Inkscape, but the package 'transparent.sty' is not loaded}%
    \renewcommand\transparent[1]{}%
  }%
  \providecommand\rotatebox[2]{#2}%
  \newcommand*\fsize{\dimexpr\f@size pt\relax}%
  \newcommand*\lineheight[1]{\fontsize{\fsize}{#1\fsize}\selectfont}%
  \ifx\svgwidth\undefined%
    \setlength{\unitlength}{373.67147827bp}%
    \ifx\svgscale\undefined%
      \relax%
    \else%
      \setlength{\unitlength}{\unitlength * \real{\svgscale}}%
    \fi%
  \else%
    \setlength{\unitlength}{\svgwidth}%
  \fi%
  \global\let\svgwidth\undefined%
  \global\let\svgscale\undefined%
  \makeatother%
  \begin{picture}(1,0.68651373)%
    \lineheight{1}%
    \setlength\tabcolsep{0pt}%
    \put(0,0){\includegraphics[width=\unitlength,page=1]{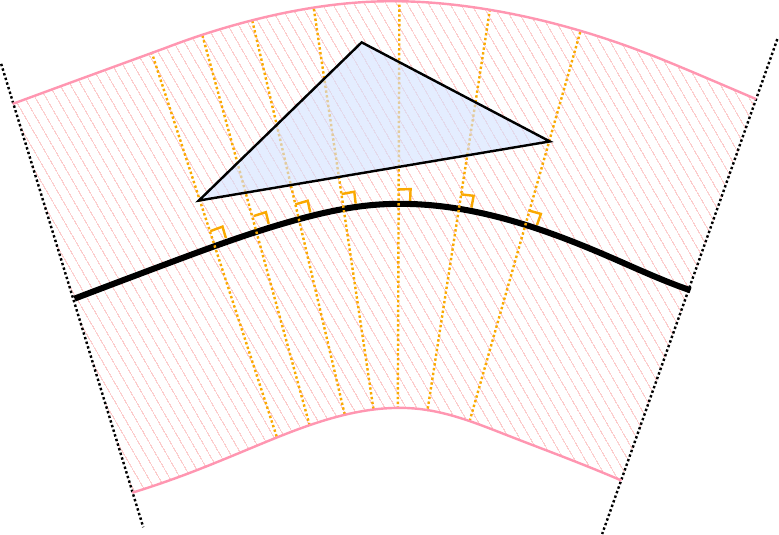}}%
    \put(0.10678879,0.34372967){\makebox(0,0)[lt]{\lineheight{1.25}\smash{\begin{tabular}[t]{l}\M\end{tabular}}}}%
    \put(0.73025487,0.33183589){\color[rgb]{0,0,0}\makebox(0,0)[lt]{\lineheight{1.25}\smash{\begin{tabular}[t]{l}$m$\end{tabular}}}}%
    \put(0.47255693,0.53308083){\makebox(0,0)[lt]{\lineheight{1.25}\smash{\begin{tabular}[t]{l}$\sigma$\end{tabular}}}}%
    \put(0,0){\includegraphics[width=\unitlength,page=2]{simplex-modified.pdf}}%
    \put(0.81039637,0.48523032){\makebox(0,0)[lt]{\lineheight{1.25}\smash{\begin{tabular}[t]{l}$\normal m$\end{tabular}}}}%
    \put(0.02161277,0.607475){\makebox(0,0)[lt]{\lineheight{1.25}\smash{\begin{tabular}[t]{l}$\Offset \M r$\end{tabular}}}}%
    \put(0,0){\includegraphics[width=\unitlength,page=3]{simplex-modified.pdf}}%
  \end{picture}%
\endgroup%

  \caption{Upper (smooth green edges) and lower (dotted pink edge) facets of a $2$-simplex $\sigma\subseteq\Rspace^2$.}
  \label{figure:simplex}
\end{figure}

\begin{restatable}{lemmma}{facetsLem}
  \label{lemma:facets}
  Consider a non-degenerate $d$-simplex $\sigma \subseteq \Rspace^d$ such
  that $\Conv \sigma \subseteq \Offset \M \rtube$ for some $\rtube <
  \Reach{\M}$. If $\sigma$ has no vertical facets relative to \M,
  then $\Above \sigma$ and $\Below \sigma$ are non-empty sets that
  partition the facets of $\sigma$.
\end{restatable}

\section{Vertically collapsing simplicial complexes}
\label{section:vertical-collapses}

In this section, assuming that $\US{K} \subseteq \Offset \M \rtube$
for some $\rtube < \Reach{\M}$, we introduce an algorithm for
simplifying $K$ using vertical collapses relative to $\M$
(Section~\ref{section:algorithm}) and establish conditions for when it
outputs a triangulation of $\M$
(Section~\ref{section:correctness}). We first present a naive version
that requires the knowledge of $\M$ and then present a practical
version (Section~\ref{section:in-practice}).

\subsection{Naive algorithm}
\label{section:algorithm}

\begin{definition}[Vertically free simplices]
\label{definition:vertically-free-simplex}
  A simplex $\tau$ is said to be {\em free from
    above (resp., free from below)} in $K$ relative to \M if
  \begin{itemize}
  \item $\tau$ is a free simplex of $K$;
  \item the unique inclusion-maximal simplex $\sigma$ in $\Star \tau K$ has dimension $d$;
  \item the set of $(d-1)$-simplices in $\Star \tau K$ is exactly the set of upper (resp.,
    lower) facets of $\sigma$ relative to \M.
  \end{itemize}
  We say that $\tau$ is {\em vertically free} in $K$ relative to \M if
  $\tau$ is either free from above or from below in $K$ relative to
  \M. See Figures~\ref{figure:collapses} and \ref{figure:sources-and-sinks}
for a depiction.
\end{definition}

\begin{remark}
  \label{remark:vertically-free-relative-to-hyperplanes}
  Definition \ref{definition:vertically-free-simplex} can be naturally extended to non-compact submanifolds \M. In particular, it holds for hyperplanes, a fact that we use in
  Algorithm~\ref{algo:practical}.
\end{remark}

\ignore{
  \bianca{Note that the definition can be naturally extended to non-compact surfaces.
  }
  \biancacmt{I don't think there is need to specifically talk abouit the hyperplanes once we said it holds for non-compact \M. But if you feel it would be more adequate, we could add something like:
  	
  	In particular, it holds for hyperplanes, a fact that we use in Algorithm~\ref{algo:practical}.}
}

\begin{definition}[Vertical collapse]
A {\em vertical collapse} of $K$ 
relative to \M is the operation of removing the star of a simplex $\tau \in K$ that is vertically free
relative to \M.
\end{definition}

\begin{figure}[htb]
  \def\svgwidth{1.0\linewidth}
  \centering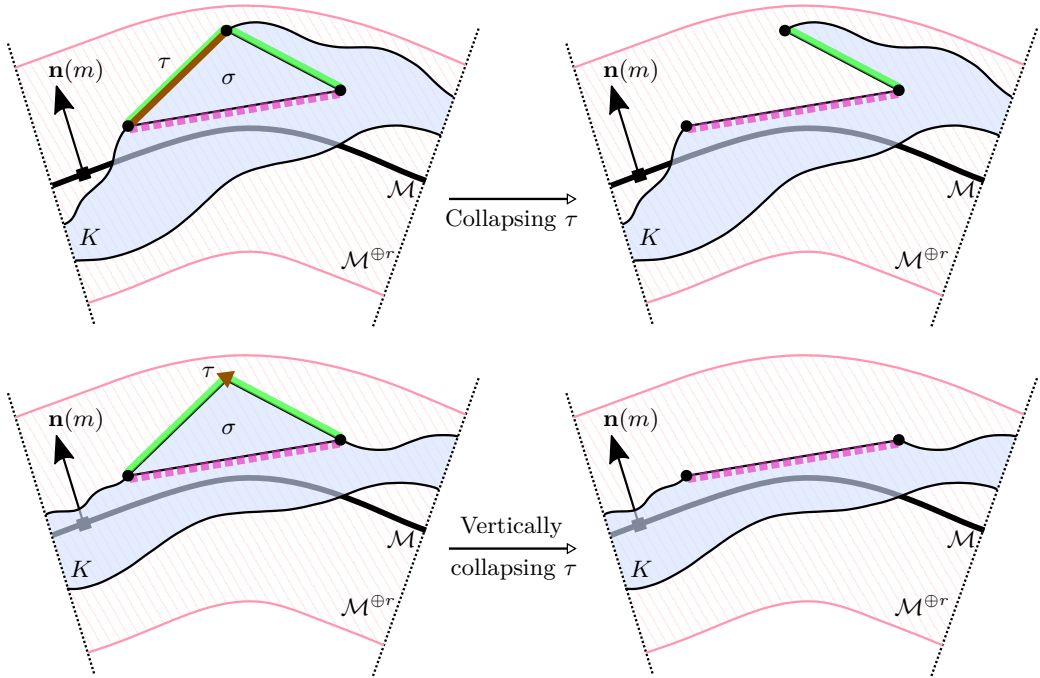
  \caption{ Schematic drawings of $K$ in blue (smooth filled
    areas). Top row: the edge $\tau$ is free but not vertically free
    relative to \M and collapsing $\tau$ does not preserve the
    vertical convexity of $K$. Bottom row: the vertex $\tau$ is free
    from above relative to \M, so that collapsing $\tau$ preserves the
    vertical convexity of $K$ (Lemma~\ref{lemma:invariant}). The
    $(d-1)$-simplices of $K$ that disappear with $\tau$ are precisely
    the upper facets of $\sigma$ (smooth edges, in green).  }
  \label{figure:collapses}
\end{figure}

A vertical collapse of $K$ can be seen as compressing the underlying
space of $K$ by shifting its upper or lower skin along directions
normal to $\M$; see Figure~\ref{figure:collapses}. Our
first algorithm, outlined in Algorithm~\ref{algo:generic}, simplifies
$K$ by iteratively applying vertical collapses relative to $\M$. It is
worth noting that the algorithm operates on any simplicial complex $K$
with $\US{K} \subseteq \Offset{\M}{\rtube}$ for $\rtube < \Reach{\M}$,
irrespective of whether $K$ is vertically convex relative to $\M$ or
not.

\begin{algorithm}
  \caption{\hypertarget{naivealgo}{$\NaiveVertSimp(K)$}}
  \begin{algorithmic}
    \WHILE{ there is a simplex $\tau$ vertically free in $K$ relative to \M }
    \STATE{ Collapse $\tau$ in $K$; }
    \ENDWHILE
  \end{algorithmic}
  \label{algo:generic}
\end{algorithm}

\subsection{Correctness}
\label{section:correctness}

We now establish conditions under which $\NaiveVertSimp(K)$ transforms $K$
into a triangulation of \M.  For that,
we introduce a binary relation over $d$-simplices:

\begin{definition}[Below relation $\below$]
  Let $\sigma_0, \sigma_1 \subseteq \Rspace^d$ be two $d$-simplices
  sharing a common facet $\nu = \sigma_0 \cap \sigma_1$
  and  let $\Conv{\sigma_0} \cup \Conv{\sigma_1} \subseteq
  \Offset \M \rtube$ for some $\rtube < \Reach{\M}$. We say that
  $\sigma_0$ is {\em below} $\sigma_1$ (or that $\sigma_1$ is {\em
    above} $\sigma_0$) relative to \M, denoted $\sigma_0 \below
  \sigma_1$, if $\nu$ is an upper facet of $\sigma_0$ and a lower
  facet of $\sigma_1$ relative to \M.
\end{definition}

Note that the relation $\below$ is not acyclic in general, see Figure~\ref{figure:cycle-2d}.

\begin{figure}[htb]
  \def\svgwidth{0.5\linewidth}
  \centering%% Creator: Inkscape 1.3.2 (1:1.3.2+202311252150+091e20ef0f), www.inkscape.org
%% Accompanies image file 'socg-cyclic-complex-modified.pdf' (pdf, eps, ps)
\begingroup%
  \makeatletter%
  \providecommand\color[2][]{%
    \errmessage{(Inkscape) Color is used for the text in Inkscape, but the package 'color.sty' is not loaded}%
    \renewcommand\color[2][]{}%
  }%
  \providecommand\transparent[1]{%
    \errmessage{(Inkscape) Transparency is used (non-zero) for the text in Inkscape, but the package 'transparent.sty' is not loaded}%
    \renewcommand\transparent[1]{}%
  }%
  \providecommand\rotatebox[2]{#2}%
  \newcommand*\fsize{\dimexpr\f@size pt\relax}%
  \newcommand*\lineheight[1]{\fontsize{\fsize}{#1\fsize}\selectfont}%
  \ifx\svgwidth\undefined%
    \setlength{\unitlength}{239.38801575bp}%
    \ifx\svgscale\undefined%
      \relax%
    \else%
      \setlength{\unitlength}{\unitlength * \real{\svgscale}}%
    \fi%
  \else%
    \setlength{\unitlength}{\svgwidth}%
  \fi%
  \global\let\svgwidth\undefined%
  \global\let\svgscale\undefined%
  \makeatother%
  \begin{picture}(1,0.87090571)%
    \lineheight{1}%
    \setlength\tabcolsep{0pt}%
    \put(0,0){\includegraphics[width=\unitlength,page=1]{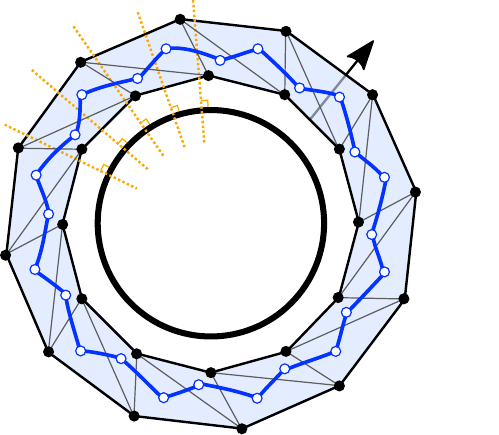}}%
    \put(0.27066703,0.27748119){\makebox(0,0)[lt]{\lineheight{1.25}\smash{\begin{tabular}[t]{l}\M\end{tabular}}}}%
    \put(0.7623089,0.81124543){\makebox(0,0)[lt]{\lineheight{1.25}\smash{\begin{tabular}[t]{l}$\normal m$\end{tabular}}}}%
    \put(0,0){\includegraphics[width=\unitlength,page=2]{socg-cyclic-complex-modified.pdf}}%
    \put(0.53114746,0.54425234){\makebox(0,0)[lt]{\lineheight{1.25}\smash{\begin{tabular}[t]{l}$m$\end{tabular}}}}%
    \put(0.83647339,0.3748277){\makebox(0,0)[lt]{\lineheight{1.25}\smash{\begin{tabular}[t]{l}$K$\end{tabular}}}}%
    \put(0,0){\includegraphics[width=\unitlength,page=3]{socg-cyclic-complex-modified.pdf}}%
  \end{picture}%
\endgroup%

  \caption{ Non-Delaunay triangles that form a cycle in the $\below$
    relation and their dual graph. \label{figure:cycle-2d} }
\end{figure}

\begin{theorem}[Correctness]
  \label{theorem:correctness-generic-simplification}
  Consider $K$ such that $\US{K} \subseteq \Offset \M \rtube$ for some
  $\rtube < \Reach{\M}$ and assume the following:
  \begin{description}
  \item[Injective projection:] $K$ has a non-vertical skeleton relative to \M.
  \item[Covering projection:] $\pi_\M(\US{K}) = \M$.
  \item[Vertical convexity:] $K$ is vertically convex relative to \M.
  \item[Acyclicity:] $\below$ is acyclic over $d$-simplices of $K$.
  \end{description}
  Then, $\NaiveVertSimp(K)$ transforms $K$ into a
  triangulation of $\M$.
\end{theorem}

The remaining of this section aims to prove
Theorem~\ref{theorem:correctness-generic-simplification} and we
consider $K$ such that $\US{K} \subseteq \Offset \M r$ for some $r < \Reach{\M}$.
Using the
relation $\below$, associate to $K$ its dual graph $G_\M(K)$ that has
one node for each $d$-simplex of $K$ and one arc for each pair of
$d$-simplices $\sigma_0, \sigma_1 \in K$ that share a common
facet $\sigma_0 \cap \sigma_1$. Direct an arc from $\sigma_0$
to $\sigma_1$ if $\sigma_0 \below \sigma_1$, and from $\sigma_1$ to
$\sigma_0$ otherwise. Since either $\sigma_0 \below \sigma_1$ or
$\sigma_1 \below \sigma_0$, this yields a well-defined orientation for
each arc in the dual graph. Figures~\ref{figure:cycle-2d} and
\ref{figure:sources-and-sinks} show examples.

\begin{figure}[htb]
  \def\svgwidth{.8\linewidth}
  \centering%% Creator: Inkscape 1.3.2 (091e20e, 2023-11-25), www.inkscape.org
%% Accompanies image file 'socg-sources-and-sinks-modified.pdf' (pdf, eps, ps)
\begingroup%
  \makeatletter%
  \providecommand\color[2][]{%
    \errmessage{(Inkscape) Color is used for the text in Inkscape, but the package 'color.sty' is not loaded}%
    \renewcommand\color[2][]{}%
  }%
  \providecommand\transparent[1]{%
    \errmessage{(Inkscape) Transparency is used (non-zero) for the text in Inkscape, but the package 'transparent.sty' is not loaded}%
    \renewcommand\transparent[1]{}%
  }%
  \providecommand\rotatebox[2]{#2}%
  \newcommand*\fsize{\dimexpr\f@size pt\relax}%
  \newcommand*\lineheight[1]{\fontsize{\fsize}{#1\fsize}\selectfont}%
  \ifx\svgwidth\undefined%
    \setlength{\unitlength}{425.59446716bp}%
    \ifx\svgscale\undefined%
      \relax%
    \else%
      \setlength{\unitlength}{\unitlength * \real{\svgscale}}%
    \fi%
  \else%
    \setlength{\unitlength}{\svgwidth}%
  \fi%
  \global\let\svgwidth\undefined%
  \global\let\svgscale\undefined%
  \makeatother%
  \begin{picture}(1,0.60558033)%
    \lineheight{1}%
    \setlength\tabcolsep{0pt}%
    \put(0,0){\includegraphics[width=\unitlength,page=1]{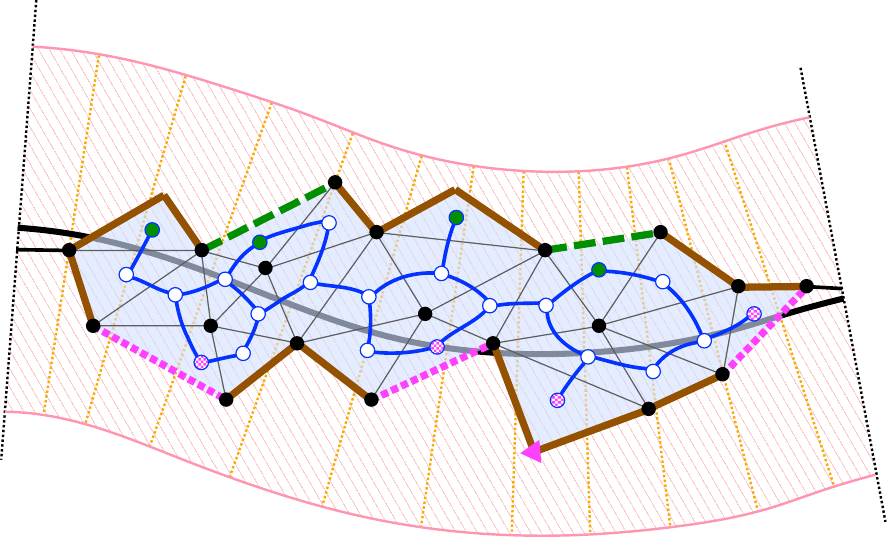}}%
    \put(0.03324047,0.36120361){\makebox(0,0)[lt]{\lineheight{1.25}\smash{\begin{tabular}[t]{l}\M\end{tabular}}}}%
    \put(0.05268301,0.56750332){\makebox(0,0)[lt]{\lineheight{1.25}\smash{\begin{tabular}[t]{l}$\Offset \M r$\end{tabular}}}}%
    \put(0.7873321,0.3276917){\makebox(0,0)[lt]{\lineheight{1.25}\smash{\begin{tabular}[t]{l}$K$\end{tabular}}}}%
    \put(0.91745154,0.22809074){\color[rgb]{0,0,0}\makebox(0,0)[lt]{\lineheight{1.25}\smash{\begin{tabular}[t]{l}$m$\end{tabular}}}}%
    \put(0,0){\includegraphics[width=\unitlength,page=2]{socg-sources-and-sinks-modified.pdf}}%
    \put(0.85251308,0.41854269){\makebox(0,0)[lt]{\lineheight{1.25}\smash{\begin{tabular}[t]{l}$\normal m$\end{tabular}}}}%
    \put(0,0){\includegraphics[width=\unitlength,page=3]{socg-sources-and-sinks-modified.pdf}}%
  \end{picture}%
\endgroup%

  \caption{A vertically convex simplicial complex $K$ relative to $\M$. All free simplices are highlighted by thickness. There are four simplices free from below (represented by three dotted edges and one triangular vertex, in pink) and four simplices free from above (represented by two dashed edges and two square vertices, in green). The dual graph (oriented edges, in blue) has four sources (vertices filled by dots, in pink) and four sinks (vertices with a smooth filling, in green), in
    	one-to-one correspondence with the free simplices from below and above.}
  \label{figure:sources-and-sinks}
\end{figure}

In a directed graph, a \emph{source} is a node with only outgoing
arcs, while a \emph{sink} is a node with only incoming arcs. The next
lemma states that a vertical collapse in $K$ corresponds to the
removal of either a sink or a source in $G_\M(K)$ and conversely. For that, given a
finite set of abstract simplices $\Sigma = \{\sigma_1, \sigma_2,
\ldots, \sigma_k\}$, let $\bigcap \Sigma = \bigcap_{i=1}^k \sigma_i$
denote the set of vertices that belong to all simplices in
$\Sigma$. If $\bigcap \Sigma \neq \emptyset$, it forms an abstract
simplex.

\begin{restatable}[Sinks and sources]{lemmma}{correspondFreeTerminalNode}
  \label{lemma:correspondance-free-terminal-node}
  Consider $K$ such that $\US{K} \subseteq \Offset \M
    \rtube$ for some $\rtube < \Reach{\M}$ and assume that $K$
    satisfies the injective projection, covering projection and
    vertical convexity assumptions of Theorem~\ref{theorem:correctness-generic-simplification}.
  Consider a $d$-simplex $\sigma \in K$ and let $\tau = \bigcap
    {\Above \sigma}$ and $\tau' = \bigcap {\Below \sigma}$. Then,
  \begin{itemize}
  \item $\tau$ is a free simplex of $K$ from above relative to \M
    $\iff$ $\sigma$ is a sink of $G_\M(K)$.
  \item $\tau'$ is a free simplex of $K$ from below relative to \M
    $\iff$ $\sigma$ is a source of $G_\M(K)$.
  \end{itemize}
\end{restatable}

\medskip

The next lemma provides an invariant for the while-loop of
Algorithm~\ref{algo:generic}.

\begin{restatable}[Loop invariant]{lemma}{AlgorithmInvariant}
  \label{lemma:invariant}
  Consider $K$ 
  such that $\US{K} \subseteq \Offset \M \rtube$ for some
  $\rtube < \Reach{\M}$. Let $\tau$ be a vertically free simplex in
  $K$ relative to \M. Let $K'$ be obtained from $K$ by collapsing
  $\tau$ in $K$. If $K$ satisfies the assumptions of
  Theorem~\ref{theorem:correctness-generic-simplification}, then so
  does $K'$.  
\end{restatable}

  \begin{lemma}[Upon termination]
    \label{lemma:termination}
    Consider $K$ such that $\US{K} \subseteq \Offset \M \rtube$ for
    some $\rtube < \Reach{\M}$ and assume that $K$ satisfies the
    injective projection and vertical convexity assumptions of
    Theorem~\ref{theorem:correctness-generic-simplification}.  If
    $G_\M(K) = \emptyset$, then
    $\LowerSkin{\US{K}}=\UpperSkin{\US{K}}$.
\end{lemma}

\begin{proof}
  We establish the contrapositive:
  \begin{equation*}
    \LowerSkin{\US{K}} \neq \UpperSkin{\US{K}} \quad \implies \quad G_\M(K) \neq \emptyset
  \end{equation*}
  Suppose that the two skins are distinct, in other words, that there
  exists $m \in \M$ such that $\low {\US{K}} m \neq \up {\US{K}}
  m$ and let us show that the segment $[\low {\US K} m,\up K m]$
  intersects the support of at least one $d$-simplex of $K$, implying
  $G_\M(K)\neq\emptyset$.  Suppose, for a contradiction, that $[\low
    {\US{K}} m,\up {\US K} m]$ only intersects the support of
  $i$-simplices of $K$ for $i<d$.  As $\low {\US K} m \neq \up {\US K}
  m$ and $K$ has a finite number of simplices, at least one of these
  $i$-simplices, say $\nu$, intersects $[\low {\US K} m,\up {\US K}
    m]$ in a non-zero length segment containing distinct points $x,y
  \in \Conv \nu \cap [\low {\US K} m,\up {\US K} m]$.  Hence, $x$ and
  $y$ share the same orthogonal projection $m$ onto \M, implying that
  $\nu$ is vertical relative to \M.  This contradicts the injective
  projection assumption on $K$ and therefore establishes the contrapositive.
\end{proof}

We now prove the correctness of $\NaiveVertSimp(K)$.

\begin{proof}[Proof of Theorem \ref{theorem:correctness-generic-simplification}.]
  The algorithm starts with $K$ that satisfies the theorem
  assumptions.  By Lemma~\ref{lemma:invariant}, after each iteration
  of the while-loop we obtain a new $K$ that continues to satisfy
  those assumptions.  Since each iteration involves a vertical
  collapse of $K$ relative to $\M$, the number of $d$-simplices of $K$
  is reduced. Thus, the algorithm must terminate. Upon termination,
  there are no vertically free simplices in $K$ relative to \M. By
  Lemma \ref{lemma:correspondance-free-terminal-node}, this implies
  that, when Algorithm~\ref{algo:generic} terminates, $G_\M(K)$ has no
  terminal node (neither a source nor a sink) and is therefore
  empty. By Lemma \ref{lemma:termination}, it follows that
    $\LowerSkin{\US{K}}=\UpperSkin{\US{K}}$.  By Lemma
  \ref{lemma:upper-equals-lower-consequence}, we have $K = \partial K$
  and Lemma~\ref{lemma:upper-and-lower-complex-equal-boundary} implies
  \[
  K = \partial K = \UpperC K = \LowerC K,
  \]
  with $K$ being a triangulation of \M.
\end{proof}

\subsection{Practical version}
\label{section:in-practice}

Algorithm~\ref{algo:generic} relies on knowledge of $\M$, which
renders it impractical for implementation since $\M$ is typically
unknown.  In this section, we introduce
Algorithm~\ref{algo:practical}, a feasible variant that is correct if
the $(d-1)$-simplices of $K$ form a sufficiently small angle with \M;
see Appendix \ref{appendix:angles} for a definition of the
angle between affine spaces. In this variant, we assign an affine
space $\HH_\tau$ to each $\tau \in K$: for a free simplex $\tau$ with
a $d$-dimensional coface $\sigma$, $\HH_\tau$ is defined as the
hyperplane spanned by any facet of $\sigma$. Otherwise, set $\HH_\tau
= \emptyset$. We also use the notion of vertically free simplices relative to
$\HH_\tau$, extending Definition~\ref{definition:vertically-free-simplex} as indicated in Remark~\ref{remark:vertically-free-relative-to-hyperplanes}.

\begin{algorithm}
  \caption{\hypertarget{practicalalgo}{$\PracticalVertSimp(K)$}}
  \begin{algorithmic}
    \WHILE{ there is a simplex $\tau$ vertically free in $K$ relative to $\HH_\tau$ }
    \STATE{ Collapse $\tau$ in $K$; }
    \ENDWHILE
  \end{algorithmic}
  \label{algo:practical}
\end{algorithm}

\begin{restatable}[Correctness]{theorem}{InPractice}
  \label{theorem:in-practice}
  Suppose that $K$ satisfies the assumptions of
  Theorem~\ref{theorem:correctness-generic-simplification} and, in addition, for all
  $(d-1)$-simplices $\nu$ of $K$
  \begin{equation}
    \label{eq:angle-in-practice}
    \max_{a \in \nu} \angle(\Aff \nu, \Tangent {\pi_\M(a)} \M) \, < \, \frac{\pi}{4}.    
  \end{equation}
  Then, $\PracticalVertSimp(K)$ transforms $K$ into a triangulation
  of $\M$.
\end{restatable}

\section{Correct reconstructions from \texorpdfstring{$\alpha$}{α}-complexes}
\label{section:alpha-complexes}

In this section, we assume that $\M$ is sampled by a finite point set $P$ and consider Algorithms~\ref{algo:naive-squash} and~\ref{algo:practical-squash}, which apply vertical collapses to $\Del{P,\alpha}$ either straightforwardly or practically. We introduce two parameters, $\varepsilon \geq 0$ and $\delta \geq 0$, to control the sample density and noise of $P$, respectively, and a scale parameter $\alpha \geq 0$. Section~\ref{section:stating-reconstruction-results} establishes conditions ensuring the correctness of Algorithms~\ref{algo:naive-squash} and~\ref{algo:practical-squash}, with the proof outlined in Section~\ref{section:proof-technique}. Section~\ref{section:surfaces-in-3D} shows how these conditions hold for a wide range of $\frac{\varepsilon}{\reach}$ and $\frac{\alpha}{\reach}$ when $\M$ is a surface in $\Rspace^3$ and $P$ is noiseless. This result is extended to the restricted Delaunay complex of $P$ in Section~\ref{section:restricted-Delaunay-complex}.
	
\begin{algorithm}
  \caption{\hypertarget{naivesquash}{$\NaiveSquash(P,\alpha)$}}
  \begin{algorithmic}
	\STATE{ $K \leftarrow \Del{P,\alpha}$;  $\NaiveVertSimp(K)$; {\bf return} $K$; }
  \end{algorithmic}
  \label{algo:naive-squash}
\end{algorithm}

\begin{algorithm}
  \caption{\hypertarget{practicalsquash}{$\PracticalSquash(P,\alpha)$}}
  \begin{algorithmic}
	\STATE{ $K \leftarrow \Del{P,\alpha}$;  $\PracticalVertSimp(K)$; {\bf return} $K$; }
  \end{algorithmic}
  \label{algo:practical-squash}
\end{algorithm}

\subsection{Sampling and angular conditions in $\Rspace^d$}
\label{section:stating-reconstruction-results}

	The next definition enables us to express our results in $\Rspace^d$ more concisely.

	\begin{definition}[Strict homotopy condition]
		We say that $\varepsilon, \delta \geq 0$ satisfy the {\em strict
			homotopy condition} if $(\reach-\delta)^2 - \varepsilon ^2 >
		(4\sqrt{2}-5)\reach^2$ for $\delta \leq \varepsilon$ and
		$\varepsilon + \sqrt{2}\delta < (\sqrt{2}-1)\reach$ for
		$\delta \geq \varepsilon$. 
	\end{definition}

	Let $I(\varepsilon,\delta)$ be an interval of $\alpha$ values so
    that $\Offset{P}{\alpha}$ is vertically convex with relation to
    $\M$. The exact definition can be found in
    Appendix~\ref{appendix:range-alpha}. The
    fact that this interval guarantees vertical convexity follows from
    the specialization of Propositions $5$ and $7$ in
    \cite{socg24-NSW} to the case where $\M$ has codimension-one:

\begin{theorem}[Specialization of \cite{socg24-NSW}]
	\label{theorem:NSW}
		Suppose that $\M \subseteq \Offset P \varepsilon$ and $P
        \subseteq \Offset \M \delta$ with $\varepsilon, \delta \geq 0$
        that satisfy the strict homotopy condition. Then, for all
        $\alpha \in I(\varepsilon,\delta)$, $\pi_\M(\Offset P \alpha)
        = \M$ and $\Offset P \alpha$ is vertically convex relative to
        $\M$. Thus, $\Offset P \alpha$ has associated upper and lower
        skins and deformation-retracts onto $\M$ along $\pi_\M$. In
        addition, the two skins partition $\partial \Offset
        P \alpha$.
\end{theorem}

	The above concepts can be put together to state our main theorem: 

    \ignore{
    \domi{I changed the range for $\alpha$ in the theorem so that the ratios are in $[-1,1]$.}
    \sidebianca{Doesn't it make the theorem look more complicated?}
    \domi{Unfortunately, necessary to be sure that \eqref{eq:angle-condition-for-vertical-convexity} and \eqref{eq:angle-condition-acyclicity} are well-defined (unless we are able to show that redundant with $\alpha \in I(\varepsilon,\delta)$). }
    }

\begin{restatable}{theorem}{TheoremAlphaComplex}
	\label{theorem:alpha-complex}
	Assume $\M \subseteq \Offset P \varepsilon$ and $P \subseteq \Offset \M \delta$ for $\varepsilon, \delta \geq 0$ that satisfy the strict homotopy condition. Let $\alpha \in
    \left[\delta, \frac{2(\reach - \delta)}{3}\right) \cap I(\varepsilon, \delta)$ and $\beta > 0$ such that $\Offset \M \beta \subseteq \Offset P \alpha$. Suppose that for all $i$-simplices $\tau \in \Del{P, \alpha}$, $0 < i < d$ and all $(d-1)$-simplices $\nu \in \Del{P, \alpha}$ it holds:

	\begin{align}
		\max_{x \in \Conv\tau} \angle(\Aff \tau, \Tangent {\pi_\M(x)} \M) \, &< \, \frac{\pi}{2},
		\label{eq:angle-condition-for-skeleton}
		\\
		\min_{a \in \tau} \angle(\Aff \tau, \Tangent {\pi_\M(a)} \M) \, &< \, \arcsin\left(\frac{(\reach+\beta)^2 - (\reach+\delta)^2 - \alpha^2}{2(\reach+\delta)\alpha}\right) %
		\,\text{and}
		\label{eq:angle-condition-for-vertical-convexity}
		\\
		\min_{x \in \Conv\nu} \angle(\Aff \nu, \Tangent {\pi_\M(x)} \M)
		\, &< \, \frac{\pi}{2} - 2 \arcsin\left( \frac{\alpha}{2(\reach - \delta - \alpha)} \right).
		\label{eq:angle-condition-acyclicity}
	\end{align}
	Then, $\Del{P,\alpha}$ satisfies the injective projection, covering
	projection, vertical convexity and acyclicity assumptions of
	Theorem~\ref{theorem:correctness-generic-simplification}.  Furthermore, both the upper and lower complexes of $\Del{P, \alpha}$ relative to $\M$ are triangulations of $\M$ and $\NaiveSquash(P, \alpha)$ returns a triangulation of $\M$.

\end{restatable}

One can check that Conditions~\eqref{eq:angle-condition-for-skeleton},
\eqref{eq:angle-condition-for-vertical-convexity} and
\eqref{eq:angle-condition-acyclicity}
are well-defined; see Remark~\ref{remark:conditions-well-defined}.

\begin{corollary}
	\label{theorem:correcness-squash-practical}
	Suppose the assumptions of Theorem \ref{theorem:alpha-complex} are
	satisfied and furthermore that for all $(d-1)$-simplices $\nu \in
	\Del{P,\alpha}$,
	\begin{align}
		\label{eq:angle-condition-practical}
		\max_{a \in \nu} \angle(\Aff \nu, \Tangent {\pi_\M(a)} \M) \, < \, \frac{\pi}{4}.
	\end{align}
	Then, $\PracticalSquash(P,\alpha)$ returns a triangulation of $\M$.
\end{corollary}
	
\subsection{Partial proof technique for Theorem \ref{theorem:alpha-complex}}
\label{section:proof-technique}
In this section, we establish the covering projection and vertical convexity of $\Del{P,\alpha}$, as guaranteed by Theorem~\ref{theorem:alpha-complex}. The complete proof of that theorem can be found in Appendix~\ref{appendix:alpha-complex-theorem}.

\begin{lemma}
\label{lemma:alpha-complex-vertical-convexity}
Assume $\M \subseteq \Offset P \varepsilon$ and $P \subseteq \Offset
\M \delta$ for $\varepsilon,\delta \geq 0$ that satisfy the strict homotopy
condition. Let $\alpha \in [\delta,\reach-\delta) \cap
  I{(\varepsilon,\delta)}$ and $\beta > 0$ be such that $\Offset \M
  \beta \subseteq \Offset P \alpha$. Suppose that for all
  $i$-simplices $\tau \in \Del{P,\alpha}$ with $0 < i < d$, Conditions
  \eqref{eq:angle-condition-for-skeleton} and
  \eqref{eq:angle-condition-for-vertical-convexity} hold.  Then,
  $\pi_\M(\US{\Del{P,\alpha}}) = \M$ and $\Del{P,\alpha}$ is
  vertically convex relative to \M.
\end{lemma}

\begin{figure}
\centering
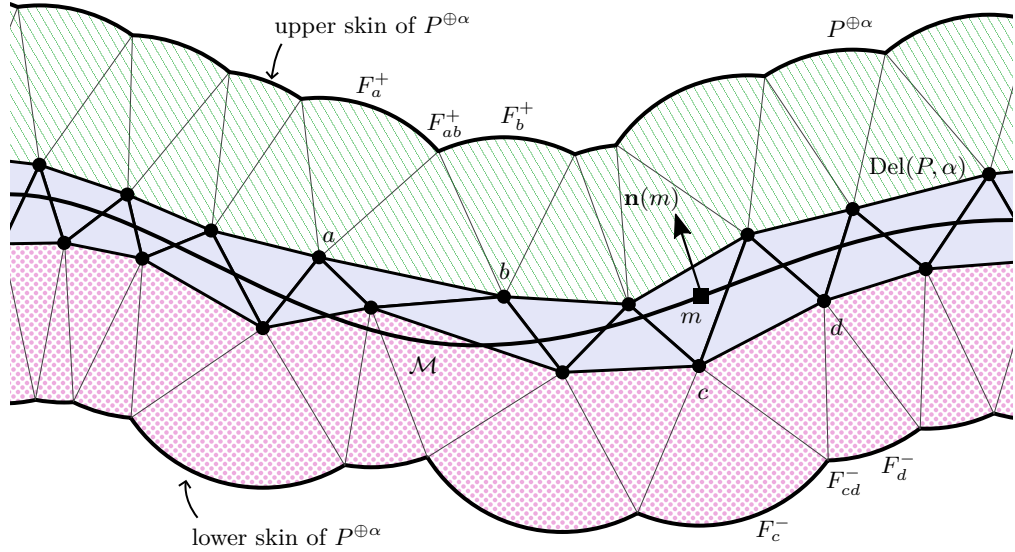
\caption{ Decomposing $\Offset P \alpha \setminus
  \US{\Del{P,\alpha}}^\circ$ into upper joins (hashed, in green) and lower joins (dotted, in pink).
  \label{figure:upper-and-lower-joins}
}
\end{figure}

\ignore{
\biancacmt{We could use subparagraphs instead of paragraphs here. It would be actually shorter than adding a phrase. And I personally find the structured presentation helpful.}
\subparagraph{Upper and lower joins.} Consider.....
\biancacmt{Of course there would be no need for a subparagraph "Proof sketch" given that the "Proof of Lemma 20" already fills that purpose. I can also see how having a single subparagraph environment can feel unnecessary, so I'm fine with what you choose.}
}

\subparagraph{Upper and lower joins.}
For proving this lemma, we introduce upper and lower joins.  Consider
$\M$, $P$, $\varepsilon$, $\delta$, $\alpha$ and $\beta$ that satisfy
the assumptions of Lemma~\ref{lemma:alpha-complex-vertical-convexity}
and notice that they also meet the conditions of
Theorem~\ref{theorem:NSW}.  Therefore, $\Offset P \alpha$ has
associated upper and lower skins and the two skins form a partition of
$\partial \Offset P \alpha$. Using that partition, we decompose the
set difference $\Offset P \alpha \setminus\US{\Del{P,\alpha}}^\circ$
into upper and lower joins, slightly adapting what is done
in~\cite{edelsbrunner2011alpha}; see Figures
\ref{figure:alpha-complex-simple} and
\ref{figure:upper-and-lower-joins}.  For that, notice that $\partial
\Offset P \alpha$ can be decomposed into faces, each face being the
restriction of $\partial \Offset P \alpha$ to a Voronoi cell of $P$.
There is a one-to-one correspondence between simplices of $\partial
\Del{P,\alpha}$ and faces of $\partial \Offset P \alpha$: the simplex
$\tau \in \partial \Del{P,\alpha}$ corresponds to the face $F_\tau =
V(\tau,P) \cap \partial \Offset P \alpha$ and conversely.  We can
further partition each face $F_\tau$ of $\partial \Offset P \alpha$
into a portion $F_\tau^+$ that lies on the upper skin of $\partial
\Offset P \alpha$ and a portion $F_\tau^-$ that lies on the lower skin
of $\partial \Offset P \alpha$. Note that $F_\tau^+$ or $F_\tau^-$ can
be empty. We refer to $F_\tau^+$ as an upper face and $F_\tau^-$ as a
lower face. The set of upper faces decompose the upper skin, while the
set of lower faces decompose the lower skin.  A \emph{join} $X * Y$ is
defined as the set of segments $[x,y]$ where $x \in X$ and $y \in
Y$~\cite{edelsbrunner2011alpha}. We call $F_\tau^+ * \Conv{\tau}$ an
\emph{upper join} and $F_\tau^- * \Conv{\tau}$ a \emph{lower
join}. The next lemma, proved in Appendix \ref{appendix:alpha-complex-vertical-convexity}, identifies
points in $\partial|\Del{P,\alpha}|$ that are connected to upper or
lower joins.

\begin{remark}
\label{remark:covering—joins}
The collection of upper and lower joins cover the set $\Offset P
\alpha \setminus \US{\Del{P,\alpha}}^\circ$.
\end{remark}

\begin{remark}
\label{remark:intersection-joins}
If an upper join and a lower join have a non-empty intersection, the
common intersection belongs to $\US{\Del{P,\alpha}}$.
\end{remark}

\begin{restatable}{lemma}{ConnectingBoundariesSimplified}
  \label{lemma:connecting-boundaries-simplified}
  Under the assumptions of Lemma
  \ref{lemma:alpha-complex-vertical-convexity}, let $\spx \in
  \Del{P,\alpha}$ and $x \in \Interior{\Conv{\gamma}}$. If for some
  $\lambda > 0$ ({\em resp.}  $\lambda<0$), the segment $(x,x+\lambda
  \normal {\pi_\M(x)}]$ lies outside $\US{\Del{P,\alpha}}$, then it
intersects an upper ({\em resp. lower} join).
\end{restatable}

\begin{proof}[Proof of Lemma \ref{lemma:alpha-complex-vertical-convexity}]
Consider $\M$, $P$, $\varepsilon$, $\delta$,
$\alpha$ and $\beta$ that satisfy the assumptions of
Lemma~\ref{lemma:alpha-complex-vertical-convexity}. As noted before, they also meet the
conditions of Theorem~\ref{theorem:NSW}. Therefore,
$\pi_\M(\Offset P\alpha) = \M$ and $\Offset P \alpha$ is
vertically convex relative to $\M$. Hence, there
exists $r< \Reach{\M}$ such that $\Offset P \alpha \subseteq \Offset
\M r$ and $\Offset P \alpha \cap \Normal m \M \cap B(m,\rtube)$ is a
line segment for any $m \in \M$. Fix $m \in \M$ arbitrarily. We show
that $\US{\Del{P,\alpha}} \cap \Normal m \M \cap B(m,\rtube)$ is also
a line segment.

First, we show  by contradiction that it is non-empty.
Let $u^+$ ({\em resp. $u^-$}) be the endpoint of the segment $\Offset P \alpha \cap \Normal m \M \cap B(m,\rtube)$ that lies on the upper (\emph{resp.} lower) skin of $\Offset P \alpha$ and hence is contained in an upper (\emph{resp.} lower) join.
By Remark~\ref{remark:covering—joins}, the entire segment $[u^+,u^-]$ is
covered by upper and lower joins.  Thus, at some point $c$ of $[u^+,u^-]$, an upper join and a lower join intersect. By Remark~\ref{remark:intersection-joins}, such an intersection places $c$ on $\US{\Del{P,\alpha}}$ as well.

\begin{figure}[htb]
	\centering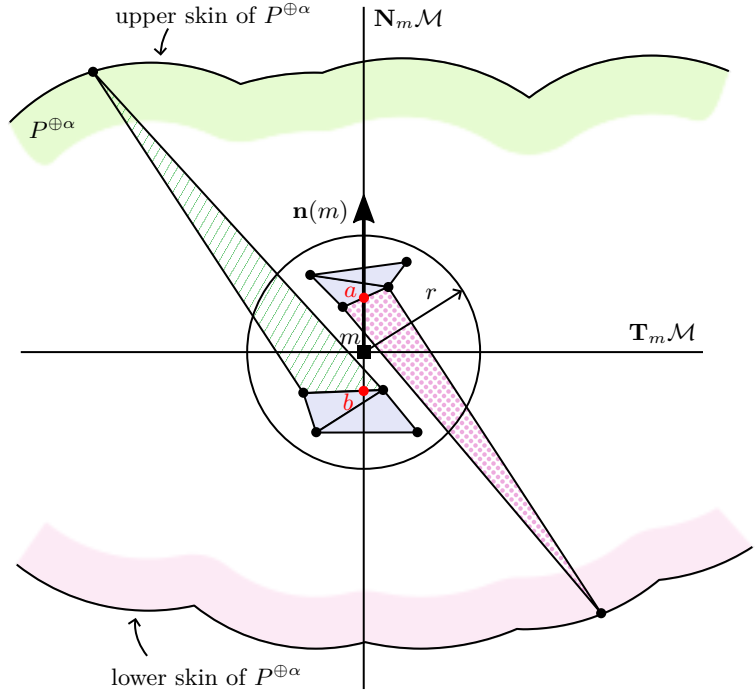 
	\caption{
		Reaching a contradiction in the proof of Lemma~\ref{lemma:alpha-complex-vertical-convexity}. We see the simplices of $\Del{P,\alpha}$ whose support intersects $\Normal m \M \cap B(m,r)$ (smooth filling, in pale blue) and one upper (hashed, in green) and one lower (dotted, in pink) joins that intersect $[a,b]$.
	\label{figure:proof-alpha-complex-vertical-convexity}
    }
\end{figure}

Second, we show by contradiction that $\US{\Del{P,\alpha}}\cap \Normal
m \M \cap B(m,r)$ is connected. Suppose that
$a,b\in\US{\Del{P,\alpha}}\cap\Normal m \M \cap B(m,r)$ are such that
$ [a,b] \cap \US{\Del{P,\alpha}} = \{a,b\}$ with $a$ being above
  $b$ along the direction of $\normal m$.  Since $a,b \in
\US{\Del{P,\alpha}} \subseteq \Offset P \alpha$ and $\Offset P \alpha$
is vertically convex, the segment $[a,b]$ is contained in $\Offset
P\alpha$. By Remark~\ref{remark:covering—joins}, the entire
  segment $[a,b]$ is covered by upper and lower joins. Let $\spx_a$
and $\spx_b$ be the simplices of $\Del{P,\alpha}$ that contain $a$ and
$b$, respectively, in their relative interior. Letting $\lambda =
  \frac{\|a-b\|}{2}$, then both segments $\left(a,a-\lambda \normal m\right]$ and
$\left(b,b+\lambda \normal m\right]$ lie outside $\US{\Del{P,\alpha}}$.
It follows, by Lemma~\ref{lemma:connecting-boundaries-simplified}, that
there are at least one lower and one upper joins among the joins that
cover $(a,b)$; see Figure
\ref{figure:proof-alpha-complex-vertical-convexity}. Hence, an upper
and a lower joins intersect at a point $c$ of the segment $(a,b)$. By
Remark~\ref{remark:intersection-joins}, $c$ lies in
$\US{\Del{P,\alpha}}$, a contradiction. Therefore, for all $m \in \M$,
$\US{\Del{P,\alpha}}\cap \Normal m \M \cap B(m,r)$ is non-empty and
  connected,
thus forming a line segment.
\end{proof}

\subsection{Sampling conditions for surfaces in $\Rspace^3$.}
\label{section:surfaces-in-3D}

Theorem~\ref{theorem:alpha-complex} requires that the $i$-simplices of $\Del{P,\alpha}$ form a small angle with the manifold \M, for $0 < i < d$. Ensuring this can be challenging in practice, especially for $i \geq 3$. 
However, in the specific case of noiseless edges ($i=1$) or triangles ($i=2$), it is possible to upper bound the angle these simplices form with \M.
For edges, it is known that:

\begin{lemma}[{\cite[Lemma 7.8]{boissonnat2018geometric}}]\label{lemma:angleBetweenEdgeAndTangentPlane}
	If $ab$ is a non-degenerate edge with $a,b \in \M$, then
	\[
	\sin \angle \Aff{ab}, \Tangent a \M \, \leq \,  \frac{\|b-a\|}{2\,\reach}.
	\]
\end{lemma}\label{lemma:edges_bound_r3}

For triangles, let $\rcirc{\tau}$
be the radius of the smallest $(d-1)$-sphere circumscribing
$\tau$. We establish a simple bound that is tighter than the previous one (see~\cite[Lemma 3.5]{Dey:2006:reconstruction_book}):

\begin{lemma}
	\label{lemma:angleBetweenTriangleAndTangentPlane}
 	If $abc$ is a non-degenerate triangle with
	longest edge $bc$, for $a,b,c \in \M$, then
	\begin{align*}
		\sin \angle \Aff{abc}, \Tangent a \M \, &\leq \, \frac{\rcirc{abc}}{\reach} && \text{if $abc$ is an obtuse triangle and} \\
		\sin \angle \Aff{abc}, \Tangent a \M \, &\leq \,  \frac{\sqrt{3}\, \rcirc{abc}}{\reach} && \text{if $abc$ is an acute triangle.}
	\end{align*}
	If $abc$ is obtuse, the bound is tight and
	happens when $\M$ is a
	sphere of radius $\reach$.
\end{lemma}

The proof is technical and is therefore provided in Appendix \ref{appendix:angles}. For the same reason, the proof of the following result, where we use the bounds for edges and triangles to establish sampling conditions for surfaces in $\Rspace^3$, is in Appendix \ref{appendix:alpha-complex-second-corrolary}.
Let us define
\begin{equation*}
  \label{eq:good-beta}
  \beta_{\varepsilon,\alpha} = -\frac{\varepsilon^2}{2\reach} + \sqrt{\alpha^2 + \frac{\varepsilon^4}{4\reach^2}-\varepsilon^2},
\end{equation*}
which is one particular value of $\beta$ that guarantees $\Offset \M \beta
\subseteq \Offset P \alpha$; see Appendix~\ref{appendix:size-r-beta}.

\begin{restatable}{theorem}{CorollaryThreeDimensionTheory}
  \label{corollary:3D-theory}
  Let $\M$ be a $C^2$ surface in $\Rspace^3$ whose reach is at
    least $\reach>0$. Let $P$ be a finite point set such that $P
    \subseteq \M \subseteq \Offset P \varepsilon$.
  \begin{enumerate}
  \item \label{item:naive-3D-sampling}
    For all $\varepsilon, \alpha \geq 0$ that satisfy
    $\frac{\sqrt{3} \, \alpha}{\reach} < \min \left\{
    \frac{(\reach+\beta_{\varepsilon,\alpha})^2 - \reach^2 - \alpha^2}{2\reach\alpha},\,
    \cos\left( 2 \arcsin\left( \frac{\alpha}{\reach} \right)\right)
    \right\}$,
	\begin{itemize}
	\item the upper and lower complexes of $\Del{P,\alpha}$ relative to \M are
	  triangulations of \M;
	\item $\NaiveSquash(P,\alpha)$ returns a triangulation of $\M$.
	\end{itemize}
  \item \label{item:practical-3D-sampling}
    For all  $\varepsilon, \alpha \geq 0$ that satisfy in addition
    $\frac{\sqrt{3} \, \alpha}{\reach} < \sin\left( \frac{\pi}{4} - 2 \arcsin \left( \frac{\alpha}{\reach} \right)\right)$,
	\begin{itemize}
	\item $\PracticalSquash(P,\alpha)$ returns a triangulation of $\M$.
	\end{itemize}
  \end{enumerate}
\end{restatable}

  \medskip \noindent
  The pairs of
  $(\frac{\varepsilon}{\reach},\frac{\alpha}{\reach})$ that satisfy
  \itemref{item:naive-3D-sampling} and
  \itemref{item:practical-3D-sampling} are depicted in Figures~\ref{fig:sampling-condition-3D-left} and \ref{fig:sampling-condition-3D-right}, respectively.

\ignore{
\bianca{I like this approach, but %
  I would add an hyperlink for $\beta_{\epsilon,\alpha}$, or even move its definition to just before this theorem, as we don't use it before.}
\DomiEmph{The two options would work, although, at the moment, I have a preference for the second option.}
}

\begin{remark}
	In particular, \itemref{item:naive-3D-sampling} holds for
    $\frac{\varepsilon}{\reach} \leq 0.225$ and $\frac{\alpha}{\reach}
    = 0.359$; and \itemref{item:practical-3D-sampling} for $\frac{\varepsilon}{\reach} \leq 0.178$
    and $\frac{\alpha}{\reach} = 0.207$, better bounds
    than the previous existing ones \cite[Theorem
      13.16]{cheng2013delaunay}\cite{dey2017curve}.
\end{remark}

\begin{figure}
	\centering
	\begin{subfigure}[b]{1.6em}
		\subcaption{}\label{fig:sampling-condition-3D-left}
	\end{subfigure}%
	\raisebox{-0.1cm}{
		\def\svgwidth{.35\textwidth}
		%% Creator: Inkscape 1.3.2 (1:1.3.2+202311252150+091e20ef0f), www.inkscape.org
%% Accompanies image file 'corollary-3D-theory.pdf' (pdf, eps, ps)
\begingroup%
  \makeatletter%
  \providecommand\color[2][]{%
    \errmessage{(Inkscape) Color is used for the text in Inkscape, but the package 'color.sty' is not loaded}%
    \renewcommand\color[2][]{}%
  }%
  \providecommand\transparent[1]{%
    \errmessage{(Inkscape) Transparency is used (non-zero) for the text in Inkscape, but the package 'transparent.sty' is not loaded}%
    \renewcommand\transparent[1]{}%
  }%
  \providecommand\rotatebox[2]{#2}%
  \newcommand*\fsize{\dimexpr\f@size pt\relax}%
  \newcommand*\lineheight[1]{\fontsize{\fsize}{#1\fsize}\selectfont}%
  \ifx\svgwidth\undefined%
    \setlength{\unitlength}{227.92469788bp}%
    \ifx\svgscale\undefined%
      \relax%
    \else%
      \setlength{\unitlength}{\unitlength * \real{\svgscale}}%
    \fi%
  \else%
    \setlength{\unitlength}{\svgwidth}%
  \fi%
  \global\let\svgwidth\undefined%
  \global\let\svgscale\undefined%
  \makeatother%
  \begin{picture}(1,1.3542089)%
    \lineheight{1}%
    \setlength\tabcolsep{0pt}%
    \put(0,0){\includegraphics[width=\unitlength,page=1]{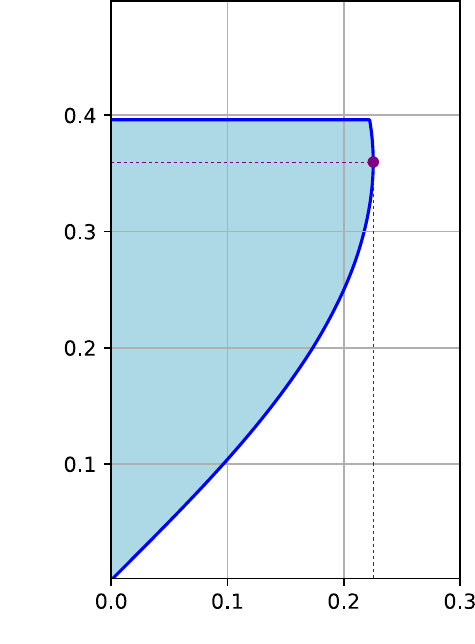}}%
    \put(0.5452672,0.00638831){\color[rgb]{0,0,0}\makebox(0,0)[lt]{\lineheight{1.25}\smash{\begin{tabular}[t]{l}$\varepsilon/\reach$\end{tabular}}}}%
    \put(-0.00254933,0.72754445){\color[rgb]{0,0,0}\makebox(0,0)[lt]{\lineheight{1.25}\smash{\begin{tabular}[t]{l}$\alpha/\reach$\end{tabular}}}}%
    \put(0,0){\includegraphics[width=\unitlength,page=2]{corollary-3D-theory.pdf}}%
    \put(0.61496754,0.16925673){\color[rgb]{0.50196078,0,0.50196078}\makebox(0,0)[lt]{\lineheight{1.25}\smash{\begin{tabular}[t]{l}$0.225$\end{tabular}}}}%
    \put(0.26069551,0.94885556){\color[rgb]{0.50196078,0,0.50196078}\makebox(0,0)[lt]{\lineheight{1.25}\smash{\begin{tabular}[t]{l}$0.359$\end{tabular}}}}%
    \put(0.26069551,1.13826739){\color[rgb]{0.50196078,0,0.50196078}\makebox(0,0)[lt]{\lineheight{1.25}\smash{\begin{tabular}[t]{l}$0.396$\end{tabular}}}}%
  \end{picture}%
\endgroup%

	}
	\hspace{0.8cm}
	\begin{subfigure}[b]{1.6em}
		\subcaption{}\label{fig:sampling-condition-3D-right}
	\end{subfigure}
	\raisebox{-0.1cm}{
		\def\svgwidth{.35\textwidth}
		%% Creator: Inkscape 1.3.2 (1:1.3.2+202311252150+091e20ef0f), www.inkscape.org
%% Accompanies image file 'corollary-3D-practical.pdf' (pdf, eps, ps)
\begingroup%
  \makeatletter%
  \providecommand\color[2][]{%
    \errmessage{(Inkscape) Color is used for the text in Inkscape, but the package 'color.sty' is not loaded}%
    \renewcommand\color[2][]{}%
  }%
  \providecommand\transparent[1]{%
    \errmessage{(Inkscape) Transparency is used (non-zero) for the text in Inkscape, but the package 'transparent.sty' is not loaded}%
    \renewcommand\transparent[1]{}%
  }%
  \providecommand\rotatebox[2]{#2}%
  \newcommand*\fsize{\dimexpr\f@size pt\relax}%
  \newcommand*\lineheight[1]{\fontsize{\fsize}{#1\fsize}\selectfont}%
  \ifx\svgwidth\undefined%
    \setlength{\unitlength}{227.93559265bp}%
    \ifx\svgscale\undefined%
      \relax%
    \else%
      \setlength{\unitlength}{\unitlength * \real{\svgscale}}%
    \fi%
  \else%
    \setlength{\unitlength}{\svgwidth}%
  \fi%
  \global\let\svgwidth\undefined%
  \global\let\svgscale\undefined%
  \makeatother%
  \begin{picture}(1,1.35398391)%
    \lineheight{1}%
    \setlength\tabcolsep{0pt}%
    \put(0,0){\includegraphics[width=\unitlength,page=1]{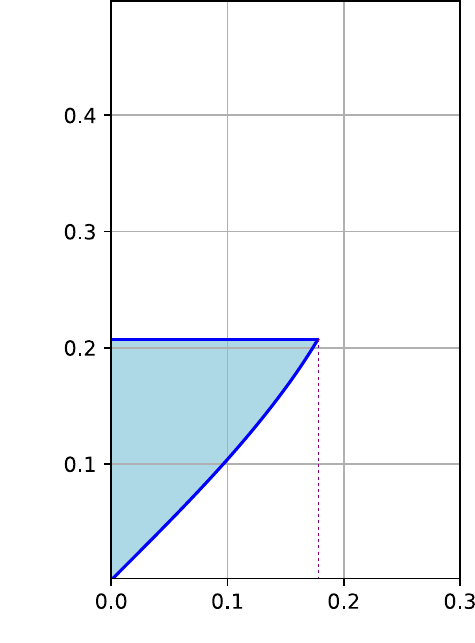}}%
    \put(0.54524107,0.00638801){\color[rgb]{0,0,0}\makebox(0,0)[lt]{\lineheight{1.25}\smash{\begin{tabular}[t]{l}$\varepsilon/\reach$\end{tabular}}}}%
    \put(-0.00254921,0.72750947){\color[rgb]{0,0,0}\makebox(0,0)[lt]{\lineheight{1.25}\smash{\begin{tabular}[t]{l}$\alpha/\reach$\end{tabular}}}}%
    \put(0,0){\includegraphics[width=\unitlength,page=2]{corollary-3D-practical.pdf}}%
    \put(0.26129193,0.67274098){\color[rgb]{0.50196078,0,0.50196078}\makebox(0,0)[lt]{\lineheight{1.25}\smash{\begin{tabular}[t]{l}$0.207$\end{tabular}}}}%
    \put(0.50100817,0.16069126){\color[rgb]{0.50196078,0,0.50196078}\makebox(0,0)[lt]{\lineheight{1.25}\smash{\begin{tabular}[t]{l}$0.178$\end{tabular}}}}%
    \put(0,0){\includegraphics[width=\unitlength,page=3]{corollary-3D-practical.pdf}}%
  \end{picture}%
\endgroup%

	}
	\caption{Pairs of $(\frac{\varepsilon}{\reach},\frac{\alpha}{\reach})$ for which
		$\NaiveSquash(P,\alpha)$ (a) and $\PracticalSquash(P,\alpha)$ (b) are
	  correct, for $P \subseteq \M \subseteq \Offset P \varepsilon$ and $d=3$.
      }
	\label{fig:sampling-condition-3D}
\end{figure}

\subsection{The restricted Delaunay complex}
\label{section:restricted-Delaunay-complex}

We recall from \cite{edelsbrunner1997triangulating} that the {\em restricted Delaunay complex} is
\[
\DelR{P} = \{ \sigma \subseteq P \mid  \sigma \neq \emptyset \text{ and }  V(\sigma,P) \cap \M \neq \emptyset \}.
\]

\begin{restatable}{theorem}{TheoremRestrictedDelaunayComplex}
  \label{theorem:3D-restricted-Delaunay-complex}
  Let $\M$ be a $C^2$ surface in $\Rspace^3$ whose reach is at least
  $\reach>0$. Let $P$ be a finite set such that $P \subseteq \M
  \subseteq \Offset P \varepsilon$ for $0 \leq
  \frac{\varepsilon}{\reach} \leq 0.225$.  Under the additional
    generic assumption that all Voronoi cells of $P$ intersect \M
    transversally, $\DelR{P}$ is a triangulation of \M.
\end{restatable}

\clearpage
\bibstyle{plainurl}
%%% \bibliography{references}

\clearpage
\appendix

\renewcommand{\thetheorem}{\Alph{section}\arabic{theorem}} 

\numberwithin{lemma}{section}
\renewcommand{\thelemma}{\Alph{section}\arabic{lemma}}

\renewcommand{\theremark}{\Alph{section}\arabic{remark}}
\renewcommand{\theproperty}{\Alph{section}\arabic{property}}

\etocarticlestyle

\etocruledstyle[1]{\etocfontminusone\color{gray}%
	\fboxrule1pt\fboxsep1ex
	\framebox[\linewidth]
	{\normalcolor\hss APPENDIX CONTENTS\hss}
}

\etocdepthtag.toc{mtappendix}
\etocsettagdepth{mtchapter}{none}
\etocsettagdepth{mtappendix}{subsection}
\tableofcontents

\clearpage
\section{Angle between flats}
\label{appendix:angles}

\subsection{Definitions and basic properties}

Consider two affine spaces $X,Y \subseteq \Rspace^d$. Let $U$ and $V$
be the vector spaces associated to $X$ and $Y$, respectively.  We
recall that the angle between two affine spaces $X$ and $Y$ is the
same as the angle between the two vector spaces $U$ and $V$ and is
defined as
\[
\angle(X,Y) = \angle(U,V) = \max_{\substack{u \in U \\ \|u\| = 1}}
\min_{\substack{v \in V \\ \|v\| = 1}} \angle (u,v).
\]
The definition is not symmetric in $X$ and $Y$, unless the two affine spaces $X$
and $Y$ (or equivalently, the two vector spaces $U$ and $V$) have same dimension. Given a vector space $V$, denote by
$V^\perp$ its orthogonal complement. The following properties will be
useful:

\begin{property}
  Let $U$ and $V$ be two vector spaces. Then,
  \[
  \angle(U,V) = \angle(V^\perp,U^\perp).
  \]
\end{property}

\begin{property}
  \label{property:complementary-angles}
  Let $U$ and $V$ be two vector spaces, with $\dim U = 1$. Then,
  \[
  \angle(U,V) + \angle(U,V^\perp) = \frac{\pi}{2}.
  \]
\end{property}

The angle between two vector spaces can also be written in terms of all hyperplanes that contain the second vector space:

\begin{lemma}\label{lemma:angleSpacesAsMaxOnHyperplane}
  Let $U$, $V$ be two vector spaces and $u\in U$. Then,
  \begin{align}\label{eq:angleSpacesAsMaxOnHyperplane_1}
	\min_{\substack{v \in V \\ \|v\| = 1}} \angle (u,v)
	= \max_{\substack {H \supseteq V  \\  \operatorname{codim}(H) = 1}} \min_{\substack{h \in H \\ \|h\|=1}} \angle (u,h)
  \end{align}
  and
  \begin{align}\label{eq:angleSpacesAsMaxOnHyperplane_2}
	\angle (U, V)  = \max_{\substack {H \supseteq V \\  \operatorname{codim}(H) = 1}} \angle (U,H),
  \end{align}
  where $H$ designates a vector space and $\operatorname{codim}(H) = \dim(H^\perp)$.
\end{lemma}
\begin{proof}
  For any codimension one vector space $H\supseteq V$, it holds that
  \[
  \min_{\substack{v \in V \\ \|v\| = 1}} \angle (u,v) \geq \min_{\substack{h \in H \\ \|h\|=1}} \angle (u,h).
  \]
  Thus, we have one direction of \eqref{eq:angleSpacesAsMaxOnHyperplane_1}:
	\[
	\min_{\substack{v \in V \\ \|v\| = 1}} \angle (u,v) 
	\geq \max_{ {H \supseteq V} \atop { \operatorname{codim}(H) = 1}} \min_{\substack{h \in H \\ \|h\|=1}} \angle (u,h).
	\]
	Note that if $u \in V$, \eqref{eq:angleSpacesAsMaxOnHyperplane_1}
    is trivially satisfied as both sides are $0$.  Hence, for the
    other direction of \eqref{eq:angleSpacesAsMaxOnHyperplane_1}, we consider $u \notin V$.
    Denote the vector in $V$ closest to $u$ by $\pi_{V} (u) = \arg
    \min_{v\in V} (v-u)^2$. From the definition of the angle, it
    holds that
	\[
    \min_{\substack{v \in V \\ \|v\| = 1}} \angle (u, v)
    =  \angle (u, \pi_{V} (u))
    = \arcsin \frac{\|\pi_{V} (u) - u\|}{\|u\|}.
    \]
	Next, define
	\[
	w = u - \pi_{V} (u) \in V^{\perp}
	\]
    and let $W$ be the vector space spanned by $w$. Since $\omega \in
    V^{\perp}$, one has $V \subseteq W^\perp$. Moreover, $\pi_{W^\perp}
    (u) = \pi_{V} (u)$ because $u - \pi_{V} (u) = \omega \in
    W$. Hence,
    \begin{align*}
	  \min_{\substack{v \in V \\ \|v\| = 1}} \angle (u,v) &=  \angle (u, \pi_{V} (u)) 
      =  \angle (u, \pi_{W^\perp} (u)) 
	  =  \min_{\substack{h \in W^\perp \\ \|h\|=1}} \angle (u,h) \\
	  &\leq \max_{ {H \supseteq V} \atop { \operatorname{codim}(H) = 1}} \min_{\substack{h \in H \\ \|h\|=1}} \angle (u,h). 
	\end{align*}
	Finally,~\eqref{eq:angleSpacesAsMaxOnHyperplane_2} holds because 
	\begin{align*}
		\angle (U, V)  &= \max_{\substack{u \in U \\ \|u\| = 1}} \min_{\substack{v \in V \\ \|v\| = 1}}  \angle (u,v) \\
		&=  \max_{\substack{u \in U \\ \|u\| = 1}} \max_{ {H \supseteq V} \atop { \operatorname{codim}(H) = 1}}   \min_{\substack{h \in H \\ \|h\|=1}} \angle (u,h) \\
		&= \max_{ {H \supseteq V} \atop { \operatorname{codim}(H) = 1}} \max_{\substack{u \in U \\ \|u\| = 1}}  \min_{\substack{h \in H \\ \|h\|=1}} \angle (u,h) \\
		&= \max_{ {H \supseteq V} \atop { \operatorname{codim}(H) = 1}} \angle (U,H). \qedhere
	\end{align*}
\end{proof}

\subsection{Angle between a triangle and a tangent space to the manifold}

Recall that Lemmas~\ref{lemma:angleBetweenEdgeAndTangentPlane}
and~\ref{lemma:angleBetweenTriangleAndTangentPlane} bound the angle
between the affine space spanned by a $1$- or $2$-simplex and a
tangent space to the manifold $\M$ in terms of the reach of $\M$. In this
section, we generalize these results by replacing the dependency on
the reach with a dependency on the local feature size. Additionally,
we relax the assumptions that $\M$ has codimension one and is
$C^2$. For that, assume $\M$ is a compact submanifold of $\Rspace^d$
with positive reach, from which it follows that $\M$ is
$C^{1,1}$~\cite[Proposition 1.4]{lytchak2005almost}.
Recall from~\cite{amenta1999surface} that the \emph{local feature size} at a
point $x$ on $\M$ is $\lfs{x} = d(x,\MA{\M})$ and that $\Reach{\M}
\leq \lfs{x}$ for all $x \in \M$.

\begin{lemma}
  \label{lemma:angleBetweenEdgeAndTangentPlane-lfs}
  Let \M be a compact submanifold of $\Rspace^d$ with positive reach.
	If $ab$ is a non-degenerate edge with $a,b \in \M$, then
	\[
	\sin \angle \Aff{ab}, \Tangent a \M \, \leq \,  \frac{\|b-a\|}{2\,\lfs{a}}.
	\]
\end{lemma}

\begin{proof}
  Point $b$ lies on or outside the set of spheres of radius $\lfs{a}$ that are tangent to $\M$ at $a$. In the worst case, where $b$ lies on one of these spheres, 
  $ab$ is a chord of that sphere.
\end{proof}

Let $\rcirc{\tau}$ designate the radius of the smallest $(d-1)$-sphere circumscribing a simplex $\tau$. 

\begin{lemma}
  \label{lemma:angleBetweenTriangleAndTangentPlane-lfs}
  Let \M be a compact submanifold of $\Rspace^d$ with positive reach.
 	If $abc$ is a non-degenerate triangle with
	longest edge $bc$, for $a,b,c \in \M$, then
	\begin{align*}
		\sin \angle \Aff{abc}, \Tangent a \M \, &\leq \, \frac{\rcirc{abc}}{\lfs{a}} && \text{if $abc$ is an obtuse triangle and} \\
		\sin \angle \Aff{abc}, \Tangent a \M \, &\leq \,  \frac{\sqrt{3}\, \rcirc{abc}}{\lfs{a}} && \text{if $abc$ is an acute triangle.}
	\end{align*}
	If $abc$ is obtuse, the bound is tight and happens when $\M$ is a
    $k$-sphere of radius $\reach$ with $2 \leq k \leq d-1$.
\end{lemma}

\begin{proof}
  Let $A = \angle cab$, $\vec{u} = \frac{b-a}{\|b-a\|}$, and $\vec{v}
  = \frac{c-a}{\|c-a\|}$. Let $H$ be any hyperplane that contains
  $\Tangent a \M$ and let $\vec{n}$ be an unit vector orthogonal to
  $H$. By construction, $\vec{n}$ is orthogonal to \M at
  $a$. Introduce the two angles
  \begin{align*}
    \theta_u &= \arccos(\vec{u} \cdot \vec{n}), \\
    \theta_v &= \arccos(\vec{v} \cdot \vec{n}).
  \end{align*}
  Thanks to Lemma \ref{lemma:angleBetweenEdgeAndTangentPlane-lfs}, we have that
    \begin{align}
      \left| \cos \theta_u \right| &= \left| \vec{u} \cdot \vec{n} \right| = \sin \angle \Aff{ab}, H \leq \frac{ \| b-a \| }{2 \, \lfs{a}}, \label{eq.UDotN}\\
    \left| \cos \theta_v \right| &= \left| \vec{v} \cdot \vec{n} \right| = \sin \angle \Aff{ac}, H\leq \frac{ \| c-a \| }{2 \, \lfs{a}}.\label{eq.VDotN}
  \end{align}
 Consider the vector
  \[
  \vec{w} = \frac{\vec{v} - (\vec{u} \cdot \vec{v})\, \vec{u}}{\left\| \vec{v} - (\vec{u} \cdot \vec{v})\, \vec{u} \right\| }
  = \frac{1}{\sin A} \left( \vec{v} - \cos A \, \vec{u} \right).
  \]
  By construction, $\vec{w}$ is a unit vector parallel to the affine
  plane spanned by triangle $abc$ and orthogonal to $\vec{v}$. It
  follows that $(\vec{u}, \vec{w} )$ is an orthonormal basis of the
  vector plane supporting the triangle $abc$ and we can write
  \begin{align*}
    \sin \angle \Aff{abc}, H  &= \max_{t\in [0,2\pi]}  (\cos t\, \vec{u} + \sin t\, \vec{w}) \cdot \vec{n}  \\
    &= \max_{t\in [0,2\pi]}   (\vec{u} \cdot \vec{n}) \cos t + (\vec{w} \cdot \vec{n} ) \sin t .
  \end{align*}
  The maximum over $[0,2\pi]$ of the map $t \mapsto x \cos t + y \sin t$ is $\sqrt{x^2+ y^2}$. Thus,
  \begin{align*}
        \left(\sin \angle \Aff{abc},H \right)^2 &=  (\vec{u} \cdot \vec{n})^2 +  (\vec{w} \cdot \vec{n})^2 \\
    &= \cos^2 \theta_u + \frac{1}{\sin^2 A} \left(\cos \theta_v - \cos A \cos  \theta_u  \right)^2 \\
    &= \frac{1}{\sin^2 A} \left(\cos^2 \theta_u +  \cos^2 \theta_v - 2 \cos A \cos  \theta_u  \cos \theta_v \right).
  \end{align*}
  For fixed $A$, the last expression is a homogeneous degree $2$
  polynomial in $\cos \theta_u$ and $\cos \theta_v$. Under the
  constraints \eqref{eq.UDotN} and \eqref{eq.VDotN}, its maximum is
  reached at one of the four corners of the constraint rectangle, i.e.
  when $\cos \theta_u = \pm \frac{ \| b-a \| }{2 \lfs{a}}$ and $ \cos
  \theta_v = \pm \frac{ \| c-a \| }{2 \lfs{a}}$. We distinguish two
  cases, depending on whether the triangle $abc$ is obtuse or acute.

  \medskip

  If $abc$ is acute, then $A \leq \frac{\pi}{2}$ and $\cos A \geq
  0$. It follows that the maximum of the polynomial is reached either when $\cos
  \theta_u = \frac{ \| b-a \| }{2 \lfs{a}}$ and $ \cos \theta_v =
  -\frac{ \| c-a \| }{2 \lfs{a}}$ or when $\cos \theta_u = -\frac{ \|
    b-a \| }{2 \lfs{a}}$ and $ \cos \theta_v = \frac{ \| c-a \| }{2
    \lfs{a}}$. Letting $b' = 2a - b$ be the symmetric point of $b$ with
  respect to $a$, we obtain
  \begin{align*}
      \left(\sin \angle \Aff{abc}, H \right)^2 
    &\leq  \frac{1}{4 \lfs{a}^2} \frac{1}{\sin^2 A} \left( \| b-a \|^2 +   \| c-a \|^2 + 2 \cos A  \| b-a \|   \| c-a \| \right) \\
    &=  \frac{1}{4 \lfs{a}^2} \frac{1}{\sin^2 A}  \| b'-c\|^2 \\
    &=  \frac{\rcirc{ab'c}^2}{\lfs{a}^2}.
  \end{align*}
  Since the last equation holds for any hyperplane $H$ that
  contains $\Tangent a \M$, it follows from
  \eqref{eq:angleSpacesAsMaxOnHyperplane_2} in Lemma
  \ref{lemma:angleSpacesAsMaxOnHyperplane} that
  \begin{align*}
    \left(\sin \angle \Aff{abc}, \Tangent a \M  \right)^2 \leq  
    \frac{\rcirc{ab'c}^2}{\lfs{a}^2}.
  \end{align*}
  Under the constraint that the longest edge in the triangle $abc$ is
  $bc$, the ratio $\frac{\rcirc{ab'c}}{\rcirc{abc}}$ is maximized when
  $abc$ is an equilateral triangle, in which case the ratio is equal
  to $\sqrt{3}$. The bound for acute triangles follows.
  
   \medskip

  If $abc$ is obtuse, then $A \geq \frac{\pi}{2}$ since $bc$ is the longest edge. Therefore, $\cos A \leq
  0$ and the maximum of the polynomial is reached either when $\cos
  \theta_u = \frac{ \| b-a \| }{2 \lfs{a}}$ and $\cos \theta_v = \frac{
  	\| c-a \| }{2 \lfs{a}}$ or when $\cos \theta_u = -\frac{ \| b-a \|
  }{2 \lfs{a}}$ and $\cos \theta_v = -\frac{ \| c-a \| }{2 \lfs{a}}$.
  We obtain that
  \begin{align*}
  	\left(\sin \angle \Aff{abc}, H \right)^2 
  	&\leq  \frac{1}{4 \lfs{a}^2} \frac{1}{\sin^2 A} \left( \| b-a \|^2 +   \| c-a \|^2 - 2 \cos A  \| b-a \|   \| c-a \| \right) \\
  	&=  \frac{1}{4 \lfs{a}^2} \frac{1}{\sin^2 A}  \| b-c \|^2 \\
  	&=  \frac{\rcirc{abc}^2}{\lfs{a}^2},
  \end{align*}
  where we used the formula $\rcirc{abc} = \frac{\| b-c \|}{2 \sin A}$
  to obtain the last equality. As for the acute case, we get that:
 \begin{align*}
  	\left(\sin \angle \Aff{abc}, \Tangent a \M \right)^2 
        \leq   \frac{\rcirc{abc}^2}{\lfs{a}^2},
  \end{align*}
  
 To see that the bound is tight in the obtuse case, consider a
 $k$-sphere of radius $\reach$ and center $o$ with $2 \leq k
 \leq d-1$ as $\M$. Let $o'$ be the center of the smallest
 $2$-sphere circumscribing $abc$. Then, the
 angle $\theta$ between the lines $oa$ (or $ob$ or $oc$) and $oo'$ is
 \[
 \theta
 = \angle \Aff{abc}, \Tangent a \M
 = \angle \Aff{abc}, \Tangent b \M
 = \angle \Aff{abc}, \Tangent c \M
 \]
  and $\sin \theta$  is precisely the bound $\frac{\rho(abc)}{\reach}=\frac{\rho(abc)}{\lfs{a}}$.
\end{proof}

\subsection{Difference between angular deviations at vertices}

Consider a simplex $\tau$ with vertices on $\M$. In this section, we
bound the variation of the angle $\angle(\Aff \tau, \Tangent{a} \M)$
as $a$ varies over the vertices of $\tau$.  Recall that $\rcirc{\tau}$
denotes the radius of the smallest $(d-1)$-sphere circumscribing the
simplex $\tau$.

\begin{lemma}
  \label{lemma:difference-max-min-angular-deviation-noiseless}
  Let $0 < i < d$. For any $i$-simplex $\tau$ whose vertices lie
  on $\M$ with $\rcirc{\tau} < \reach$,
  \begin{equation}
      \max_{a \in\tau} \angle(\Aff \tau, \Tangent {a} \M)
      -
      \min_{a \in\tau} \angle(\Aff \tau, \Tangent {a} \M)
      \leq 2 \arcsin{\left( \frac{\rcirc{\tau}}{\reach} \right)}.
  \end{equation}
\end{lemma}

\begin{proof}
  By \cite[Corollary 3]{boissonnat2019reach}, for all $m,m' \in \M$, we have $\angle (\Tangent m \M, \Tangent {m'} \M ) \leq 2 \arcsin \left( \frac{\|m'-m\|}{2\reach} \right)$.   Consider two vertices $a,b \in\tau$.
Since $\|a-b\| \leq 2\rcirc{\tau}$, it holds that
  \begin{equation*}
    \label{eq:angular-deviation-c-x}
    \angle (\Tangent {b} \M, \Tangent {a} \M )
    \leq 2 \arcsin\left( \frac{\rcirc{\tau}}{\reach} \right).
  \end{equation*}
  The result follows by combining the above inequality and the triangular inequality:
  \[
  \angle ( \Aff \tau, \Tangent {b} \M ) - \angle ( \Aff \tau, \Tangent {a} \M )
  \, \leq \,  \angle (\Tangent {b} \M, \Tangent {a} \M ) \,\leq\, 2 \arcsin\left( \frac{\rcirc{\tau}}{\reach} \right). \qedhere 
  \]
\end{proof}

\clearpage
\section{Omitted proofs of Section \ref{section:vertically-convex-simplicial-complexes}}
\label{appendix:vertically-convex-simplicial-complexes}

In this section, we provide the proofs omitted in Section~\ref{section:vertically-convex-simplicial-complexes}, restating the statements for clarity.

\subsection{Upper and lower skins}

\SkinsHomeomorphicToManifold*

\begin{proof}
  It suffices to show that $\UpperSkin {\US{K}}$ is a closed set and that the
  mapping $\pi_\M : \UpperSkin {\US{K}} \to \N$ is a homeomorphism. This
  result will automatically extend to $\LowerSkin {\US{K}}$ because the lower
  skin of $\US{K}$ becomes the upper skin of ${\US{K}}$ when the orientation of
  ${\US{K}}$ is reversed.
  	  	
  We begin by proving that $\UpperSkin {\US{K}}$ is a closed set. To do this,
  we define a set $A$, which is closed by construction, and show that
  $A = \UpperSkin {\US{K}}$. Let the set of points lying above ${\US{K}}$ in the
  offset region $\Offset \M \rtube$ be
  \[
  \AboveSet {\US{K}} = \bigcup_{m\in \N} [\up {\US{K}} m, m + \rtube \normal{m}]
  \]
  and set
  \[
  A = \Closure{\AboveSet {\US{K}}} \cap |K|.
  \]
  $A$ is the intersection of two closed sets and hence it is closed. Recalling that
  \[
  \US{K}  = \bigcup_{m \in \N} [\low {\US{K}} m, \up {\US{K}} m],
  \]
  we observe that the sets $\AboveSet {\US{K}}$ and $\US{K}$ intersect at a
  single point, $\up {\US{K}} m$, along each normal segment $[m - \rtube
    \normal{m}, m + \rtube \normal{m}]$. Thus, by construction,
  $\UpperSkin {\US{K}} = \AboveSet {\US{K}} \cap \US{K}$ and hence $\UpperSkin {\US{K}}
  \subseteq A$.
  	
  To prove the reverse inclusion, we first show that each normal
  segment $\Normal m \M \cap B(m, \rtube)$ intersects $A$ at exactly
  one point. We proceed in three stages.

  \begin{figure}[htb]
  	\centering\includegraphics{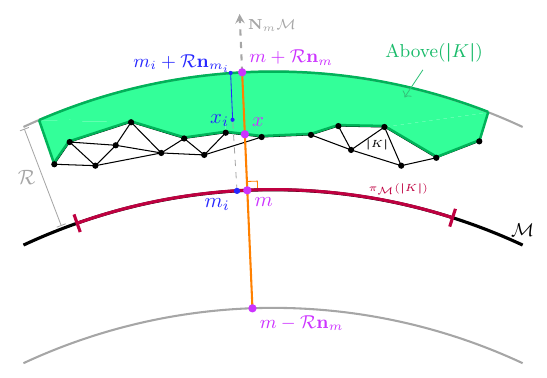}
  	\caption{Schematics of the construction in Stage one.}
  	\label{fig:appBStage}
  \end{figure}

  \smallskip\noindent{\bf\underline{Stage one.}} We establish the following implication:
  \[
  x \in \Closure{\AboveSet {\US{K}}} \implies [x,m + \rtube \normal{m}] \subseteq \Closure{\AboveSet {\US{K}}}.
  \]
  Consider a point $x \in \Closure{\AboveSet {\US{K}}}$, implying the
  existence of a sequence of points $(x_i)_{i\in \Nintegers}$ in
  $\AboveSet {\US{K}}$ that converges to $x$ as $i \to +\infty$. Let $m =
  \pi_{\M}(x)$ and $m_i = \pi_{\M}(x_i)$, and consider the sequence of
  segments $([x_i, m_i+ \rtube \normal{m_i}])_{i \in \Nintegers}$ (see
  Figure~\ref{fig:appBStage} for a depiction). Each segment $[x_i,
    m_i+ \rtube \normal{m_i}]$ in the sequence belongs to $\AboveSet
  {\US{K}}$ since, given $x_i \in \AboveSet {\US{K}}$, it holds that  
  \[
    [x_i, m_i+ \rtube \normal{m_i} ] \subseteq [\up{K}{m_i} , m_i+ \rtube \normal{m_i} ] \subseteq \AboveSet {\US{K}}.
  \]
  Furthermore, we show that the sequence of segments $([x_i, m_i+
    \rtube \normal{m_i}])_{i\in \Nintegers}$ converges in Hausdorff
  distance to $[x, m+\rtube \normal{m}]$, completing the proof. It is
  enough to show the convergence of the endpoints of segments and
  since $\lim_{i \to +\infty} x_i = x$, we only need to show that
  $\lim_{i \to +\infty} (m_i + \rtube \normal{m_i}) = m + \rtube
  \normal{m}$. Let $f: \AboveSet {\US{K}} \to \AboveSet {\US{K}}$ be defined by $f(z) =
  \pi_\M(z) + \rtube \normal{\pi_\M(z)}$. Both $\pi_\M: \AboveSet{\US{K}}  \to \M$ and $\mathbf{n} : \M \to \normal m$ are continuous maps,
  hence $f$ is continuous and $\lim_{i \to +\infty} f(x_i) =
  f(x)$. Thus, $\lim_{i \to +\infty} (m_i + \rtube \normal{m_i}) = m +
  \rtube \normal{m}$.
       
  \smallskip\noindent{\bf\underline{Stage two.}} We show that for all
  $m \in \N$, the set $A \cap \Normal m \M \cap B(m,\rtube)$ is
  connected. Consider two points $x, y \in A \cap \Normal m \M \cap
  B(m,\rtube)$. Since $x, y \in \US{K} \cap \Normal m \M \cap
  B(m,\rtube)$, the vertical convexity assumption guarantees that $[x,y]
  \in \US{K} \cap \Normal m \M \cap B(m,\rtube)$. Since $x, y \in
  \Closure{\AboveSet {\US{K}}} \cap \Normal m \M \cap B(m,\rtube)$, the previous
  stage guarantees that $[x,y] \in \Closure{\AboveSet {\US{K}}} \cap \Normal m
  \M \cap B(m,\rtube)$. Therefore, $[x,y]$ belongs to $A \cap \Normal m
  \M \cap B(m,\rtube)$.

  \smallskip\noindent{\bf\underline{Stage three.}}  We show that $A
  \cap \Normal m \M \cap B(m,\rtube)$ is reduced to a single point for
  all $m \in \N$. In the previous stage, we established that $A \cap
  \Normal m \M \cap B(m,\rtube)$ is connected, thus forming a
  segment. Assume, for a contradiction, that the endpoints $x$ and $y$
  of this segment are distinct. For each $z \in [x,y]$, let $\sigma_z$
  denote the unique simplex of ${\US{K}}$ containing $z$ in its relative
  interior. Since $A \subseteq \partial \US{K} = \US{\partial K}$,
  $\sigma_z$ belongs to the boundary complex $\partial {\US{K}}$ of ${\US{K}}$ and
  has $\dim \sigma_z < d$ for all $z \in [x, y]$. Given that $x \neq
  y$, there exists $z\in [x,y]$ with $\dim \sigma_z > 0$. Since
  $\sigma_z$ is a boundary simplex of ${\US{K}}$ containing $[x,y]$ locally
  around $z$, it must be vertical with $0 < \dim \sigma_z < d$,
  contradicting our assumption that ${\US{K}}$ contains no vertical
  $i$-simplices for $0 < i < d$.

  \smallskip

	From stage three, $A$ intersects $[m - \rtube\normal{m}, m
		+ \rtube \normal{m}]$ at a single point for all $m \in \N$.
		By construction, so does $\UpperSkin {\US{K}}$. Since $\UpperSkin {\US{K}}
		\subseteq A$, it follows that $\UpperSkin {\US{K}} = A $ along the
		normal segment $[m - \rtube \normal{m}, m + \rtube \normal{m}]$
		for all $m \in \N$. Hence, the upper skin of ${\US{K}}$ coincides with $A$ and is therefore closed.
		
  	We are now ready to conclude. The upper skin of ${\US{K}}$ is a closed
    set and, therefore, compact. The map $\pi_{\M}: \UpperSkin {\US{K}} \to
    \N$ is a continuous bijection from one compact set to another, and
    is thus a homeomorphism. Let us prove Equation \eqref{eq:boundary-equals-upper-and-lower}, that is,
    \[
    \partial \US{K} = \UpperSkin {\US{K}} \cup \LowerSkin {\US{K}}.
    \]
    By definition, the boundary of $\US{K}$ is the set $\partial \US{K} = \US{K}
    \cap \Closure{ \Rspace^d \setminus \US{K}}$.

    Introducing the following sets (see Figure~\ref{fig:appB_schematics})
  \begin{align*}
    \BelowSet{{\US{K}}} &= \bigcup_{m \in \N} [m - \rtube \normal m, \low {\US{K}} m] \,\,\,\,\text{and}\\
    \Side{{\US{K}}} &= \big\{ x \in  \Offset \M \rtube \mid \pi_\M (x) \notin \N \big\},
  \end{align*}
  we can write
  \[
  \Rspace^d \setminus \US{K} = \AboveSet {\US{K}} \cup \BelowSet {\US{K}} \cup (\Rspace^d \setminus \Offset \M \rtube) \cup \Side{\US{K}} .
  \]
  
  \begin{figure}[htb]
  	\centering
  	\includegraphics{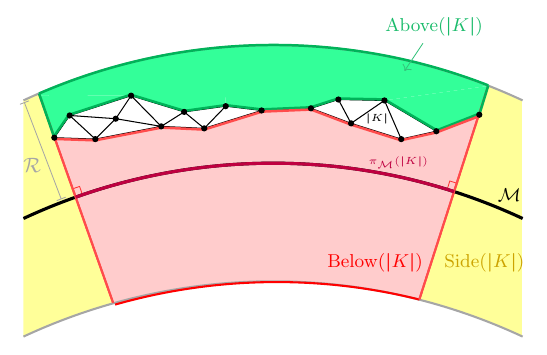}
  	\caption{The different regions in $\Offset \M \rtube$.}
 	\label{fig:appB_schematics}
  \end{figure}
  
  Taking the closure on both sides, it follows that
  \[\arraycolsep=1.4pt\def\arraystretch{2.2}
  \begin{array}{rcl}
  \Closure{\Rspace^d \setminus \US{K}} &= & \Closure{\AboveSet {\US{K}}} \cup \Closure{\BelowSet {\US{K}}} \cup \\[-3ex]
  & & \cup \Closure{\Rspace^d \setminus \Offset \M \rtube}
  \cup \Closure{\Side{\US{K}}}.
  \end{array}\]
  Intersecting both sides of the equality with $\US{K}$ and noting
  that the third term in the union has an empty
  intersection with $\US{K}$, we obtain that
  \begin{align*}
    \partial \US{K} &= \US{K} \cap \Closure{ \Rspace^d \setminus \US{K}}\\
    &= \US{K} \cap \left(  \Closure{\AboveSet {\US{K}}} \cup \Closure{\BelowSet {\US{K}}} \cup \Closure{\Side {\US{K}}} \right)\\
    &= \left( \US{K} \cap  \Closure{\AboveSet {\US{K}}} \right) \cup \left( \US{K} \cap \Closure{\BelowSet {\US{K}}} \right)\\
    &= \UpperSkin {\US{K}} \cup \LowerSkin {\US{K}},
  \end{align*}
  which concludes the proof.
\end{proof}

\UpperEqualsLowerConsequence*

\begin{proof}
  Since $\UpperSkin {\US{K}} = \LowerSkin {\US{K}}$, it follows that
  $\low {\US{K}} m = \up {\US{K}} m$ for all $m \in \pi_\M(\US{K})$
  and $\US{K} \subseteq \UpperSkin {\US{K}} = \LowerSkin {\US{K}}$. By
  Lemma~\ref{lemma:upper-and-lower-skins-homeomorphic-to-manifold},
  \[
  \US{K} \subseteq \UpperSkin {\US{K}} \cup \LowerSkin {\US{K}} = \partial\US{K} = \US{\partial K}.
  \]
  Thus, $\US{K} \subseteq \US{\partial K}$ and therefore $K = \partial
  K$.
\end{proof}

By Lemma~\ref{lemma:upper-and-lower-skins-homeomorphic-to-manifold}, we already know
that the union of the upper and lower skins is a simplicial
complex. However, since their intersection is non-empty, it is unclear
whether each one of them separately is a simplicial complex as well.
Lemma \ref{lemma:upper-and-lower-complex-equal-boundary}
establishes exactly that, replacing each set in
Equation~(\ref{eq:boundary-equals-upper-and-lower}) by a simplicial
counterpart. Its proof relies on the next lemma.

\begin{lemma}
  \label{lemma:exclusive-upper-and-lower}
  Suppose that $K$ has a non-vertical skeleton and is vertically convex relative to
  \M. Then, for all simplices $\nu \in {{K}}$, the following two implications hold:
  \begin{enumerate}
  \item $\Interior{\Conv\nu} \cap \UpperSkin {\US{K}} \neq \emptyset$ $\implies$ $\Conv\nu \subseteq \UpperSkin {\US{K}}$;
  \item $\Interior{\Conv\nu} \cap \LowerSkin {\US{K}} \neq \emptyset$ $\implies$ $\Conv\nu \subseteq \LowerSkin {\US{K}}$.
  \end{enumerate}
\end{lemma}

\begin{proof}
Let $\nu$ be a simplex of ${K}$ and $x \in
\Interior{\Conv\nu}$. Consider $t>0$ such that the only simplices of
${{K}}$ whose support intersects $B(x,t)$ are in the star of $\nu$.
Following ~\cite{federer1959curvature}, we define the \emph{tangent cone} of $x$ in ${\US{K}}$ as
\begin{equation}\label{eq:tangent_cone} \TanSpace(x,\US{K})=   \Big\{ \lambda u \mid x + u \in |K| \cap B(x, t) , \, \lambda \geq 0 \Big\}.
\end{equation}
Similarly, for any $\tau \in \Star{\nu}{{{K}}}$, we introduce the cone
\[
\TanSpace(x,\tau) =   \Big\{ \lambda u \mid x + u \in  \Interior{\Conv\tau} \cap B(x, t) , \, \lambda \geq 0 \Big\}. 
\]
Since $K$ is embedded, the family $\{\TanSpace(x,\tau) \mid \tau \in \Star{\nu}{{{K}}}\}$ is a partition of $\TanSpace(x,\US{K})$. Note that $x \mapsto \TanSpace(x,\US{K})$ remains constant on $\Interior{\Conv\nu}$ for each $\nu \in {{K}}$, thus depending only on the simplex of ${{K}}$ whose support contains $x$ in its relative interior. Additionally, we have:
\[
x \in \UpperSkin{{\US{K}}} \quad \iff \quad \normal{\pi_\M(x)} \notin \TanSpace(x,\US{K}).
\]
We next establish that:
\[
\Interior{\Conv\nu} \cap \UpperSkin {\US{K}}  \neq \emptyset \implies \Interior{\Conv\nu}
\subseteq \UpperSkin {\US{K}}.
\]
If $\dim \nu = 0$, the result is clear. For $\dim \nu > 0$, assume for contradiction that $x_0, x_1 \in \Interior{\Conv\nu}$ with $x_0 \in \UpperSkin{\US{K}}$ but $x_1 \notin \UpperSkin{\US{K}}$. The segment $[x_0, x_1]$ lies entirely within $\Interior{\Conv\nu}$. Since $K$ is vertically convex relative to $\M$, there exists $0 < r < \Reach{\M}$ such that $\US{K} \subseteq \Offset{\M}{r}$. Therefore, the projection map $\pi_\M$ is well-defined and continuous on $\US{K}$, as is the map $x \mapsto \normal{\pi_\M(x)}$. Consequently, at some point $x \in [x_0, x_1]$, we must have $\normal{\pi_\M(x)} \in \partial \TanSpace(x_0,\US{K})$, or equivalently:
\[
\normal{\pi_\M(x)} \in \partial \TanSpace(x,\US{K}).
\]
Observe that $\partial \TanSpace(x,\US{K})$ is contained in $\bigcup_{\tau} \TanSpace(x,\tau)$ such that $\tau \in \Star{\nu}{{{K}}}$ has dimension at most $d-1$. Hence, $\normal{\pi_\M(x)}$ belongs to $\TanSpace(x,\tau)$ for some $\tau \in \Star{\nu}{{K}}$ with $0 < \dim \tau < d$. This implies $\tau$ is vertical with $0 < \dim \tau < d$, contradicting the assumption that $K$ has a non-vertical skeleton. It follows that the relative interior of $\nu$ is contained in the upper skin. Moreover, the upper skin is closed, by Lemma~\ref{lemma:upper-and-lower-skins-homeomorphic-to-manifold}. Hence,
\[
\Interior{\Conv\nu} \subseteq \UpperSkin{{\US{K}}} \implies \Conv{\nu} \subseteq \UpperSkin{{\US{K}}}.
\]
The first result, for the upper skin, follows. The proof for the lower skin is analogous. 
\end{proof}

\uppLowComplex*

\begin{proof}
  By construction, $\US{\UpperC K} \subseteq \UpperSkin{\US{K}}$. For
  the reverse inclusion, consider $x \in \UpperSkin{\US{K}}$ and let
  $\nu \in K$ be the simplex whose support contains $x$ in its
  relative interior. By Lemma \ref{lemma:exclusive-upper-and-lower},
  $\Conv \nu \subseteq \UpperSkin{\US{K}}$, showing that $\nu \in
  \UpperC K$ and $x \in \US{\UpperC K}$. We have just established the
  first equality. The second one is obtained similarly. For the third
  one, applying Lemma
  \ref{lemma:upper-and-lower-skins-homeomorphic-to-manifold}, we obtain that
  \begin{align*}
    \US{\partial K}
    &= \UpperSkin{\US{K}} \cup \LowerSkin{\US{K}} \\
    &= \US{\UpperC K} \cup \US{\LowerC K} \\
    &= \US{\UpperC K \cup \LowerC K}.
  \end{align*}
  Since $\UpperC K \cup \LowerC K \subseteq \partial K$, the third
  equality follows. Finally, by Lemma
  \ref{lemma:upper-and-lower-skins-homeomorphic-to-manifold}, when
  $\pi_\M(\US{K}) = \M$, both the upper complex and lower complex of
  $K$ are triangulations of \M.
\end{proof}

\subsection{Upper and lower facets of a $d$-simplex}

\facetsLem*

\begin{proof}
  Recall that the set of $i$-simplices of a simplicial complex $K$ is denoted $K^{[i]}$.
  By Lemma \ref{lemma:upper-and-lower-complex-equal-boundary} applied with $K = \Cl\sigma$,
  \[
  (\Cl\sigma)^{[d-1]} = \Above \sigma \cup \Below \sigma.
  \]
  We show that $\Above \sigma \cap \Below \sigma =
  \emptyset$. Suppose, for contradiction, that a facet $\nu$ of
  $\sigma$ satisfies $\Conv\nu \subseteq \UpperSkin{\US{K}}$ and
  $\Conv\nu \subseteq \LowerSkin{\US{K}}$. This would imply
  $\Conv\sigma \subseteq \Aff \nu$, making $\sigma$ degenerate, which
  is impossible.
  	
  Next, we prove that $\Above \sigma \neq \emptyset$ by
  contradiction. If $\Above \sigma = \emptyset$, then $\Below \sigma$
  would contain all facets of $\sigma$, contradicting Lemma
  \ref{lemma:upper-and-lower-skins-homeomorphic-to-manifold}, which
  states that the two skins of $\Cl \sigma$ (relative to $\M$) are
  homeomorphic. Similarly, $\Below \sigma \neq \emptyset$.
\end{proof}

\clearpage
\section{Omitted proofs of Section \ref{section:vertical-collapses}}

\subsection{Correctness}

\correspondFreeTerminalNode*

\begin{proof}
  Since the second item of the lemma can be obtained from the first
  one by reversing the orientation of \M, we only need to prove the first
  item. We start with the forward direction $\Longrightarrow$. Suppose
  that $\tau \in K$ is a free simplex of $K$ from above relative to
  \M. This implies that
  \[
  {\Star \tau K}^{[d-1]} = \Above \sigma,
  \]
  where $\Sigma^{[i]}$ stands for the set of $i$-simplices of
  $\Sigma$.  Observe that $\Star \tau K
  \setminus \{\sigma\} \subseteq \partial K$ and therefore $\Above
  \sigma \subseteq \partial K$. Hence, $\sigma$ has no
  $d$-simplex in $K$ above it. In other words, $\sigma$ is a sink
  of the dual graph $G_\M(K)$.

  Let us prove the backward direction $\Longleftarrow$. Suppose that
  $\sigma$ is a sink of $G_\M(K)$. Equivalently, $\sigma$ is a
  $d$-simplex of $K$ such that ${\Above \sigma} \subseteq \partial
  K$. By Lemma \ref{lemma:facets}, $\Above \sigma \neq \emptyset$ and it
  follows that $\tau \neq \emptyset$ and $\tau \in \partial K$.

  To establish the backward direction $\Longleftarrow$, it
    suffices to show that $\tau$ is a free simplex of $K$. Indeed, if
    we are able to show this, it follows immediately that $\tau$ is also free from
    above with respect to \M since its coface $\sigma$ has dimension
    $d$ and its $(d-1)$-cofaces are the upper facets of $\sigma$ by
    construction, {\em i.e.}
  \[
    {\Star \tau K}^{[d-1]} = \Above \sigma.
  \]
  The rest of the proof is devoted to showing that $\tau$ is a free
  simplex of $K$. By definition, this boils down to proving that
  $\Link \tau K$ is the closure of a simplex. We do this by showing
  that
  \begin{equation}\label{eq:LinksEqual}
    \Link \tau K = \Link \tau {\Cl \sigma}.      
  \end{equation}
  We first claim that the above equation is implied by the following one
  \begin{equation}\label{eq:BoundaryOfLinksEqual}
     \Link{\tau}{ \partial K } =  \Link{\tau}{ \partial \Cl{\sigma}}.
  \end{equation}
  Indeed, suppose that we have \eqref{eq:BoundaryOfLinksEqual}. Since $
  \Cl{\sigma}\subseteq K$, we get that $\Link{\tau}{ \Cl{\sigma}}
  \subseteq \Link{\tau}{ K }$ and we only need to establish the
  reverse inclusion to get \eqref{eq:LinksEqual}. For this, we use the fact that for any simplicial
  complex $X$ embedded in $\Rspace^d$ and any $\tau \in X$, we have
  that $\partial \Link \tau X = \Link \tau {\partial X}$. We thus deduce
  from \eqref{eq:BoundaryOfLinksEqual} that
  \[
  \partial \Link{\tau}{ K } =  \partial \Link{\tau}{ \Cl{\sigma}}.
  \]
   Because $\tau \in \partial K$, it follows that $\US{\Link{\tau}{ K }}$
   is topologically equivalent to a proper subset of the
   $(d-k-1)$-sphere, where $k = \dim \tau$. Because $\tau \in \partial
   \Cl \sigma$, we have similarly that $\US{\Link{\tau}{ \Cl \sigma }}$ is also
   topologically equivalent to a proper subset of the
   $(d-k-1)$-sphere.  Since the latter link is a subset of the former
   link and the two links share the same boundary complex, the only
   possibility is that they are equal and we get
   \eqref{eq:LinksEqual}. We have just established our claim that
   \eqref{eq:BoundaryOfLinksEqual} $\implies$
   \eqref{eq:LinksEqual}. Hence, to prove that $\tau$ is a free
   simplex of $K$, it suffices to establish
   \eqref{eq:BoundaryOfLinksEqual}. Let
   \begin{align*}
     \Bold A &= \Link \tau {\partial K} \\
     \Bold B &= \Link \tau {\UpperC K} \\
     \Bold C &= \Link \tau {\UpperC {\Cl \sigma}}\\
     \Bold D &= \Link \tau {\partial \Cl\sigma}.
   \end{align*}
   We prove \eqref{eq:BoundaryOfLinksEqual} by showing that $\Bold A =
   \Bold B = \Bold C = \Bold D$ in three steps. Before we start, we
   write two useful equalities due to the fact $\partial \Cl \sigma$ is a pure simplicial $(d-1)$-complex:
   \begin{align}
     \Cl{(\Above {\sigma})} &= \UpperC {\Cl \sigma}, \label{eq:from-upper-facets-to-complex}\\
     \Cl{(\Below {\sigma})} &= \LowerC {\Cl \sigma}. \label{eq:from-lower-facets-to-complex}
   \end{align}
   We refer the reader to Figure~\ref{figure:proof-free-terminal} for a pictorial representation of the objects in the proof. 

   \begin{figure}[htb]
     \def\svgwidth{0.7\linewidth}
     \centering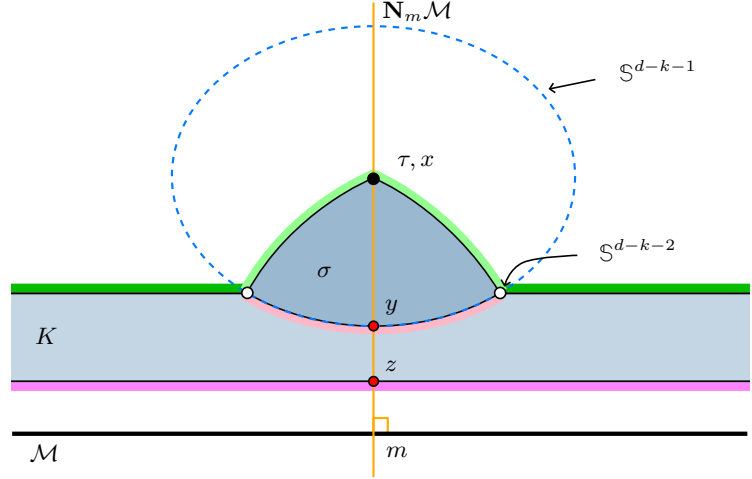
     \caption{Notation for the proof of
       Lemma~\ref{lemma:correspondance-free-terminal-node}. The upper
       skin of $K$ is in green and its lower skin in pink. We have
       schematically represented $\sigma$ using curved edges. Its
       upper skin (in pale green) is a subset of the upper skin of
       $K$. The $k$-simplex $\tau$ is represented as a black dot. Its
       link in $\partial K$ is represented by the two white dots and
       is proved to triangulate the $(d-k-2)$-sphere.
       \label{figure:proof-free-terminal}}
   \end{figure}

   \smallskip{\bf Step 1:} $\Bold A = \Bold B$. By Lemma
   \ref{lemma:facets}, $\Above \sigma$ and $\Below \sigma$ partition
   the facets of $\sigma$.  Let us make a general observation. Suppose
   that $\Sigma$ and $\Sigma'$ partition the set of facets of
   $\sigma$. Let $\tau = \bigcap \Sigma$ and consider a point $x$ in
   the relative interior of $\Conv {\Sigma}$. Since $\Cl \sigma$ is
   embedded, $x$ does not lie on $\US{\Cl{(\Sigma')}}$. Applying this
   with $\Sigma = \Above \sigma$ and $\Sigma' = \Below \sigma$, we get
   that for any point $x$ in the relative interior of $\Conv\sigma$,
   we have that
   \[
   x \not \in \US{\Cl{(\Below {\sigma})}} = \US{\LowerC {\Cl \sigma}} = \LowerSkin {\Cl \sigma},
   \]
   where the first equality comes from \eqref{eq:from-lower-facets-to-complex} and the second
   equality comes from Lemma \ref{lemma:upper-and-lower-complex-equal-boundary}.  Let $m = \pi_\M(x)$. By
   construction, the normal segment $\Normal m \M \cap B(m,\rtube)$
   passes through $x$. It intersects $\LowerSkin {\Cl \sigma}$ in a
   point $y$ and $\LowerSkin {K}$ in a point $z$ in such a way that
   along direction $\normal m$, $x$ is strictly above $y$ and $z$
   cannot be above $y$.  This shows that $x \neq z$ and consequently
   $x \not\in \LowerSkin K$. Since $\LowerSkin{K}$ is a closed set by
   Lemma \ref{lemma:upper-and-lower-skins-homeomorphic-to-manifold},
   $x$ is at distance greater than some $\epsilon > 0$ of
   $\LowerSkin{K}$.
      It follows
      that
   \begin{equation*}\label{eq:InBallAroundXSkinIsUpperSkin}
     B(x, \epsilon ) \cap  \US{ \partial K }   = B(x, \epsilon ) \cap  \US{ \UpperC {K} }.
   \end{equation*}
   Since $K$ is canonically embedded, the star (and therefore the link) 
   of $\tau$ in any subcomplex of $K$ depends only upon simplices whose support intersects some neighborhood of $x$. It follows that:
   \[
   \Link{\tau}{\partial K } =  \Link{\tau}{ \UpperC {K} }.
   \]

   \smallskip{\bf Step 2:} $\Bold C = \Bold D$. Consider again two sets $\Sigma$ and $\Sigma'$ that partition the set of facets of
   $\sigma$. Setting $\tau =
   \bigcap \Sigma$, one can observe that $\Link \tau {\Cl{(\Sigma)}} = \Link \tau {\partial
     \Cl \sigma}$. Applying this with $\Sigma = \Above \sigma$ and
   $\Sigma' = \Below \sigma$ and using \eqref{eq:from-upper-facets-to-complex}, we get that
   \[
   \Link \tau {\UpperC {\Cl \sigma}} = \Link \tau {\partial \Cl\sigma}.
   \]
   We see in particular that $\US{\Link \tau {\UpperC {\Cl \sigma}}}$ is a topological $(d-k-2)$-sphere.

   \smallskip{\bf Step 3:} $\Bold B = \Bold C$. Recall that
   ${\Above\sigma} \subseteq \partial K$. By Lemma
   \ref{lemma:upper-and-lower-complex-equal-boundary}, $\partial K$ is
   the union of the two subcomplexes $\UpperC K$ and $\LowerC K$ whose
   underlying spaces are the upper and lower skins of $K$ relative
   to \M, respectively. Each upper facet of $\sigma$ is contained in
   one of those two skins but not both. The only possibility is thus
   that $\Above {\sigma} \subseteq \UpperC {K}$ and therefore
   using~\eqref{eq:from-upper-facets-to-complex}
  \[
  \UpperC {\Cl{\sigma}} \subseteq  \UpperC {K}.
  \]
  Taking the  link of $\tau$ in each of those two subcomplexes
  of $K$, we get that
  \[
  \Link \tau {\UpperC {\Cl{\sigma}}} \subseteq \Link \tau {\UpperC {K}}.
  \]
  On the left side, we know from the previous step that the link is a
  triangulation of the $(d-k-2)$-sphere. We now study the link on the
  right side. By Lemma
  \ref{lemma:upper-and-lower-skins-homeomorphic-to-manifold},
  $\US{\UpperC {K}}$ is homeomorphic to \M. Since \M is a
  $(d-1)$-manifold without boundary, we deduce that so is $\US{\UpperC
    {K}}$. This implies that the link of any simplex in $\UpperC{K}$
  is a triangulation of the $(d-k-2)$-sphere. In summary, writing
  $X \approx Y$ if $X$ is homeomorphic to $Y$ and $\Sphere^i$ for the
  $i$-sphere, we obtain that
  \[
  \Sphere^{d-k-2} \, \approx \, \US{\Link{\tau}{\UpperC {\Cl{\sigma}} }} \, \subseteq \, \US{\Link{\tau}{ \UpperC {K} }} \, \approx \, \Sphere^{d-k-2}.
  \]
  The only subset of $\Sphere^{d-k-2}$ homeomorphic to
  $\Sphere^{d-k-2}$ being $\Sphere^{d-k-2}$ itself, we get that:
  \begin{equation*}\label{eq:LinkAreEqualsOnBoundary}
    \Link{\tau}{\UpperC {\Cl{\sigma}} }   = \Link{\tau}{ \UpperC {K} }
  \end{equation*}

  We have just proved that $\tau$ is a free simplex of $K$. This concludes the proof.
\end{proof}

\AlgorithmInvariant*

\begin{proof}
  Because $K' \subseteq K$, clearly $\US{K'} \subseteq \Offset \M
  \rtube$ and $K'$ has a non-vertical skeleton with respect to \M. Let us establish
  vertical convexity and covering projection for $K'$. Because $\tau$ is
  a vertically free simplex in $K$ relative to \M, the star
  $\Star \tau K$ has a unique inclusion-maximal simplex $\sigma$ whose
  dimension is $d$. Furthermore, one of the following two situations occurs:
  \[
    \text{(1) }  {\Star \tau K}^{[d-1]} = \Above \sigma \quad \text{or} \quad \text{(2) }  {\Star \tau K}^{[d-1]} = \Below \sigma,
  \]
  where $\Sigma^{[i]}$ stands for the set of $i$-simplices of
  $\Sigma$. Suppose that situation (1) occurs. In other words, $\tau$
  is a free simplex of $K$ from above relative to \M.
  This implies that the upper skin of $\US{\sigma}$
    is contained in the upper skin of $\US{K}$ and $\up {\Cl \sigma} m = \up K
    m$ for all $m \in \pi_\M(\Conv \sigma)$. We deduce that:
  \[
  \US{K'} \cap \Normal m \M \cap B(m,\rtube) =
  \begin{cases}
    [\low K m, \low {\Cl \sigma} m] & \text{if $m \in \pi_\M(\Conv \sigma)$},\\
    \US{K} \cap \Normal m \M \cap B(m,\rtube) & \text{otherwise}.
  \end{cases}
  \]
  We conclude that $K'$ is both vertically convex relative to \M and
  has a covering projection. Similarly, we obtain the same conclusion
  when situation (2) occurs. It remains to check that the acyclicity
  condition holds for $K'$ as well. Note that acyclicity of the
  relation $\below$ over $d$-simplices of $K$ is equivalent to
  acyclicity of the graph $G_\M(K)$.  By Lemma
  \ref{lemma:correspondance-free-terminal-node}, a vertical collapse
  in $K$ corresponds to the removal of a terminal node of
  $G_\M(K)$. We thus obtain acyclicity of the graph $G_\M(K')$ or,
  equivalently, acyclicity of the relation $\below$ over $d$-simplices
  of $K'$.
\end{proof}

\subsection{Practical version}

\InPractice*

\begin{proof}
  Let $\tau$ be a free simplex of $K$ whose inclusion-maximal coface
  $\sigma$ has dimension $d$. We prove the theorem by establishing the
  following equivalence: $\tau$ is vertically free in $K$ relative
  to \M if and only if $\tau$ is vertically free in $K$ relative to
  $\HH_\tau$.

  For any facet $\nu$ of $\sigma$, let $N_\nu$ denote the unit normal
  to $\Conv\nu$ pointing outwards $\Conv\sigma$, as illustrated in
  Figure~\ref{figure:outward-normals}.  We have assumed that for any
  facet $\nu$ of $\sigma$ and any vertex $a$ of $\nu$:
  \[
  \angle(\Aff \nu, \Tangent {\pi_\M(a)} \M) \, < \, \frac{\pi}{4}.
  \]
  It follows that for any facet $\nu$ of $\sigma$ and any vertex $a
  \in \nu$,
  \begin{align*}
    \angle(\normal {\pi_\M(a)}, N_\nu) &< \frac{\pi}{4}, && \text{if $\nu \in \Above \sigma$},\\
    \angle(\normal {\pi_\M(a)}, N_\nu) &> \frac{3\pi}{4}, && \text{if $\nu \in \Below \sigma$}.
  \end{align*}

  \begin{figure}[htb]
    \centering{%% Creator: Inkscape 1.3.2 (091e20e, 2023-11-25), www.inkscape.org
%% Accompanies image file 'socg-simplex-outward-normals.pdf' (pdf, eps, ps)
\begingroup%
  \makeatletter%
  \providecommand\color[2][]{%
    \errmessage{(Inkscape) Color is used for the text in Inkscape, but the package 'color.sty' is not loaded}%
    \renewcommand\color[2][]{}%
  }%
  \providecommand\transparent[1]{%
    \errmessage{(Inkscape) Transparency is used (non-zero) for the text in Inkscape, but the package 'transparent.sty' is not loaded}%
    \renewcommand\transparent[1]{}%
  }%
  \providecommand\rotatebox[2]{#2}%
  \newcommand*\fsize{\dimexpr\f@size pt\relax}%
  \newcommand*\lineheight[1]{\fontsize{\fsize}{#1\fsize}\selectfont}%
  \ifx\svgwidth\undefined%
    \setlength{\unitlength}{244.54839751bp}%
    \ifx\svgscale\undefined%
      \relax%
    \else%
      \setlength{\unitlength}{\unitlength * \real{\svgscale}}%
    \fi%
  \else%
    \setlength{\unitlength}{\svgwidth}%
  \fi%
  \global\let\svgwidth\undefined%
  \global\let\svgscale\undefined%
  \makeatother%
  \begin{picture}(1,0.49209294)%
    \lineheight{1}%
    \setlength\tabcolsep{0pt}%
    \put(0,0){\includegraphics[width=\unitlength,page=1]{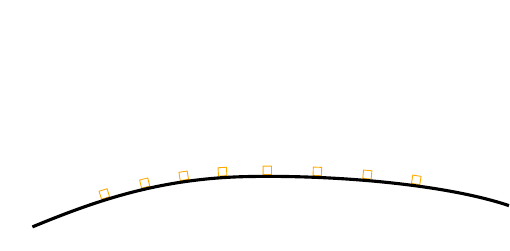}}%
    \put(0.04243613,0.07577115){\color[rgb]{0,0,0}\makebox(0,0)[t]{\lineheight{1.25}\smash{\begin{tabular}[t]{c}$\M$\end{tabular}}}}%
    \put(0,0){\includegraphics[width=\unitlength,page=2]{socg-simplex-outward-normals.pdf}}%
    \put(0.66211937,0.30469948){\makebox(0,0)[lt]{\lineheight{1.25}\smash{\begin{tabular}[t]{l}$\nu$\\\end{tabular}}}}%
    \put(0.65586693,0.46413164){\makebox(0,0)[lt]{\lineheight{1.25}\smash{\begin{tabular}[t]{l}$N_\nu$\end{tabular}}}}%
    \put(0.43182852,0.37347412){\makebox(0,0)[lt]{\lineheight{1.25}\smash{\begin{tabular}[t]{l}$\tau$\end{tabular}}}}%
    \put(0.44954319,0.25572367){\makebox(0,0)[lt]{\lineheight{1.25}\smash{\begin{tabular}[t]{l}$\sigma$\end{tabular}}}}%
  \end{picture}%
\endgroup%
}
    \caption{Outward-pointing normal vectors to a triangle in
      $\Rspace^2$ can be partitioned in two groups, depending on the
      angle they make. One of the group consists of normal vectors to upper
      facets of $\sigma$, the other group consists of normal vectors to lower facets
      of $\sigma$.
      \label{figure:outward-normals} }
  \end{figure}

  Given any pair of facets $\nu_0, \nu_1$ of $\sigma$, we note that we
  can always find a common vertex $a \in \nu_0 \cap \nu_1$. It follows
  that $\angle(N_{\nu_0},N_{\nu_1}) < \frac{\pi}{2}$ if $\nu_0$ and
  $\nu_1$ are both upper facets of $\sigma$ relative to \M or both
  lower facets of $\sigma$ relative to \M and
  $\angle(N_{\nu_0},N_{\nu_1}) > \frac{\pi}{2}$ otherwise.  We can
  thus partition the facets of $\sigma$ in two groups, depicted in two
  different colors in Figure \ref{figure:outward-normals}. In each
  group, any two facets have their outward-pointing normal vectors
  that make an angle smaller than $\frac{\pi}{2}$. Two facets not in
  the same group have their outward-pointing normal vectors making an
  angle greater than $\frac{\pi}{2}$. Let $\nu$ be any facet of
  $\sigma$. Let $N$ be one of the two unit normal vectors to $\Aff\nu$
  and use it to orient $\Aff\nu$.  If $N = N_\nu$ ({\em resp.} $N = -
  N_\nu$), facets in the group of $\nu$ are precisely those that are
  upper ({\em resp.} lower) facets of $\sigma$ relative to
  $\Aff\nu$.  We thus deduce that either
  \[
  \begin{cases}
    \UpperFacets \sigma \M = \UpperFacets \sigma {\Aff \nu},\\
    \LowerFacets \sigma \M = \LowerFacets \sigma {\Aff \nu},
  \end{cases}
  \]
  or
  \[
  \begin{cases}
    \UpperFacets \sigma \M = \LowerFacets \sigma {\Aff \nu},\\
    \LowerFacets \sigma \M = \UpperFacets \sigma {\Aff \nu},
  \end{cases}
  \]
  depending upon the choice of $N$ to orient $\Aff\nu$.  The result
  follows immediately by noting that for $\X \in \{\M,\Aff\nu\}$,
  the simplex $\tau$ is vertically free in $K$ relative to \X if and
  only if
  \[
  {\Star \tau K}^{[d-1]} = \UpperFacets \sigma \X \quad \text{or} \quad {\Star \tau K}^{[d-1]} = \LowerFacets \sigma \X,
  \]
  where $\Sigma^{[i]}$ stands for the set of $i$-simplices of
  $\Sigma$.  Thus, $\tau$ is vertically free in $K$ relative to \M
  if and only if $\tau$ is vertically free in $K$ relative to
  $\Aff\nu$.  We get the desired equivalence.
\end{proof}

\clearpage
\section{Additional material for Section \ref{section:alpha-complexes}}
\label{appendix:missing-material}

This appendix contains material that was omitted in Section
\ref{section:alpha-complexes}.

\subsection{Definition of $I(\varepsilon,\delta)$}
\label{appendix:range-alpha}

We define $I(\varepsilon,\delta)$ as the range of values of $\alpha$
used in \cite{attali2024optimal} for Theorem
\ref{theorem:NSW}. Following \cite{attali2024optimal}, we define
$I(\varepsilon,\delta)$ by distinguishing two cases. If $\delta \geq
\varepsilon$, then the authors in \cite{attali2024optimal} establish
that the following interval of values for $\alpha$ is valid:
\[
I(\varepsilon,\delta) = \left  [ \frac{1}{2} \left(\reach+\varepsilon  -  \sqrt{\Delta_0}\right),  
\frac{1}{2} \left(\reach+\varepsilon +  \sqrt{\Delta_0}\right)  \right  ],
\]
where $\Delta_0 = 2(\reach-\delta)^2-(\reach+\varepsilon)^2$. The
definition of $I(\varepsilon,\delta)$ is more involved when $\delta
\leq \varepsilon$. Let
\[
\Delta_1 = \frac{1}{\reach^2}\left(\varepsilon^2 - \left(\reach-\delta\right)^2\right)^2 - 10\left(\varepsilon^2 - \left(\reach-\delta\right)^2\right) - 7\reach^2.
\]
Assuming that $\varepsilon < \reach$ and $\delta < \reach$ satisfy the
strict homotopy condition:
\[
(\reach-\delta)^2 - \varepsilon ^2 \geq (4\sqrt{2}-5)\reach^2,
\]
it has been shown in \cite{attali2024optimal} that $\Delta_1$ is non-negative and thus, it makes sense to define
\[\begin{array}{l}
  \beta_{\min} =  \dfrac{1}{4} \left(\dfrac{(\reach-\delta)^2 +\reach^2 -\varepsilon ^2}{\reach} - \sqrt{ \Delta_1 }\right)
  \quad \text{ and } \\[3ex]
\beta_{\max} =  
\dfrac{1}{4} \left(\dfrac{(\reach-\delta)^2 +\reach^2 -\varepsilon ^2}{\reach} +\sqrt{ \Delta_1 }\right).
\label{eq:DefAlphaminMax} 
\end{array}\]
$\Delta_1$ in \cite{attali2024optimal} has been chosen precisely to be
the discriminant of the quadratic polynomial:
\[
\varepsilon^2 - \left(\reach-\delta\right)^2 + \reach^2 + \beta \,
\frac{1}{\reach}\left(\varepsilon^2 - \reach^2-
\left(\reach-\delta\right)^2\right) +2\beta^2
\]
and $[\beta_{\min},\beta_{\max}]$ represents the interval of values of
$\beta$ for which the polynomial is non-negative. The non-negativity
can be rewritten as:
\begin{equation}
  \label{eq:beta}
\left(1 +  \frac{\beta }{\reach} \right) \varepsilon^2 + \beta^2 + \frac{\beta}{\reach} \left( \reach^2 - (\reach - \delta)^2\right)
 \, \leq \,  (\reach -  \delta)^2 - (\reach - \beta)^2.  
\end{equation}
Let
\begin{align*}
  \alpha_{\min} &= \sqrt{\left(1 +  \frac{\beta_{\min} }{\reach} \right) \varepsilon^2 + \beta_{\min}^2 + \frac{\beta_{\min}}{\reach} \left( \reach^2 - (\reach - \delta)^2\right)},\\
  \alpha_{\max} &=  \sqrt{(\reach -  \delta)^2 - (\reach - \beta_{\max})^2},
\end{align*}
Equation \eqref{eq:beta} shows that $\alpha_{\min} \leq
\alpha_{\max}$. It is then proven in \cite{attali2024optimal} that
when $\delta \leq \varepsilon$ the non-empty interval
$I(\varepsilon,\delta) = [\alpha_{\min},\alpha_{\max}]$ can be used for Theorem \ref{theorem:NSW}.

\subsection{Possible $\beta$ for Theorem \ref{theorem:alpha-complex}}
\label{appendix:size-r-beta}

The next lemma enunciates a result of~\cite{attali2024optimal} that
provides a ``large'' value of $\beta$ for which $\Offset \M \beta
\subseteq \Offset P \alpha$ and which can be plugged in Theorem \ref{theorem:alpha-complex}.

\begin{lemma}[\cite{attali2024optimal}]
  \label{lemma:tubular-neighborhood-in-alpha-offset}
  Let $\M$ be a smooth submanifold whose reach is larger than or equal
  to \reach and let $P$ be a point set such that $\M \subseteq \Offset P
  \varepsilon$ for some $\varepsilon \geq 0$. Then, $\Offset \M \beta
  \subseteq \Offset P \alpha$ whenever
  \[
  \beta = -\frac{\varepsilon^2}{2\reach} + \sqrt{\alpha^2 +
      \frac{\varepsilon^4}{4\reach^2}-\varepsilon^2}.
  \]
\end{lemma}

\clearpage
\section{Proof of Theorem \ref{theorem:alpha-complex} in Section \ref{section:alpha-complexes}}
\label{appendix:alpha-complex-theorem}

The goal of this appendix is to establish Theorem
\ref{theorem:alpha-complex}:

\TheoremAlphaComplex*

\begin{remark}
  \label{remark:conditions-well-defined}
  It is easy to check that Conditions \eqref{eq:angle-condition-for-skeleton},
  \eqref{eq:angle-condition-for-vertical-convexity} and
  \eqref{eq:angle-condition-acyclicity} in Theorem
  \ref{theorem:alpha-complex} are well-defined. Because $P \subseteq
  \Offset \M \delta$, we obtain that $\US{\Del{P,\alpha}} \subseteq
  \Offset P \alpha \subseteq \Offset \M {(\delta+\alpha)}$. Since
  $\delta + \alpha < \reach \leq \Reach{\M}$,
    \[
    \US{\Del{P,\alpha}} \subseteq \Rspace^d \setminus \MA{\M}
    \]
    and the projection map $\pi_\M$ is thus well-defined on the
    support of any simplex in $\Del{P,\alpha}$. It follows that the
    left sides of Conditions \eqref{eq:angle-condition-for-skeleton},
    \eqref{eq:angle-condition-for-vertical-convexity} and
    \eqref{eq:angle-condition-acyclicity} are well-defined.
    We claim that 
    \[
    -1 \leq \frac{(\reach+\beta)^2 - (\reach+\delta)^2 - \alpha^2}{2(\reach+\delta)\alpha} \leq 1.
    \]
    The left inequality is equivalent to $\beta + \alpha \geq \delta$
    which holds since we have assumed that $\alpha \geq \delta$. The
    right inequality is equivalent to $\beta \leq \delta + \alpha$
    which holds because we have assumed that $P \subseteq \Offset
    \M \delta$ and $\Offset \M \beta \subseteq \Offset P \alpha$,
    implying that $\Offset \M \beta \subseteq \Offset \M {(\delta +
      \alpha)}$. Hence, the right side of Condition
    \eqref{eq:angle-condition-for-vertical-convexity} is
    well-defined. Similarly, using the assumption $\alpha \leq
    \frac{2}{3}(\reach - \delta)$, we easily get that
    \[
    0 \leq \frac{\alpha}{2(\reach - \delta - \alpha)} \leq 1,
    \]
    showing that the right side of Condition \eqref{eq:angle-condition-acyclicity} is also well-defined.
\end{remark}

\begin{proof}[Proof of Theorem \ref{theorem:alpha-complex}.]
  Injective projection for $\Del{P,\alpha}$ is shown in
  Section~\ref{appendix:alpha-complex-injective-projection} (Lemma~\ref{lemma:non-vertical}).
  Covering projection and vertical
  convexity for $\Del{P,\alpha}$ are established in
  Section~\ref{section:proof-technique}  (Lemma~\ref{lemma:alpha-complex-vertical-convexity}).
  Acyclicity is established in Section \ref{appendix:acyclicity}   (Lemma \ref{lemma:acyclicity}).
  Hence, all the assumptions of
  Lemma~\ref{lemma:upper-and-lower-complex-equal-boundary} and
  Theorem~\ref{theorem:correctness-generic-simplification} are
  satisfied. By
  Lemma~\ref{lemma:upper-and-lower-complex-equal-boundary}, the upper
  and lower complexes of $\Del{P,\alpha}$ relative to \M are each a
  triangulation of \M. By
  Theorem~\ref{theorem:correctness-generic-simplification},
  $\NaiveSquash(P,\alpha)$ returns a triangulation of \M.
\end{proof}

\subsection{Injective projection}
\label{appendix:alpha-complex-injective-projection}

In this section, we characterize non-verticality of an $i$-simplex $\tau
\in \Del{P,\alpha}$ relative to $\M$.

\begin{lemma}
  \label{lemma:non-vertical}
  Let $P$ be a finite point set in $\Rspace^d$ such that $P \subseteq
  \Offset \M \delta$ for some $\delta \geq 0$. Let $\alpha \in [0,
    \reach - \delta)$. Then, for all $i$-simplices $\tau \in
    \Del{P,\alpha}$ with $0 < i < d$,
  \[
  \eqref{eq:angle-condition-for-skeleton}
  \quad \iff \quad
  \text{$\tau$ is non-vertical relative to \M}
  \]
\end{lemma}

\begin{proof}
  It is easy to see that $\tau$ is
  non-vertical relative to $\M$ if and only if for all $x \in
  \Conv\tau$
  \begin{equation*}
    \angle( \Normal {\pi_\M(x)} \M, \Aff\tau) ) > 0.
  \end{equation*}
  Using Property \ref{property:complementary-angles}, 
  $\angle( \Normal {\pi_\M(x)} \M, \Aff\tau) ) = \frac{\pi}{2} -
  \angle( \Aff \tau, \Tangent {\pi_\M(x)} \M )$ and the above
  inequality is equivalent to
  $
  \angle( \Aff \tau, \Tangent {\pi_\M(x)} \M ) < \frac{\pi}{2}.
  $
  Hence, $\tau$ is non-vertical relative to \M if and only if
  \[
  \max_{x \in \Conv\tau} \angle(\Aff \tau, \Tangent {\pi_\M(x)} \M) \, < \, \frac{\pi}{2},
  \]
  which is equivalent to \eqref{eq:angle-condition-for-skeleton}.
\end{proof}

\subsection{Vertical convexity}
\label{appendix:alpha-complex-vertical-convexity}

Recall that vertical convexity for $\Del{P,\alpha}$ was proved in Section~\ref{section:proof-technique} (Lemma~\ref{lemma:alpha-complex-vertical-convexity}) and that the proof relies on Lemma~\ref{lemma:connecting-boundaries-simplified}. This section is dedicated to proving the latter lemma.

\ConnectingBoundariesSimplified*

Plugging the assumptions of Lemma
\ref{lemma:alpha-complex-vertical-convexity} inside Lemma
\ref{lemma:connecting-boundaries-simplified} and using Lemma
\ref{lemma:non-vertical} to replace
Condition~\eqref{eq:angle-condition-for-skeleton} with the assumption
that $\tau$ is non-vertical relative to \M, we can rewrite Lemma
\ref{lemma:connecting-boundaries-simplified} as follows:

\begin{lemma}
	\label{lemma:connecting-boundaries}
	Let $\M \subseteq \Offset P \varepsilon$ and $P \subseteq \Offset
    \M \delta$ for some $\varepsilon,\delta \geq 0$ that satisfy the
    strict homotopy condition.  Let $\alpha \in [\delta,\reach-\delta)
      \cap I{(\varepsilon,\delta)}$ and $\beta > 0$ be such that
      $\Offset \M \beta \subseteq \Offset P \alpha$. Suppose that for
      all $i$-simplices $\tau \in \Del{P,\alpha}$ with $0 < i < d$,
      $\tau$ is non-vertical relative to \M and
      Condition~\eqref{eq:angle-condition-for-vertical-convexity}
      hold.  Let $\spx \in \Del{P,\alpha}$ and $x \in
      \Interior{\Conv{\gamma}}$. If for some $\lambda > 0$ ({\em
        resp.}  $\lambda<0$), the segment $(x,x+\lambda \normal
               {\pi_\M(x)}]$ lies outside $\US{\Del{P,\alpha}}$, then
    it intersects an upper ({\em resp. lower} join).
\end{lemma}

\begin{figure}[htb]
  \centering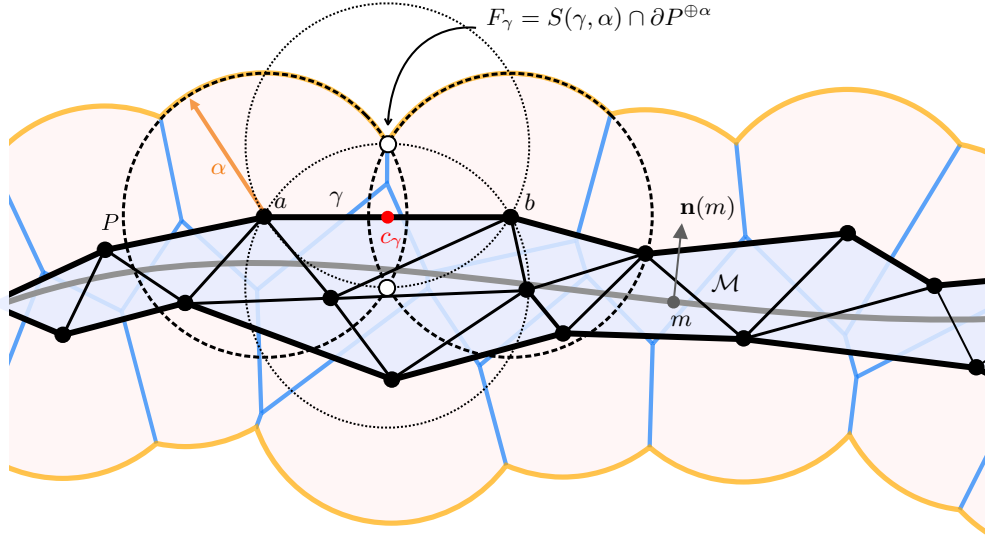
  \caption{For $d=2$ and $\gamma = \{a,b\}$, the set $S(\gamma,\alpha)$ is a 0-sphere that consists
    of the two white dots. Only one of them lies on $\partial \Offset
    P \alpha$ and therefore belongs to $F_\gamma$. In this particular case,
    $F_\gamma^+ = F_\gamma$.}
  \label{figure:notation-faces}
\end{figure}

Hence, proving Lemma \ref{lemma:connecting-boundaries-simplified} is equivalent
to proving Lemma \ref{lemma:connecting-boundaries}.  Before giving the
proof at the end of the section, consider a finite point set $P
\subseteq \Rspace^d$ and recall that, as explained in Section
\ref{section:proof-technique}, the boundary of $\Offset P \alpha$ can be
decomposed into faces. Those faces are in one-to-one correspondence
with the boundary simplices of $\Del{P,\alpha}$. Associate to each
boundary $i$-simplex $\spx \in \partial \Del{P,\alpha}$ the
$(d-i-1)$-sphere
\[
S(\spx,\alpha) = \bigcap_{p \in \spx} \partial B(p,\alpha).
\]
Equivalently, $S(\spx,\alpha)$ can be defined as the locus of the
centers of the $(d-1)$-spheres with radius $\alpha$ that circumscribe
$\spx$. The center $c_\spx$ of $S(\spx,\alpha)$ coincides with the
center of the smallest $(d-1)$-sphere circumscribing $\spx$; see
Figure \ref{figure:notation-faces}.  The face of
$\partial \Offset P \alpha$ dual to $\spx \in \partial\Del{P,\alpha}$
can then be described as
\[
F_\spx = S(\spx,\alpha) \cap \partial \Offset P \alpha.
\]
We also recall that the upper face $F_\spx^+$ ({\em resp.} lower face
$F_\spx^-$) designates the restriction of $F_\spx$ to the upper skin
({\em resp.} lower skin) of $\Offset P \alpha$. The set $F_\spx^+ *
\Conv \spx$ ({\em resp.} $F_\spx^- * \Conv \spx$) designates the
upper ({\em resp.} lower) join of $\spx$.

  \begin{figure}[htb]
    \def\svgwidth{\linewidth}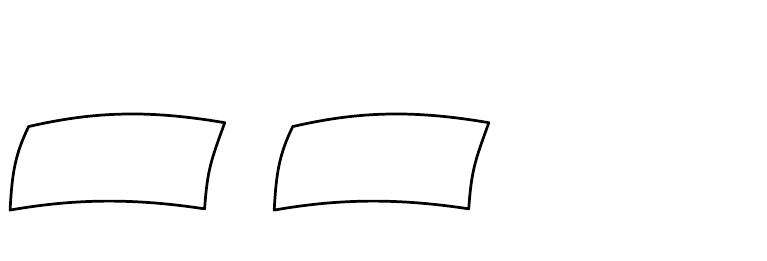
    \caption{Highest point $s_\spx(x)$ of $S(\spx,\alpha)$ in
      direction $\normal {x^*}$ when $\M$ is a surface in $\Rspace^3$
      and $\dim\spx = 0$ (left), $\dim\spx = 1$ (middle), and
      $\dim\spx = 2$ (right).
      \label{figure:highest-point-on-dual-sphere}}
  \end{figure}

  More generally, given any simplex $\spx \subseteq \Rspace^d$ (not necessarily in $\Del{P,\alpha}$), we
  define $S(\spx,\alpha) = \bigcap_{p \in \spx} \partial B(p,\alpha)$.
  Assuming that $\Conv{\spx} \subseteq \Rspace^d \setminus \MA{\M}$ and
  $S(\spx,\alpha) \neq \emptyset$, we associate to $\spx$ and any $x \in \Conv{\spx}$
  the highest point of $S(\spx,\alpha)$ in direction $\normal
  {\pi_\M(x)}$ which we denote as $s_\spx(x)$; see Figure
  \ref{figure:highest-point-on-dual-sphere}.  Let $\Bis{\spx}$ be the
  bisector of $\spx$.  By construction, $s_\spx(x) \in S(\spx,\alpha)
  \subseteq \Bis{\spx}$. Writing $x^* = \pi_\M(x)$ for short, we make
  the following remark:

  \begin{remark}
    \label{remark:cone-s-plus}
    For any simplex $\spx \subseteq \Rspace^d$ and $\alpha \geq 0$
    such that $\Conv{\spx} \subseteq \Rspace^d \setminus \MA{\M}$ and
    $S(\spx,\alpha) \neq \emptyset$ and for any $x \in \Conv{\spx}$,
    \begin{equation*}
      \angle(\normal{x^*} , s_\spx(x) - c_\spx) = \angle(\Normal {x^*} \M, \Bis \spx) = \angle(\Aff {\spx}, \Tangent {x^*} \M).
    \end{equation*}
  \end{remark}

  We start with a lemma.

  \begin{lemma}
    \label{lemma:highest-point-on-upper-skin}
    Assume that $\M \subseteq \Offset P \varepsilon$ and $P \subseteq
    \Offset \M \delta$ for some $\varepsilon,\delta \geq 0$, and let
    $\alpha \in [\delta,\reach-\delta)$. Let $\beta > 0$ be such that $\Offset
      \M \beta \subseteq \Offset P \alpha$. Consider a simplex $\tau$
      with $0 < \dim \tau < d$ satisfying~\eqref{eq:angle-condition-for-vertical-convexity}.
      Suppose $\tau$ is a boundary simplex of $\Del{P, \alpha}$. Then, $s_\tau(a)$ exists for all $a \in \tau$, and there exists $a \in \tau$ at which the following implication holds:
    \begin{equation}
      \label{eq:implication-upper-skin}
    s_\tau(a) \text{ lies on } \partial\Offset P \alpha \quad
    \implies \quad s_\tau(a) \text{ lies on the upper skin of } \Offset P \alpha.
    \end{equation}
  \end{lemma}

  \begin{figure}[htb]
    \centering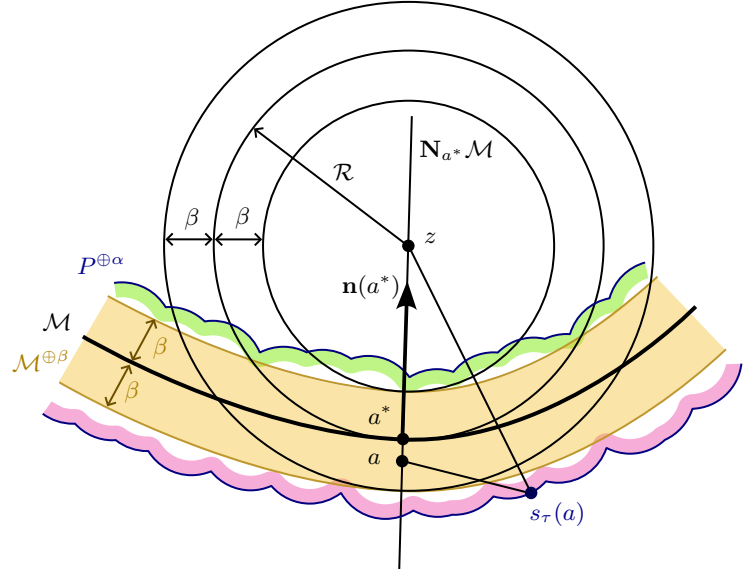
    \caption{We reach a contradiction in the proof of Lemma
      \ref{lemma:highest-point-on-upper-skin} by assuming that $s_\tau(a)$ lies on the lower skin of $\partial \Offset P \alpha$.
      \label{figure:connecting-boundaries}
    }
  \end{figure}

  \begin{proof}
    Because $\tau$ is a boundary simplex of $\Del{P,\alpha}$, we have
    that $S(\tau,\alpha) \neq \emptyset$ and for any $x \in
    \Conv{\tau}$, one can define $s_\tau(x)$ as the
    point of $S(\tau,\alpha)$ highest in direction $\normal {\pi_\M(x)}$.
    Thus, $s_\tau(x)$ exists for all $x \in \Conv{\tau}$.
    Consider $a \in \tau$
    and write $a^* = \pi_\M(a)$ for short.  Suppose that $a \in \tau$
    is such that
    \begin{align}
      \label{eq:claim-skin}
      s_\tau(a) \text{ lies on } \partial\Offset P \alpha
      \quad \text{ and } \quad  s_\tau(a) \text{ does not lie on the upper skin of } \Offset P \alpha.
    \end{align}
    The only possibility is thus that $s_\tau(a)$ lies on the lower
    skin of $P$, as illustrated in Figure
    \ref{figure:connecting-boundaries}.  Let $z = {a^*} + \reach
    \normal {a^*}$ and note that the ball $B(z,\reach)$ is tangent to
    \M at $a^*$ and does not contain any point of \M in its
    interior. Since $\Offset \M \beta \subseteq \Offset P \alpha$, we
    deduce that the lower skins of both $\Offset \M \beta$ and
    $\Offset P \alpha$ do not intersect the interior of $B(z,\beta +
    \reach)$. Hence, $\|z-{s_\tau(a)}\| \geq \reach + \beta$. By
    construction, $\|a-s_\tau(a)\| = \alpha$ and $\|a - z\| \leq
    \|a-{a^*}\| + \|{a^*}-z\| = \delta + \reach$. Consider the
    triangle with vertices $a$, ${s_\tau(a)}$ and $z$. Using the
    formula that gives the angle at vertex $a$ in triangle
    $az{s_\tau(a)}$ with respect to the length of the sides of the
    triangle, we obtain
  \begin{align}
    \notag
    \cos\left( \angle(\normal {a^*} , s_\tau(a) - a) \right)
    \, &= \, \frac{\|a-z\|^2 + \|a-{s_\tau(a)}\|^2 - \|z-{s_\tau(a)}\|^2}{2 \|a-z\| \|a-{s_\tau(a)}\|}\\
    \notag
    \, &\leq \, \frac{\|a-z\|^2 + \alpha^2 - (\reach+\beta)^2 }{2 \|a-z\| \alpha}\\
    \label{eq:angle-towards-lower-skin}
    \, &\leq \, \frac{(\delta+\reach)^2 + \alpha^2 - (\reach+\beta)^2 }{2 (\delta+\reach) \alpha},
  \end{align}
  where the first inequality comes from the fact that
  $\|a-s_\tau(a)\|=\alpha$ and $\|z-{s_\tau(a)}\| \geq \reach + \beta$
  and the second inequality comes from the fact that $\lambda \mapsto
  \frac{\lambda^2 - A}{\lambda}$ is increasing for $A >0$ and
  $\lambda>0$, applied with $\lambda = \|a-z\|$ and $A = - \alpha^2 +
  (\reach+\beta)^2$. We are going to exhibit a point $a \in \tau$ for
  which \eqref{eq:angle-towards-lower-skin} does not hold, showing
  that for that $a$, statement \eqref{eq:claim-skin} is false and its
  negation, statement \eqref{eq:implication-upper-skin}, is true. By
  Remark \ref{remark:conditions-well-defined}
  \[
    -1 \, \leq \, \frac{(\reach+\beta)^2 - (\reach+\delta)^2 -
      \alpha^2}{2(\reach+\delta)\alpha} \, \leq  \, 1.
    \]
    and it makes sense to introduce the angle
  \[
    \theta = \arcsin \left(\frac{(\reach+\beta)^2 - (\reach+\delta)^2
      - \alpha^2}{2(\reach+\delta)\alpha} \right) \in
    \left[-\frac{\pi}{2},\frac{\pi}{2}\right].
    \]
  Assumption \eqref{eq:angle-condition-for-vertical-convexity}
  can then be rewritten as $\min_{a \in \tau} \angle(\Aff \tau,
  \Tangent {\pi_\M(a)} \M) < \theta$ and we let $a \in \tau$ be any vertex
  such that
  \begin{equation}
  \label{eq:angle-condition-for-vertical-convexity-obtuse-in-proof}
  \angle(\Aff\tau,\Tangent {a^*} \M) < \theta.
  \end{equation}
  Let us show that:
  \begin{equation}
    \label{eq:inside-cone-obtuse}
    \angle (\normal {a^*} , s_\tau(a) - a) < \frac{\pi}{2} + \theta.
  \end{equation}

  \begin{figure}[htb]
    \centering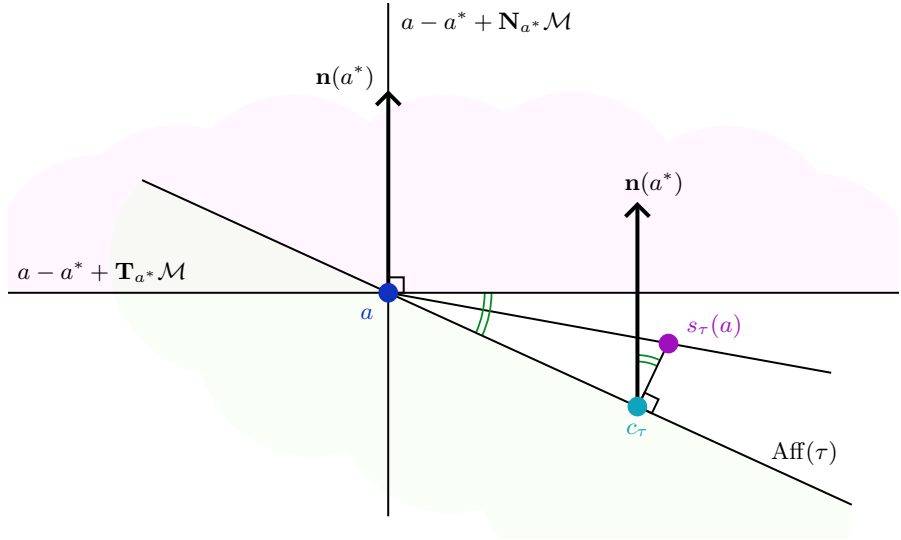
    \caption{Notation for the proof of Lemma
      \ref{lemma:highest-point-on-upper-skin}.
      \label{figure:angle-normal-direction-highest-point}
    }
  \end{figure}

  \noindent If $\angle (\normal {a^*} , s_\tau(a) - a) < \frac{\pi}{2}$,
  clearly \eqref{eq:inside-cone-obtuse} holds. If $\angle (\normal {a^*} ,
  s_\tau(a) - a) \geq \frac{\pi}{2}$ (see Figure
  \ref{figure:angle-normal-direction-highest-point}), then using
  \eqref{eq:angle-condition-for-vertical-convexity-obtuse-in-proof} and Remark
  \ref{remark:cone-s-plus}, we get that
  \begin{align*}
    \theta \, &> \, \angle(\Aff {\tau}, \Tangent {a^*} \M) = \angle (\normal {a^*}, s_\tau(a) - c_\tau) 
    = \angle (\normal {a^*}, c_\tau - a) - \frac{\pi}{2} \\
    \, &\geq \, \angle (\normal {a^*}, s_\tau(a) - a) - \frac{\pi}{2}.
  \end{align*}
  Hence, we obtain \eqref{eq:inside-cone-obtuse} as desired. Taking
  the cosine on both sides of \eqref{eq:inside-cone-obtuse}, we obtain that
  \begin{align*}    
    \cos\left( \angle(\normal {a^*} , s_\tau(a) - a) \right)
    > - \sin\theta = \frac{(\reach+\delta)^2 + \alpha^2 - (\reach+\beta)^2 }{2 (\reach+\delta)\alpha}.
  \end{align*}
  But, this contradicts \eqref{eq:angle-towards-lower-skin} which was
  a consequence of \eqref{eq:claim-skin}. Hence, for that particular
  $a$, \eqref{eq:claim-skin} is false. We thus deduce that there
  exists $a \in \tau$ such that if $s_\tau(a)$ lies on
  $\partial \Offset P \alpha$, then $s_\tau(a)$ lies on the upper skin
  of $\Offset P \alpha$. This concludes the proof.
  \end{proof}

  Before proving Lemma \ref{lemma:connecting-boundaries}, we make a useful remark:

\begin{remark}
  \label{remark:sign-dot-product-invariant}
  Consider a $(d-1)$-simplex $\nu \subseteq \Rspace^d$ that is non-vertical relative to \M
  and whose support is contained in $\Rspace^d \setminus \MA{\M}$ and let
  $N$ be a unit vector orthogonal to $\Aff{\nu}$. If $N \cdot \normal
  {\pi_\M(x)} > 0$ for some $x \in \Conv\nu$, then $N \cdot \normal
  {\pi_\M(x)} > 0$ for all $x \in \Conv\nu$.
\end{remark}

\begin{proof}[Proof of Lemma \ref{lemma:connecting-boundaries}.]
  The proof is by descending induction over the dimension of $\spx$.

  \smallskip \noindent \underline{Induction basis:} Suppose that
  $\spx$ has dimension $d-1$. Then, $\Aff\spx$ is an hyperplane and
  $S(\spx,\alpha)$ is a $0$-sphere that consists of two points
  (possibly equal).  Consider $x \in \Interior{\Conv{\spx}}$ and
  $\lambda > 0$ such that the segment $(x,x+\lambda
  \normal{\pi_\M(x)}]$ lies outside $\US{\Del{P,\alpha}}$.  Let $N$ be
  the unit vector orthogonal to $\Aff \spx$ that satisfies
  \[
  \normal{\pi_\M(x)} \cdot N > 0
  \]
  and observe that, for $\mu>0$ small enough, the segment
  $(x,x+\mu N]$ lies outside $\US{\Del{P,\alpha}}$. Let $s$ be
the highest point of $S(\spx,\alpha)$ along direction $N$. By duality
between the boundary simplices of $\Del{P,\alpha}$ and the faces of
$\partial \Offset P \alpha$, we deduce that $s$ lies on
$\partial\Offset P \alpha$.  Let $a \in \gamma$. Since $\gamma$ is non-vertical, Remark
\ref{remark:sign-dot-product-invariant} tells us that
    \[
    \normal{\pi_\M(a)} \cdot N > 0.
    \]
    Note that $s$ is also the point of $S(\spx,\alpha)$ highest along
    direction $\normal{\pi_\M(a)}$, namely $s = s_\spx(a)$. Thus,
    $s_\spx(a)$ lies on  $\partial\Offset P \alpha$. Applying
    Lemma \ref{lemma:highest-point-on-upper-skin} with $\tau=\spx$, we
    get that $s_\spx(a)$ lies on the upper skin of $\Offset P
    \alpha$. Hence, the set $\{s_\spx(a)\} * \Conv\spx$ is contained
    in the upper join of $\spx$ and the segment $(x,x+\lambda \normal
    {\pi_\M(x)}]$ has a non-empty intersection with the upper join of
      $\spx$. This concludes the induction basis.

    \smallskip \noindent \underline{Induction step:} Let $k \in
    \{0,\ldots,d-1\}$ and assume that the lemma holds for all $\spx$
    with dimension $i \geq k+1$. Let us prove that the lemma also
    holds when $\spx$ has dimension $k$.  We distinguish two cases:

    \smallskip\styleitem{(a)} $\spx$ has no proper coface in
    $\Del{P,\alpha}$. In that case, the whole sphere $S(\spx,\alpha)$
    lies on $\partial \Offset P \alpha$ and so does $s_\spx(a)$ for
    any $a \in \gamma$. Pick $a \in \gamma$. By Lemma
    \ref{lemma:highest-point-on-upper-skin} applied with $\tau =
    \spx$, $s_\spx(a)$ lies on the upper skin of $\Offset P
    \alpha$. Since for $0 \leq k < d-1$, the $(d-k-1)$-sphere
    $S(\spx,\alpha)$ is connected, the whole sphere $S(\spx,\alpha)$
    has to lie on the upper skin. Hence, $S(\spx,\alpha) * \Conv \spx$
    is the upper join of $\spx$. It follows that for any point $x \in
    \Interior{\Conv \spx}$ and any $\lambda >0$, the upper join $S(\spx,\alpha) *
    \Conv\spx$ has a non-empty intersection with the segment
    $(x,x+\lambda \normal {\pi_\M(x)}]$. This concludes
      case~\styleitem{(a)}.

  \begin{figure}[htb]
    \centering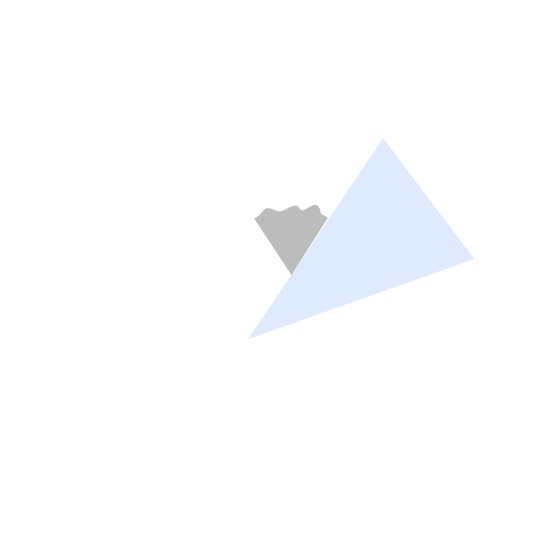
    \caption{
      Construction of a path $\Gamma$ connecting $z \in
      \Interior{\Conv\spx}$ to $z'\in\Interior{\Conv\sigma}$ when
      $d=2$ and $\spx$ is a vertex. As we move any point $x \in [z,z']$ in direction
      $\normal{x^*}$, $x$ enters the gray cone formed of all points $p$
      such that $\angle(\normal{z^*},p-x) < \varphi$ and therefore goes outside $\US{\Del{P,\alpha}}$.
      \label{figure:path-face-to-coface}}
  \end{figure}

  \smallskip\styleitem{(b)} $\spx$ has at least one proper coface in
  $\Del{P,\alpha}$; see Figure \ref{figure:path-face-to-coface}.
  Consider $z \in \Interior{\Conv \spx}$ and $\lambda > 0$ such that the
  segment $(z,z+\lambda \normal{\pi_\M(z)}]$ lies outside
  $\US{\Del{P,\alpha}}$ and let us prove that the segment intersects an upper
  join. Write $z^* = \pi_\M(z)$ for short. Let $\Sigma =
  \Star{\spx}{\Del{P,\alpha}}$ denote the set of cofaces of $\spx$ in
  $\Del{P,\alpha}$. Recall that the tangent cone to $\US{\Sigma}$
  at $z$ (see Equation~\eqref{eq:tangent_cone}) is denoted $\TanSpace(z,\US{\Sigma})$. Let
  \[
  \varphi = \min_{0 \neq v \in \TanSpace(z,\US{\Sigma})}
  \angle(\normal{z^*},v).
  \]
  Because of our assumption that $(z,z+\lambda \normal{\pi_\M(z)}]$ lies
    outside $\US{\Del{P,\alpha}}$, we have that $\varphi > 0$.
    Consider a ball $B$ centered at $z$ whose radius is small enough
    so that $B \subseteq \Offset P \alpha$ and $B$ does not intersect
    the support of any simplex in $\Del{P,\alpha} \setminus
    \Sigma$. Note that $\US{\Sigma} \cap B$ is star-shaped at $z$,
    that is,
  \[
  x \in \US{\Sigma} \cap B \quad \implies \quad [x,z] \subseteq \US{\Sigma} \cap B.
  \]
  By considering $B$ small enough, we may assume furthermore that for
  all $x \in \US{\Sigma} \cap B$ with $x \neq z$, the half-line with
  origin at $z$ and passing through $x$ intersects $\US{\Sigma} \cap
  B$ in a radius of $B$, that is, a segment connecting the center of
  $B$ to its bounding sphere $\partial B$. We also assume that $B$ is
  small enough so that for all $x \in B$
  \begin{equation*}
    \label{eq:cone-outside-around-north-pole}
    \angle (\normal{z^*},\normal{x^*}) < \varphi,
  \end{equation*}
  where, as usual, we write $x^* = \pi_\M(x)$. Consider the $(d-k)$-dimensional
  affine space $H$ passing through $z$ and orthogonal to $\Aff\spx$.
  Associate to each proper coface $\sigma \in \Sigma \setminus
  \{\spx\}$ the angle:
  \[
  \psi_\sigma = \min_{s \in \Conv{\sigma} \cap \partial B \cap H} \angle (\normal{z^*}, s-z).
  \]
  Note that  $\psi_\sigma \geq \varphi$ because $\varphi$ can be expressed as
  \[
  \varphi = \min_{s \in \US{\Sigma} \cap \partial B} \angle (\normal{z^*}, s-z).
  \]
  Let $\sigma \in \Sigma \setminus \{\spx\}$ that achieves the
  smallest $\psi_\sigma$. One can always find $z' \in
  \Interior{\Conv{\sigma}} \cap B$ and $\mu>0$ such that the
  $\mu$-offset of the segment $[z,z']$ is contained in $B$; see
  Figure \ref{figure:path-face-to-coface}.  Let
  \[
  f(x) = x + \mu \normal{x^*}.
  \]
  By construction, for all $x \in [z,z']$, we have $(x,f(x)] \subseteq
    B \subseteq \Offset P \alpha$. Moreover,
  \[
  (x,f(x)] \subseteq \Offset P \alpha \setminus \US{\Del{P,\alpha}}.
    \]
    Indeed,  $\sigma$ has been chosen so that any
    point $x \in \Conv\sigma \cap B$ is the highest point of
    $\US{\Sigma} \cap B$ along direction $\normal{z^*}$ and therefore
    \[
    (x,x+\mu \normal{z^*}] \subseteq \Offset P \alpha \setminus
      \US{\Del{P,\alpha}}
      \]
      Because $\angle (\normal{z^*},\normal{x^*}) < \varphi \leq
      \psi_\sigma$, the same remains true when replacing
      $\normal{z^*}$ with $\normal{x^*}$ in the above
      inclusion; see Figure \ref{figure:path-face-to-coface}. Consider now the path
  \[
  \Gamma = (z,f(z)] \cup f([z,z']) \cup [f(z'),z').
  \]
  By construction,
  \[
  \Gamma \subseteq  \Offset P \alpha \setminus \US{\Del{P,\alpha}}
  \]
  and, by Remark \ref{remark:intersection-joins}, $\Gamma$ is covered
  either by upper joins or by lower joins but cannot be covered by
  both. Applying the induction hypothesis to $\sigma$, we obtain that
  sufficiently close to $z' \in \Conv\sigma$, the path $\Gamma$ is
  covered by an upper join. Hence, the entire path $\Gamma$ is covered
  by upper joins and the segment $(z,z+\lambda \normal{z^*}]$ intersects
    an upper join. This concludes case~\styleitem{(b)} and the proof
    of the lemma.
\end{proof}

\subsection{Acyclicity}
\label{appendix:acyclicity}

The relation $\below$ over the $d$-simplices of a simplicial complex
$K$ is not acyclic in general; see Figure~\ref{figure:cycle-2d}.
However, when $K$ is the $\alpha$-complex of a finite point set $P$
within distance $\delta$ of $\M$, we state conditions under which
$\below$ becomes acyclic. The result is akin to an acyclicity result
on Delaunay complexes in \cite{edelsbrunner2001geometry}.

\begin{lemma}[Acyclicity]
  \label{lemma:acyclicity}
  Let $P$ be a finite point set in $\Rspace^d$ such that $P \subseteq
  \Offset \M \delta$ for some $\delta \geq 0$. Let $\alpha \in
   \left[0,\frac{2(\reach - \delta)}{3}\right)$.  Assume that all $(d-1)$-simplices $\nu \in
    \Del{P,\alpha}$ are non-vertical relative to \M and satisfy
    \[
      \min_{x \in \Conv\nu} \angle(\Aff \nu, \Tangent {\pi_\M(x)} \M) \, < \,  
      \frac{\pi}{2} - 2 \arcsin\left( \frac{\alpha}{2(\reach - \delta - \alpha)} \right).
    \]
    Then, the relation $\below$ defined over the Delaunay $d$-simplices is acyclic.
\end{lemma}

Any $(d-1)$-simplex $\nu$ that satisfies the conditions of the above
lemma is non-vertical relative to \M and has a support contained in $\Rspace^d
\setminus \MA{\M}$. 

\begin{definition}
  \label{definition:consistently-oriented}
  Consider a $(d-1)$-simplex $\nu$ that is non-vertical relative to
  \M and whose support is contained in $\Rspace^d \setminus
  \MA{\M}$. We say that the unit vector $N$ normal to $\Aff\nu$ is
     {\em consistently oriented} with \M if there exists $x \in
     \Conv\nu$ such that $N \cdot \normal {\pi_\M(x)} > 0$.
\end{definition}

By Remark \ref{remark:sign-dot-product-invariant}, $\nu$ has exactly
one of its two unit normal vectors that is consistently oriented with
\M and for that unit vector $N$, we have that $N \cdot \normal
   {\pi_\M(x)} > 0$ for all $x \in \Conv\nu$.  Let $Z(\sigma)$
   designates the center of the $(d-1)$-sphere that circumscribes a
   $d$-simplex $\sigma$ and associate to each point $z \in \Rspace^d
   \setminus \MA{\M}$ the real number:
\[
\alt{z} = (z - \pi_\M(z)) \cdot \normal{ \pi_\M(z) }.
\]

The proof of Lemma \ref{lemma:acyclicity} relies on the following
lemma:

\begin{lemma}
  \label{lemma:height-is-increasing}
  Let $P$ be a finite point set in $\Rspace^d$ such that $P \subseteq
   \Offset \M \delta$ for some $\delta \geq 0$. Let $\alpha \in
           \left[0,\frac{2(\reach - \delta)}{3}\right)$ and consider a $(d-1)$-simplex $\nu \in
             \Del{P,\alpha}$ that is non-vertical relative to
             \M. Assume that
     \[
     \min_{x \in \Conv\nu} \angle(\Aff \nu, \Tangent {\pi_\M(x)} \M) \, < \,  
     \frac{\pi}{2} - 2 \arcsin\left( \frac{\alpha}{2(\reach - \delta - \alpha)} \right).
     \]
   \begin{enumerate}
   \item Let $N$ be the unit normal vector to $\Aff\nu$ that is
     consistently oriented with \M. As we move a point $z$ on the
     segment $V(\nu,P) \cap \Offset P \alpha$ along direction $N$,
     $\alt{z}$ is increasing.
   \item Suppose that $\nu$ is the common facet of two
     $d$-simplices $\sigma_0, \sigma_1 \in \Del{P,\alpha}$ such that
     $\sigma_0 \below \sigma_1$.  As we move a point $z$ on the
     Voronoi edge dual to $\nu$ in direction $Z(\sigma_1) -
     Z(\sigma_0)$, $\alt{z}$ is increasing and therefore:
     \[
     \sigma_0 \below \sigma_1 \quad \implies \quad \alt{Z(\sigma_0)} < \alt{Z(\sigma_1)}.
     \]
   \end{enumerate} 
\end{lemma}

\begin{proof}
  We first note that $V(\nu,P) \cap \Offset P \alpha$ is indeed a segment because
  \[
  V(\nu,P) \cap \Offset P \alpha = V(\nu, P) \cap \bigcap_{p \in
    \nu} B(p,\alpha)
  \]
  and this segment is parallel to $N$.
  Let $z_0$ and $z_1$ be the two endpoints of that segment and name
  them so that $(z_1-z_0)\cdot N \geq 0$.  For all $t \in [0,1]$, let
  $z(t) = (1-t) z_0 + t z_1$ and $\varphi(t) = \alt{z(t)}$. We show
  the first item of the lemma by establishing that $\varphi'(t) > 0$
  for all $t \in (0,1)$. Since $\nabla \alt{z} = \normal {\pi_\M(z)}$
  and $z'(t) = z_1 - z_0$, we obtain that
  \begin{align*}
    \varphi'(t) &= \nabla \alt{z(t)} \cdot z'(t)\\
    &= \normal {\pi_\M(z(t))} \cdot (z_1 - z_0).
  \end{align*}
  To show that $\varphi'(t) > 0$, it suffices to show that
  $\angle(\normal {z(t)} , z_1 - z_0) < \frac{\pi}{2}$ for all $t \in
  (0,1)$; see Figure~\ref{figure:proof-increasing-height}, left. Consider
  $t \in (0,1)$ and $x \in \Conv\nu$. Write $z = z(t)$, $z^* =
  \pi_\M(z(t))$, and $x^* = \pi_\M(x)$ for short. We have
  \begin{equation*}
    \angle(\normal {z^*} , z_1 - z_0) \leq \angle(\normal {z^*}, \normal {x^*}) + \angle( \normal {x^*}, z_1-z_0 ).
  \end{equation*}

  \begin{figure}[htb]
    \begin{center}
      \def\svgwidth{.49\linewidth}
      \raisebox{1.5cm}{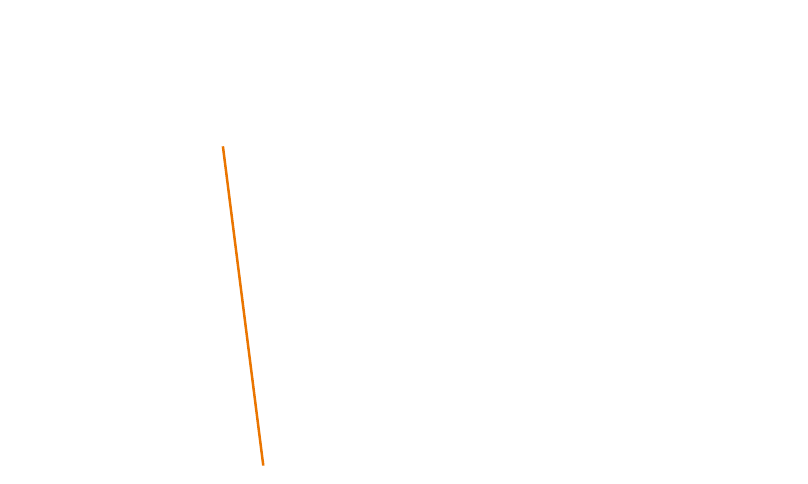}
      \def\svgwidth{.49\linewidth}
      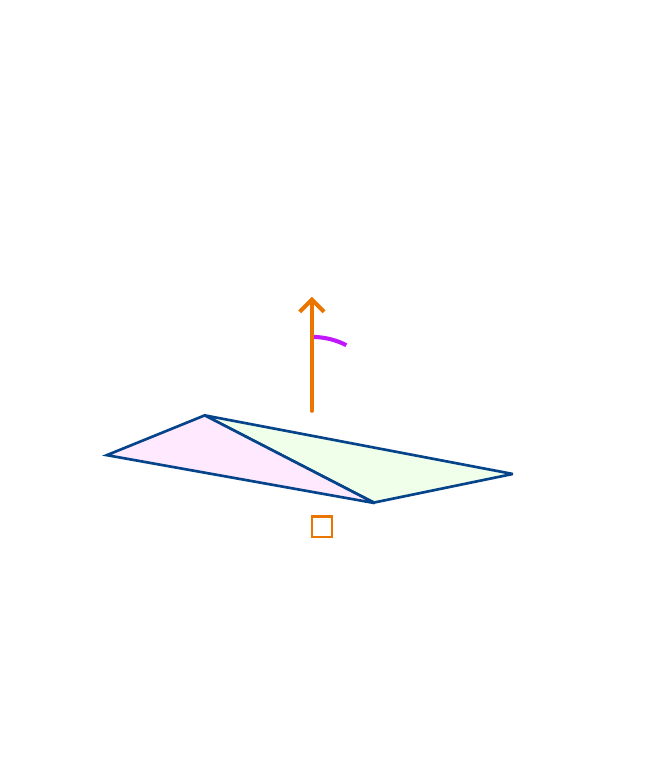
    \end{center}
    \caption{Notation for the proof of Lemma
      \ref{lemma:height-is-increasing}. The angle $\angle(\normal {x^*}, z_1 - z_0)$ is
       in purple. \label{figure:proof-increasing-height}}
  \end{figure}

  Let us first upper bound $\angle(\normal {z^*}, \normal {x^*})$. Set $r = \delta
  + \alpha < \Reach{\M}$. Since $P \subseteq \Offset \M \delta$, we get that $x,z
  \in \Offset P \alpha \subseteq \Offset \M {(\delta+\alpha)} =
  \Offset \M r$. We know from \cite[page 435]{federer1959curvature}
  that for $r < \Reach{\M}$, the projection map $\pi_\M$ onto
  $\M$ is $\left( \frac{\reach}{\reach - r} \right)$-Lipschitz for
  points in $\Offset \M r$.  Hence,
  $ \|z^* - x^*\| \leq \frac{\reach}{\reach-\alpha-\delta} \, \|z-x\|$.
  Then, using a result in \cite{boissonnat2018geometric} which says that for all 
   $m,m' \in \M$, we have 
  $\angle (\normal m, \normal {m'}) \leq 2 \arcsin \left( \frac{\|m'-m\|}{2\reach} \right)$, 
  we obtain that
 \begin{equation*}
  \angle(\normal {z^*}, \normal {x^*}) \leq 2 \arcsin\left( \frac{\|z^* -
    x^*\|}{2\reach} \right) \leq 2 \arcsin \left( \frac{\|z-x\|}{2(\reach-\alpha-\delta)} \right).
 \end{equation*}
 Because $z \in V(\nu,P) \cap \bigcap_{p
   \in \nu} B(p,\alpha)$ and $x \in \Conv{\nu}$, we deduce that
  $\|z-x\| \leq \alpha$ and 
 \begin{equation}
   \label{eq:bound-angle-between-normals}
  \angle(\normal {z^*}, \normal {x^*}) \leq 2 \arcsin \left(
  \frac{\alpha}{2(\reach-\alpha-\delta)} \right).
 \end{equation}
 Because $N$ is consistently oriented with \M, we have that $N \cdot
 \normal{x^*} > 0$ and
 \begin{equation}
   \label{eq:relation-normals-tangents}
    \angle( \normal {x^*}, z_1-z_0 ) = \angle( \normal {x^*}, N ) =  \angle (\Tangent {x^*} \M, \Aff {\nu})
 \end{equation}
 Using \eqref{eq:bound-angle-between-normals} and \eqref{eq:relation-normals-tangents}, we obtain that
 \begin{align*}
   \angle(\normal {z^*} , z_1 - z_0) &\leq \angle(\normal {z^*}, \normal {x^*}) + \angle( \normal {x^*}, z_1-z_0 )\\   &\leq  2 \arcsin \left(
  \frac{\alpha}{2(\reach-\alpha-\delta)} \right) + \angle (\Tangent {x^*} \M, \Aff {\nu}).
 \end{align*}
 Since the above inequality is true for all $x$, we get that
 \begin{equation*}
   \angle(\normal {z^*} , z_1 - z_0) \leq
   2 \arcsin \left(\frac{\alpha}{2(\reach-\alpha-\delta)} \right) + \min_{x \in \Conv\nu} \angle (\Tangent {x^*} \M, \Aff {\nu})
 \end{equation*}
 and using our angular hypothesis, we deduce that $ \angle(\normal
 {z^*} , z_1 - z_0) < \frac{\pi}{2}$ and $\varphi'(t) > 0$ for all $t \in (0,1)$ as desired.

 \medskip

 For the proof of the second item of the lemma, define $z_0$ and $z_1$
 as in the proof of the first item and note that the hypothesis
 $\sigma_0 \below \sigma_1$ implies $z_1 = Z(\sigma_1)$, $z_0 =
 Z(\sigma_0)$, and $N = \frac{z_1-z_0}{\|z_1-z_0\|}$; see
 Figure~\ref{figure:proof-increasing-height}, right.  The first item
 then implies that
 \[
 \alt{Z(\sigma_0)} < \alt{Z(\sigma_1)},
 \]
 which yields the second item.
\end{proof}

\begin{proof}[Proof of Lemma \ref{lemma:acyclicity}.]
   Let $\sigma_0, \sigma_1 \in \Del{P,\alpha}$ be two $d$-simplices
   with a common $(d-1)$-face. If $\sigma_0 \below \sigma_1$, then
   Lemma~\ref{lemma:height-is-increasing} implies that
   $\alt{Z(\sigma_0)} < \alt{Z(\sigma_1)}$. The acyclicity of the
   relation $\below$ follows, because no function can increase along a
   cycle.
\end{proof}

\clearpage
\section{Sampling conditions for surfaces in $\Rspace^3$}
\label{appendix:alpha-complex-second-corrolary}

\ignore{
\biancacmt{It would make it easier to know what we are referring to if the section title would be more explicit, like "Proof of the sampling conditions on a surface" or "When is the sample a triangulation of M".} \Domi{What about the following section title?}
}

This appendix contains the proof of
Theorem~\ref{corollary:3D-theory} in
Section~\ref{section:alpha-complexes}. Recall that
\[
  \beta_{\varepsilon,\alpha} = -\frac{\varepsilon^2}{2\reach} + \sqrt{\alpha^2 + \frac{\varepsilon^4}{4\reach^2}-\varepsilon^2}.
\]

\CorollaryThreeDimensionTheory*

\bigskip

As an intermediate step, we establish Lemma \ref{lemma:3D-conditions}.

\begin{lemma}
  \label{lemma:3D-conditions}
  Let \M be a surface in $\Rspace^3$ and $P \subseteq \M$ a finite
  point set such that $\M \subseteq \Offset P \varepsilon$. Let $\delta = 0$ and $\beta =
  \beta_{\varepsilon,\alpha}$. Consider $\varepsilon \geq 0$ and $\alpha \geq 0$ that satisfy 
  the following condition:
  \begin{equation}
    \label{eq:3D-domain}
    \frac{\sqrt{3} \, \alpha}{\reach} < \min \left\{
    \frac{(\reach+\beta_{\varepsilon,\alpha})^2 - \reach^2 - \alpha^2}{2\reach\alpha},\,
    \cos\left( 2 \arcsin\left( \frac{\alpha}{\reach} \right)\right)
    \right\}.
  \end{equation}
  Then, the assumptions of Theorem
  \ref{theorem:alpha-complex} are satisfied. If furthermore
  \begin{equation}
    \label{eq:3D-practical}
    \frac{\sqrt{3} \, \alpha}{\reach} < \sin\left( \frac{\pi}{4} - 2 \arcsin \left( \frac{\alpha}{\reach} \right)\right),
  \end{equation}
   then the assumptions of Corollary
   \ref{theorem:correcness-squash-practical} are also satisfied.
\end{lemma}

\begin{figure}
	\centering
    \def\svgwidth{.5\textwidth}
    %% Creator: Inkscape 1.3.2 (1:1.3.2+202311252150+091e20ef0f), www.inkscape.org
%% Accompanies image file 'corollary-3D-homotopy.pdf' (pdf, eps, ps)
\begingroup%
  \makeatletter%
  \providecommand\color[2][]{%
    \errmessage{(Inkscape) Color is used for the text in Inkscape, but the package 'color.sty' is not loaded}%
    \renewcommand\color[2][]{}%
  }%
  \providecommand\transparent[1]{%
    \errmessage{(Inkscape) Transparency is used (non-zero) for the text in Inkscape, but the package 'transparent.sty' is not loaded}%
    \renewcommand\transparent[1]{}%
  }%
  \providecommand\rotatebox[2]{#2}%
  \newcommand*\fsize{\dimexpr\f@size pt\relax}%
  \newcommand*\lineheight[1]{\fontsize{\fsize}{#1\fsize}\selectfont}%
  \ifx\svgwidth\undefined%
    \setlength{\unitlength}{348.17102051bp}%
    \ifx\svgscale\undefined%
      \relax%
    \else%
      \setlength{\unitlength}{\unitlength * \real{\svgscale}}%
    \fi%
  \else%
    \setlength{\unitlength}{\svgwidth}%
  \fi%
  \global\let\svgwidth\undefined%
  \global\let\svgscale\undefined%
  \makeatother%
  \begin{picture}(1,0.90454487)%
    \lineheight{1}%
    \setlength\tabcolsep{0pt}%
    \put(0,0){\includegraphics[width=\unitlength,page=1]{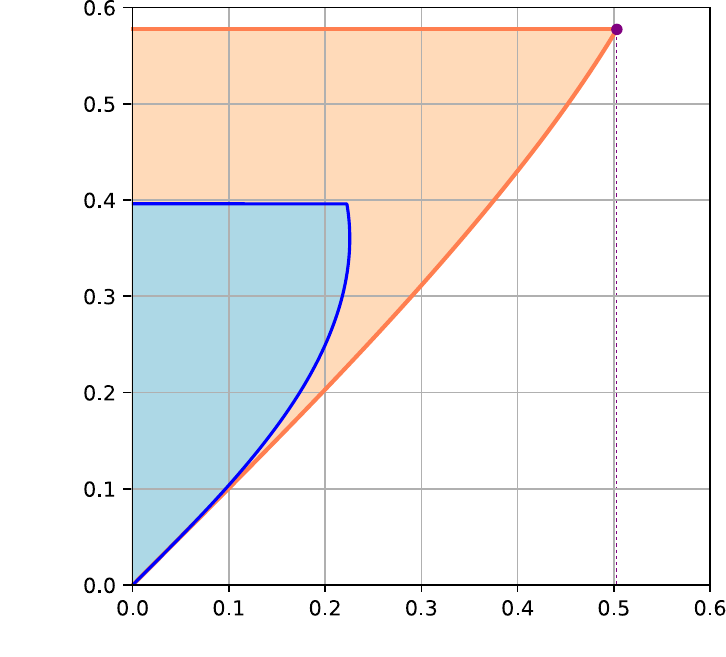}}%
    \put(0.53785752,0.00418201){\color[rgb]{0,0,0}\makebox(0,0)[lt]{\lineheight{1.25}\smash{\begin{tabular}[t]{l}$\varepsilon/\reach$\end{tabular}}}}%
    \put(-0.00166888,0.48470549){\color[rgb]{0,0,0}\makebox(0,0)[lt]{\lineheight{1.25}\smash{\begin{tabular}[t]{l}$\alpha/\reach$\end{tabular}}}}%
    \put(0.05562407,0.83065006){\color[rgb]{0.50196078,0,0.50196078}\makebox(0,0)[lt]{\lineheight{1.25}\smash{\begin{tabular}[t]{l}$1/\sqrt{3}$\end{tabular}}}}%
    \put(0.86418425,0.11580868){\color[rgb]{0.50196078,0,0.50196078}\makebox(0,0)[lt]{\lineheight{1.25}\smash{\begin{tabular}[t]{l}$0.503$\end{tabular}}}}%
    \put(0,0){\includegraphics[width=\unitlength,page=2]{corollary-3D-homotopy.pdf}}%
  \end{picture}%
\endgroup%

	\caption{Pairs of $(\frac{\varepsilon}{\reach},\frac{\alpha}{\reach})$ for which Condition \eqref{eq:3D-feasible-range} is satisfied (in orange) and for which both Conditions \eqref{eq:3D-domain} and \eqref{eq:3D-feasible-range} are satisfied (in blue), for $P \subseteq \M \subseteq \Offset P \varepsilon$ and $d=3$.}
	\label{fig:sampling-condition-3D-homotopy}
\end{figure}

\begin{proof}
  For short, set $I(\varepsilon) = I(\varepsilon,0)$.  We are going to
  establish the lemma, adding to Condition \eqref{eq:3D-domain} the
  following one:
  \begin{equation}
    \label{eq:3D-feasible-range}
    \varepsilon \leq \sqrt{6 - 4 \sqrt{2}}\, \reach \quad \text{and}
    \quad \alpha \in I(\varepsilon) \cap \left[0,\frac{\reach}{\sqrt{3}}\right].
  \end{equation}
  Indeed, one can check numerically that the domain of $\varepsilon
  \geq 0$ and $\alpha \geq 0$ defined by \eqref{eq:3D-domain}
  and \eqref{eq:3D-feasible-range} is exactly the same as the one defined by
  \eqref{eq:3D-domain} alone; see Figure \ref{fig:sampling-condition-3D-homotopy}. Hence, adding
  \eqref{eq:3D-feasible-range} to \eqref{eq:3D-domain}
  does not change the statement of the lemma, except that under that
  form, the lemma becomes easier to establish.

  Let us first check that the assumptions of Theorem
  \ref{theorem:alpha-complex} are satisfied. In Theorem
  \ref{theorem:alpha-complex}, we require that $\varepsilon, \delta
  \geq 0$ satisfy the strict homotopy condition. Since $\delta = 0$,
  this condition can be rewritten as $\reach^2-\varepsilon^2 \geq
  (4\sqrt{2}-5)\reach^2$ or equivalently
\[
\varepsilon \leq \sqrt{6 - 4 \sqrt{2}}\, \reach.
\]
We also require that $\alpha \in I(\varepsilon,\delta) \cap \left[\delta,\frac{2(\reach-\delta)}{3}\right)$ which can be rewritten as $\alpha \in I(\varepsilon) \cap \left[0,\frac{\reach}{\sqrt{3}}\right]$ since $\delta=0$.  Lemma~\ref{lemma:tubular-neighborhood-in-alpha-offset} ensures that
  for $\beta = \beta_{\varepsilon,\alpha}$, we have $\Offset \M \beta
  \subseteq \Offset P \alpha$, as required in Theorem
  \ref{theorem:alpha-complex}. Suppose that $\tau$ is either an edge
  or a triangle of $\Del{P,\alpha}$. Then, $\rcirc{\tau} \leq \alpha$
  and applying Lemmas~\ref{lemma:angleBetweenEdgeAndTangentPlane} and
  \ref{lemma:angleBetweenTriangleAndTangentPlane}, we obtain that for
  $\alpha \leq \frac{\reach}{\sqrt{3}}$
  \begin{equation}
    \label{eq:min-angle-bound}
\min_{x \in \Conv\tau} \angle(\Aff \tau, \Tangent {\pi_\M(x)} \M) \leq \min_{a \in \tau} \angle(\Aff \tau, \Tangent {\pi_\M(a)} \M) \leq
\arcsin\left( \frac{\sqrt{3} \, \alpha}{\reach}\right).    
  \end{equation}
Applying Lemma
\ref{lemma:difference-max-min-angular-deviation-noiseless}, we deduce that for any edge and triangle $\tau
\in \Del{P,\alpha}$
  \begin{equation}
    \label{eq:max-angle-bound}
\max_{a \in \tau} \angle(\Aff \tau, \Tangent {\pi_\M(a)} \M) \leq \arcsin\left(\frac{\sqrt{3} \, \alpha}{\reach}\right) + 2 \arcsin \left( \frac{\alpha}{\reach} \right).
  \end{equation}
Replacing the left side of Condition~\eqref{eq:angle-condition-for-skeleton} with the right side of the above inequality, we obtain the stronger condition:
\begin{equation}
  \label{eq:3D-domain-C1}
\arcsin\left(\frac{\sqrt{3} \, \alpha}{\reach}\right) < 
    \frac{\pi}{2} - 2 \arcsin\left( \frac{\alpha}{\reach} \right).
\end{equation}
Similarly, when replacing the left side of
Conditions~\eqref{eq:angle-condition-for-vertical-convexity} and
\eqref{eq:angle-condition-acyclicity} with
the right side of \eqref{eq:min-angle-bound} and after
plugging $\beta = \beta_{\varepsilon,\alpha}$ and $\delta = 0$ on the
right side, we get the respectively stronger conditions:
\begin{align}
  \label{eq:3D-domain-C2}
  \arcsin\left(\frac{\sqrt{3} \, \alpha}{\reach}\right) &< 
  \arcsin\left(\frac{(\reach+\beta_{\varepsilon,\alpha})^2 - \reach^2 - \alpha^2}{2\reach\alpha}\right),\\
  \label{eq:3D-domain-C3}
  \arcsin\left(\frac{\sqrt{3} \, \alpha}{\reach}\right) &< 
  \frac{\pi}{2} - 2 \arcsin\left( \frac{\alpha}{2(\reach - \alpha)} \right).
\end{align}
Hence, for $\beta=\beta_{\varepsilon,\alpha}$ and $\delta = 0$, Conditions \eqref{eq:angle-condition-for-skeleton},
\eqref{eq:angle-condition-for-vertical-convexity}, and
\eqref{eq:angle-condition-acyclicity} are implied by
\eqref{eq:3D-domain-C1}, \eqref{eq:3D-domain-C2}, and
\eqref{eq:3D-domain-C3} which can be rewritten as \eqref{eq:3D-domain}
since \eqref{eq:3D-domain-C3} is redundant with \eqref{eq:3D-domain-C1}.

Second, let us check that if furthermore Condition
\eqref{eq:3D-practical} is satisfied, the assumptions of Corollary
\ref{theorem:correcness-squash-practical} are also satisfied.
When
replacing the left side of Condition
\eqref{eq:angle-condition-practical} with the right side of 
\eqref{eq:max-angle-bound}, we get a stronger condition, namely
\eqref{eq:3D-practical}.
\end{proof}

Thanks to the above lemma, the proof of
Theorem~\ref{corollary:3D-theory} is now straightforward.

\begin{proof}[Proof of Theorem~\ref{corollary:3D-theory}.]
  Let $\delta = 0$ and $\beta = \beta_{\varepsilon,\alpha}$.  By Lemma
  \ref{lemma:3D-conditions}, if \itemref{item:naive-3D-sampling}
  holds, then the assumptions of Theorem \ref{theorem:alpha-complex}
  are met. It can thus be applied to obtain the first part of the
  theorem. In addition, if \itemref{item:practical-3D-sampling}
  holds, the assumptions of Corollary
  \ref{theorem:correcness-squash-practical} are met. It can thus be
  applied to obtain the second part of the theorem.
\end{proof}

\clearpage
\section{The restricted Delaunay complex}
\label{appendix:restricted-Delaunay-complex}

Given $\M \subseteq \Rspace^d$ and a finite $P
\subseteq \Rspace^d$, we recall that the {\em restricted Delaunay
  complex} is
\[
\DelR{P} = \{ \sigma \subseteq P \mid \sigma \neq \emptyset \text{ and
} V(\sigma,P) \cap \M \neq \emptyset \}.
\]
The goal of this section is to establish the following theorem.

\TheoremRestrictedDelaunayComplex*

Instrumental to the proof, we introduce the core Delaunay
complex. Given $\M \subseteq \Rspace^d$ and a finite point set $P
\subseteq \Rspace^d$, the {\em core Delaunay complex} is the
collection of $d-1$ simplices in the restricted Delaunay complex
together with all their faces:
\[
\DelC{P} = \Cl{\left(\DelR{P}^{[d-1]}\right)}.
\]

\begin{proof}[Proof of Theorem \ref{theorem:3D-restricted-Delaunay-complex}]
  By Lemma \ref{lemma:3D-crossing-Delaunay}, $\DelC{P}$ triangulates
  \M. By Lemma \ref{lemma:WhenRestrictedDelaunayIsPure2D}, the
  restricted Delaunay complex is a pure 2-dimensional simplicial
  complex. In other words,
  \[
  \DelR{P} = \DelC{P},
  \]
  yielding the result.
\end{proof}

In the rest of the section, we establish the two lemmas needed for the
proof of Theorem~\ref{theorem:3D-restricted-Delaunay-complex}.

\subsection{The core Delaunay complex}
\label{appendix:core-Delaunay-complex}

In this section, we provide conditions under which the core Delaunay
complex triangulates \M. The first lemma is true for all dimensions $d$
and the second one is true for $d=3$ and is the one used in the proof
of Theorem \ref{theorem:3D-restricted-Delaunay-complex}.

\begin{lemma}
	\label{lemma:crossing-Delaunay}
	Under the assumptions of Theorem~\ref{theorem:alpha-complex} and
    the additional generic assumption that none of the Voronoi
    vertices of $P$ lie on \M, there exists an execution of
    $\NaiveSquash(P,\alpha)$ that outputs $\DelC{P}$. Consequently,
    $\DelC{P}$ is a triangulation of \M.
\end{lemma}

\begin{algorithm}
  \caption{$\NonCrossingVertSimp(K)$
  }
  \begin{algorithmic}
    \WHILE{ $G_\M(K) \neq \emptyset$ }
    \IF{ $G_\M(K)$ has a sink $\sigma$ whose circumcenter lies above \M }
    \STATE { Collapse  $\bigcap \Above \sigma$ in $K$;}
    \ELSIF { $G_\M(K)$ has a source $\sigma$ whose circumcenter lies below \M }
    \STATE { Collapse  $\bigcap \Below \sigma$ in $K$;}
    \ENDIF
    \ENDWHILE
  \end{algorithmic}
  \label{algo:non-crossing-simplification}
\end{algorithm}

\begin{algorithm}
  \caption{$\NonCrossingSquash(P,\alpha)$}
  \begin{algorithmic}
	\STATE{ $K \leftarrow \Del{P,\alpha}$;  $\NonCrossingVertSimp(K)$; {\bf return} $K$; }
  \end{algorithmic}
  \label{algo:non-crossing-squash}
\end{algorithm}

\begin{proof}
  Before we start, recall that given a point $x \in \Rspace^d \setminus
  \MA{\M}$,
  \[
  \alt{x} = (x - \pi_\M(x)) \cdot \normal{ \pi_\M(x) }
  \]
  is the signed distance of $x$ from \M. We shall say that a point $x
  \in \Rspace^d \setminus \MA{X}$ lies above \M iff $\alt{x} \geq 0$,
  $x$ lies below \M iff $\alt{x} \leq 0$, and $x$ lies on \M iff
  $\alt{x} = 0$. We also note that the core Delaunay complex can 
  equivalently be defined as
  \[
  \DelC{P} = \Cl{\left( \{ \nu \in \Del{P}^{[d-1]} \mid V(\nu,P) \cap \M
    \neq \emptyset\} \right)}.
  \]
  The proof is done in four steps:

  \Step{1} We show that $\NonCrossingSquash(P,\alpha)$ corresponds to
  one particular execution of $\NaiveSquash(P,\alpha)$.  For this, we
  show that when running $\NonCrossingVertSimp(K)$
  on $K =
  \Del{P,\alpha}$, each execution of the while-loop results in one
  vertical collapse relative to \M. This boils down to proving that at
  least one of the two conditional expressions in $\NonCrossingVertSimp(K)$
  evaluates to true
  whenever $G_\M(K) \neq \emptyset$. Suppose that $G_\M(K) \neq
  \emptyset$ and contains a node $\sigma_0$, which, by definition, is
  a $d$-simplex of $K$. Let $Z(\sigma_0)$ denote the circumcenter of
  $\sigma_0$. We consider two cases:

  \smallskip \styleitem{(a)} Suppose that $Z(\sigma_0)$ lies above
  \M. Because $G_\M(K)$ is a directed acyclic graph, we can find a
  directed path that connects $\sigma_0$ to a sink $\sigma$.  By Lemma
  \ref{lemma:height-is-increasing}, the map that associates to each
  node $v$ the quantity $\alt{Z(v)}$ is increasing along the directed
  path. Since $Z(\sigma_0)$ lies above \M, so does $Z(\sigma)$. Thus,
  we find a sink whose circumcenter lies above \M.

  \smallskip \styleitem{(b)} Suppose that $Z(\sigma_0)$ lies below
  \M. Similarly, we find a source $\sigma$ whose circumcenter lies
  below \M.

  \smallskip

  This shows that at least one of the two conditional expressions in
  $\NonCrossingVertSimp(K)$
  evaluates to true whenever
  $G_\M(K) \neq \emptyset$. 

  \Step{2} We show that
  $
  \DelC P \subseteq \Del{P,\alpha}.
  $
  Indeed, consider a $(d-1)$-simplex $\gamma \in \DelC P$. By
  definition, there exists a sphere that circumscribes $\gamma$ and
  whose center lies on \M. Because $\M \subseteq \Offset P
  \varepsilon$, we deduce that the radius of that sphere is
  $\varepsilon$ at most and $\gamma \in \Del{P,\varepsilon}$.  Because
  $\alpha \in I(\varepsilon,\delta) = [\alpha_{\min},\alpha_{\max}]$,
  one can easily check from the definition of $\alpha_{\min}$ in
  Appendix~\ref{appendix:range-alpha} that $\varepsilon \leq \alpha$
  and therefore $\gamma \in \Del{P,\alpha} = K$.

  \Step{3} We show that when running
  $\NonCrossingVertSimp(K)$
  on $K = \Del{P,\alpha}$, we have the loop invariant:
  \begin{equation*}
    \label{eq:loop-invariant}
    \DelC P \subseteq K \subseteq \Del{P,\alpha}.
  \end{equation*}
  This is true before beginning the while-loop thanks to the previous
  step. We only need to prove that during an execution of the
  while-loop, one cannot remove $(d-1)$-simplices of $\DelC P$. We
  examine in turn each of the two possibilities that may occur:

  \smallskip \styleitem{(a)} $G_\M(K)$ has a sink $\sigma$ whose
  circumcenter $Z(\sigma)$ lies above \M and
  $
  \bigcap \Above \sigma
  $
  is collapsed in $K$. In that case, the set of $(d-1)$-simplices of
  $K$ that disappear are exactly the upper facets of $\sigma$. Each
  upper facet $\gamma$ of $\sigma$ is dual to some Voronoi edge
  $V(\gamma,P)$. Let
  \begin{equation}
    \label{eq:restricted-Voronoi-edge}
      \gamma^* = V(\gamma,P) \cap \Offset P \alpha = V(\gamma, P) \cap
  \bigcap_{p \in \gamma} B(p,\alpha).
  \end{equation}
  $\gamma^*$ is a segment with one endpoint at $Z(\sigma)$. By
  Lemma~\ref{lemma:height-is-increasing}, the map $x \mapsto \alt{x}$
  is increasing along the segment $\gamma^*$ as we move $x$ on
  $\gamma^*$ from $Z(\sigma)$ to the other endpoint. Since $\alt{x}>0$
  for all $x \in \gamma^*$, we deduce that $\gamma^*$ does not
  intersect \M. Since $V(\gamma,P) \setminus \gamma^*$ lies outside
  $\Offset P \alpha$ and $\M \subseteq \Offset P \alpha$, we deduce
  that $V(\gamma,P) \cap \M = \emptyset$ and $\gamma \not \in \DelC
  P$.

  \smallskip \styleitem{(b)}  $G_\M(K)$ has a source $\sigma$ whose
  circumcenter $Z(\sigma)$ lies below \M and
  $
  \bigcap \Below \sigma
  $
  is collapsed in $K$. Similarly, we can show that none of the
  $(d-1)$-simplices that disappear from $K$ during the collapse of
  $\tau$ in $K$ belong to $\DelC P$.

  \Step{4} We show that as we run
   $\NonCrossingVertSimp(K)$
  on $K = \Del{P,\alpha}$, when
  the algorithm terminates, then $K = \DelC P$. For this, we are going
  to prove that during the course of the algorithm, all
  $(d-1)$-simplices $\gamma \in \Del{P,\alpha} \setminus \DelC P$
  disappear at some point.  Consider a $(d-1)$-simplex $\gamma \in
  \Del{P,\alpha} \setminus \DelC P$ and define again $\gamma^*$ as in
  \eqref{eq:restricted-Voronoi-edge}:
  \begin{equation*}
    \label{eq:restricted-Voronoi-edge}
      \gamma^* = V(\gamma,P) \cap \Offset P \alpha = V(\gamma, P) \cap
  \bigcap_{p \in \gamma} B(p,\alpha).
  \end{equation*}
  Note that $\gamma^*$ is connected. Since $\gamma^*$ does not
  intersect \M and is connected, either $\gamma^*$ lies above \M or
  $\gamma^*$ lies below \M. We are going to charge any $(d-1)$-simplex
  $\gamma \in \Del{P,\alpha} \setminus \DelC P$ to a $d$-simplex
  $\sigma \in \Del{P,\alpha}$, considering  two cases:

  \smallskip \styleitem{(a)} Suppose that $\gamma^*$ lies above \M. We
  claim that, in that case, there is a $d$-simplex $\sigma \in
  \Del{P,\alpha}$ such that $\sigma \below \gamma$ with $Z(\sigma)$
  above \M.  Suppose for a contradiction that this is not true and
  $\gamma$ has no $d$-simplex of $\Del{P,\alpha}$ below it
  relative to \M. Then, Lemma~\ref{lemma:connecting-boundaries}
  implies that $\gamma$ belongs to a lower join of the form $c * \Conv
  \gamma$ with $c \in \gamma^*$ lying on the lower skin of $\Offset P
  \alpha$. Since the lower skin of $\Offset P \alpha$ lies below \M,
  we thus get that $c \in \gamma^*$ lies below \M, yielding a
  contradiction. In that case, we charge $\gamma$ to $\sigma$.

  \smallskip \styleitem{(b)} Suppose that $\gamma^*$ lies below
  \M. Similarly, one can show that there is a $d$-simplex $\sigma \in
  \Del{P,\alpha}$ such that $\gamma \above \sigma$ with
  $Z(\sigma)$ lying below \M. In that case, we charge $\gamma$ to
  $\sigma$.

  \smallskip
  At the end of Algorithm~\ref{algo:non-crossing-simplification}, every
  $d$-simplex $\sigma$ of $\Del{P,\alpha}$ has been removed at some
  point. The charging is done in such a way that, when the $d$-simplex
  $\sigma$ is removed from $K$, so are all $(d-1)$-simplices charged
  to it.  Hence, when Algorithm~\ref{algo:non-crossing-simplification}
  terminates, $K$ does not contain any $(d-1)$-simplices of
  $\Del{P,\alpha} \setminus \DelC P$ and therefore $K = \DelC P$.

  \medskip

  To summarize, Step~1 guarantees that $\NonCrossingSquash(P,\alpha)$
  corresponds to one particular execution of
  $\NaiveSquash(P,\alpha)$. Applying
  Theorem~\ref{theorem:alpha-complex}, we obtain that it returns a
  triangulation of \M, which, by Step~4, is $\DelC P$.
\end{proof}

\begin{lemma}
  \label{lemma:3D-crossing-Delaunay}
  Let $\M$ be a $C^2$ surface in $\Rspace^3$ whose reach is at least
  $\reach>0$. Let $P$ be a finite point set such that $P \subseteq \M
  \subseteq \Offset P \varepsilon$ for $\frac{\varepsilon}{\reach}
  \leq 0.225$.  Assuming no Voronoi vertices of $P$ lie on \M,
  $\DelC{P}$ is a triangulation of \M.
\end{lemma}

\begin{proof}
  Let $\varepsilon \in [0, 0.225\reach]$ and note that it is always
  possible to choose $\alpha \geq 0$ such that the pair
  $(\frac{\varepsilon}{\reach},\frac{\alpha}{\reach})$ belongs to the
  region depicted in Figure~\ref{fig:sampling-condition-3D-left} or
  equivalently such that $\varepsilon,\alpha \geq 0$ satisfy
  \itemref{item:naive-3D-sampling} in Lemma \ref{corollary:3D-theory}.
  Let $\delta=0$ and $\beta = \beta_{\varepsilon,\alpha}$. By Lemma
  \ref{lemma:3D-conditions}, the assumptions of Theorem
  \ref{theorem:alpha-complex} are satisfied.  It is thus possible to
  apply Lemma \ref{lemma:crossing-Delaunay} and deduce that $\DelC{P}$
  is a triangulation of \M.
\end{proof}

\subsection{A pure 2-dimensional simplicial complex when $d=3$}
\label{appendix:pure-simplicial-complex}

In this section, we establish conditions under which the restricted
Delaunay complex is a pure 2-dimensional complex, when $\M$ is a
smooth surface in $\Rspace^3$.

\begin{lemma}
  \label{lemma:WhenRestrictedDelaunayIsPure2D}
  Let $\M$ be a $C^2$ surface in $\Rspace^3$ whose reach is at least
  $\reach>0$. Let $P$ be a finite point set such that $P \subseteq \M
  \subseteq \Offset P \varepsilon$ for $\varepsilon \geq 0$. Let us make
  the generic assumption that all Voronoi cells of $P$ intersect \M
  transversally\footnote{We recall that two smooth submanifolds of
  $\Rspace^d$ intersect transversally if at every intersection point,
  the tangent spaces of the two submanifolds span $\Rspace^d$.}. For
  all $\varepsilon < \reach$ that satisfy
  \[
  2  \arcsin \frac{\varepsilon}{2\reach} +   \arcsin \frac{\varepsilon}{\reach} < \frac{\pi}{2}
  \]
  or, equivalently, for all $\frac{\varepsilon}{\reach} \leq 0.732\ldots$,
  $\DelR{P}$ is a pure 2-dimensional simplicial complex.
\end{lemma}

Before giving the proof, we start with some remarks.  Define the
restricted Voronoi region of $q \in P$ as the intersection of the
Voronoi region of $q$ with \M:
\[
V_\M(q,P) = V(q,P) \cap \M = \{ x \in \M \mid \|x-q\| \leq \|x-p\|,
\text{ for all $p \in P$} \}.
\]

\begin{remark}
  $\DelR{P}$ is the nerve of the collection of restricted Voronoi
  regions,
  \[
  \DelR{P} = \Nerve{\{ V_\M(q,P) \}_{q \in P}}.
  \]
\end{remark}

\begin{remark}
  \label{remark:restricted-region-small}
  If $\M \subseteq \Offset P \varepsilon$, then $V_\M(p,P) \subseteq B(p,\varepsilon)$ for all $p \in P$.
\end{remark}

\begin{proof}[Proof of Lemma \ref{lemma:WhenRestrictedDelaunayIsPure2D}]
Our generic assumption implies that no Voronoi vertices of $P$ lie on
\M and therefore $\DelR{P}$ has dimension 2 or less.

\begin{figure}[htb]
  \centering
  \def\svgwidth{.49\textwidth}
  %% Creator: Inkscape 1.3.2 (091e20e, 2023-11-25), www.inkscape.org
%% Accompanies image file 'proof-pure-vertices-left.pdf' (pdf, eps, ps)
\begingroup%
  \makeatletter%
  \providecommand\color[2][]{%
    \errmessage{(Inkscape) Color is used for the text in Inkscape, but the package 'color.sty' is not loaded}%
    \renewcommand\color[2][]{}%
  }%
  \providecommand\transparent[1]{%
    \errmessage{(Inkscape) Transparency is used (non-zero) for the text in Inkscape, but the package 'transparent.sty' is not loaded}%
    \renewcommand\transparent[1]{}%
  }%
  \providecommand\rotatebox[2]{#2}%
  \newcommand*\fsize{\dimexpr\f@size pt\relax}%
  \newcommand*\lineheight[1]{\fontsize{\fsize}{#1\fsize}\selectfont}%
  \ifx\svgwidth\undefined%
    \setlength{\unitlength}{263.70189276bp}%
    \ifx\svgscale\undefined%
      \relax%
    \else%
      \setlength{\unitlength}{\unitlength * \real{\svgscale}}%
    \fi%
  \else%
    \setlength{\unitlength}{\svgwidth}%
  \fi%
  \global\let\svgwidth\undefined%
  \global\let\svgscale\undefined%
  \makeatother%
  \begin{picture}(1,0.69078077)%
    \lineheight{1}%
    \setlength\tabcolsep{0pt}%
    \put(0,0){\includegraphics[width=\unitlength,page=1]{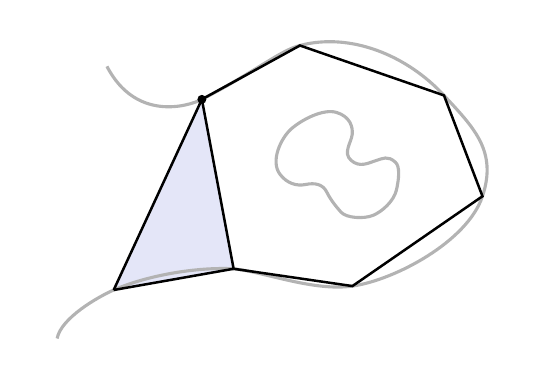}}%
    \put(0.63198455,0.36604051){\makebox(0,0)[lt]{\lineheight{1.25}\smash{\begin{tabular}[t]{l}$p$\end{tabular}}}}%
    \put(0,0){\includegraphics[width=\unitlength,page=2]{proof-pure-vertices-left.pdf}}%
    \put(0.1803362,0.58578953){\makebox(0,0)[lt]{\lineheight{1.25}\smash{\begin{tabular}[t]{l}$\M$\end{tabular}}}}%
    \put(0,0){\includegraphics[width=\unitlength,page=3]{proof-pure-vertices-left.pdf}}%
  \end{picture}%
\endgroup%
 \hfill
    \def\svgwidth{.49\textwidth}
  %% Creator: Inkscape 1.3.2 (091e20e, 2023-11-25), www.inkscape.org
%% Accompanies image file 'proof-pure-vertices-right.pdf' (pdf, eps, ps)
\begingroup%
  \makeatletter%
  \providecommand\color[2][]{%
    \errmessage{(Inkscape) Color is used for the text in Inkscape, but the package 'color.sty' is not loaded}%
    \renewcommand\color[2][]{}%
  }%
  \providecommand\transparent[1]{%
    \errmessage{(Inkscape) Transparency is used (non-zero) for the text in Inkscape, but the package 'transparent.sty' is not loaded}%
    \renewcommand\transparent[1]{}%
  }%
  \providecommand\rotatebox[2]{#2}%
  \newcommand*\fsize{\dimexpr\f@size pt\relax}%
  \newcommand*\lineheight[1]{\fontsize{\fsize}{#1\fsize}\selectfont}%
  \ifx\svgwidth\undefined%
    \setlength{\unitlength}{263.70189276bp}%
    \ifx\svgscale\undefined%
      \relax%
    \else%
      \setlength{\unitlength}{\unitlength * \real{\svgscale}}%
    \fi%
  \else%
    \setlength{\unitlength}{\svgwidth}%
  \fi%
  \global\let\svgwidth\undefined%
  \global\let\svgscale\undefined%
  \makeatother%
  \begin{picture}(1,0.69078077)%
    \lineheight{1}%
    \setlength\tabcolsep{0pt}%
    \put(0,0){\includegraphics[width=\unitlength,page=1]{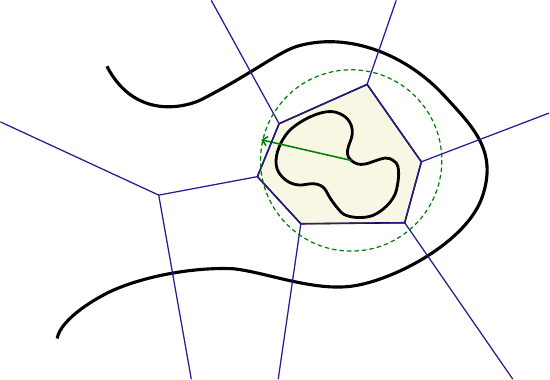}}%
    \put(0.64767789,0.46768133){\makebox(0,0)[lt]{\lineheight{1.25}\smash{\begin{tabular}[t]{l}$\M_p$\end{tabular}}}}%
    \put(0.43525727,0.44992588){\color[rgb]{0,0.50196078,0}\makebox(0,0)[lt]{\lineheight{1.25}\smash{\begin{tabular}[t]{l}$\varepsilon$\end{tabular}}}}%
    \put(0,0){\includegraphics[width=\unitlength,page=2]{proof-pure-vertices-right.pdf}}%
    \put(0.63198455,0.36604051){\makebox(0,0)[lt]{\lineheight{1.25}\smash{\begin{tabular}[t]{l}$p$\end{tabular}}}}%
    \put(0,0){\includegraphics[width=\unitlength,page=3]{proof-pure-vertices-right.pdf}}%
    \put(0.1803362,0.58578953){\makebox(0,0)[lt]{\lineheight{1.25}\smash{\begin{tabular}[t]{l}$\M$\end{tabular}}}}%
    \put(0,0){\includegraphics[width=\unitlength,page=4]{proof-pure-vertices-right.pdf}}%
  \end{picture}%
\endgroup%

    \caption{ If the restricted Delaunay complex $\DelR{P}$ has one
      vertex $p$ with no coface (left), then the connected component $\M_p$ of $\M$ passing through $p$ is contained in
      the (yellow) Voronoi region $V(p,P)$ (right).
      \label{figure:proof-pure-vertices}
    }
\end{figure}

Let us show that any vertex in $\DelR{P}$ has at least one edge as a
coface. Suppose for a contradiction that $\DelR{P}$ contains a vertex
$p \in P$ with no proper coface. Consider the connected component of
\M that passes through $p$ and denote it as $\M_p$. Because $\M_p$
passes through $p \in V(p,P)$ and does not intersect the boundary
of $V(p,P)$, we obtain that $\M_p \subseteq V(p,P)$; see Figure \ref{figure:proof-pure-vertices}. By
Remark~\ref{remark:restricted-region-small}, $V_\M(p,P) \subseteq
B(p,\varepsilon)$ and therefore
\[
\M_p \subseteq \M \cap B(p,\varepsilon).
\]
$\M_p$ being a surface without boundary, it cannot be contractible and
neither can be $\M \cap B(p,\varepsilon)$ which is a surface (possibly
with boundary) that contains $\M_p$. But, this contradicts Lemma~6 in
\cite{Dominique_Andre:2015:collapse_cech_into_triangulation} which
states that $\M \cap B(p,\varepsilon)$ is contractible, being the
non-empty intersection of a set $\M$ with a ball whose radius
$\varepsilon$ is smaller than the reach of \M.

We now focus on the non-trivial part of the proof, showing that any
edge in $\DelR{P}$ has at least one triangle as a coface. Suppose for a
contradiction that $\DelR{P}$ contains an edge $pq$ with no proper
coface. Consider the restricted Voronoi cell of $pq$:
\[
V_\M(pq,P) = V_\M(p,P) \cap V_\M(q,P) = V(pq,P) \cap \M.
\]
Recall that the Voronoi cell $V(pq,P)$ is contained in the bisector
$\Pi_{pq}$ of $p$ and $q$ and that its relative boundary is defined as
its boundary within the hyperplane $\Pi_{pq}$. We note that $pq$ is an
edge of $\DelR{P}$ with no proper coface if and only if $\M$
intersects $V(pq,P)$ and does not intersect the relative boundary of
$V(pq,P)$.

Our generic assumption implies that $\M$ intersects $V(pq,P)$ in a
smooth curve (precisely, a $1$-dimensional $C^1$ submanifold).
Thanks to Remark \ref{remark:restricted-region-small}, the curve is
contained in $B(p,\varepsilon) \cap B(q,\varepsilon)$ and therefore is
bounded.  On the other hand, the curve does not intersect the relative
boundary of $V(pq,P)$. Thus, the curve consists of one or several
loops, {\em i.e. topological circles}.  Let us choose one of these
loops which, by Jordan Theorem, is the boundary of some closed
topological disk $D$ within $\Pi_{pq}$; see Figure
\ref{figure:proof-pure-edges-3D}.  Note that
\begin{equation}
  \label{eq:D-bounded}
    D \subseteq B(p, \varepsilon) \cap B(q,\varepsilon) \cap \Pi_{pq}.  
\end{equation}

\begin{figure}[htb]
  \centering
  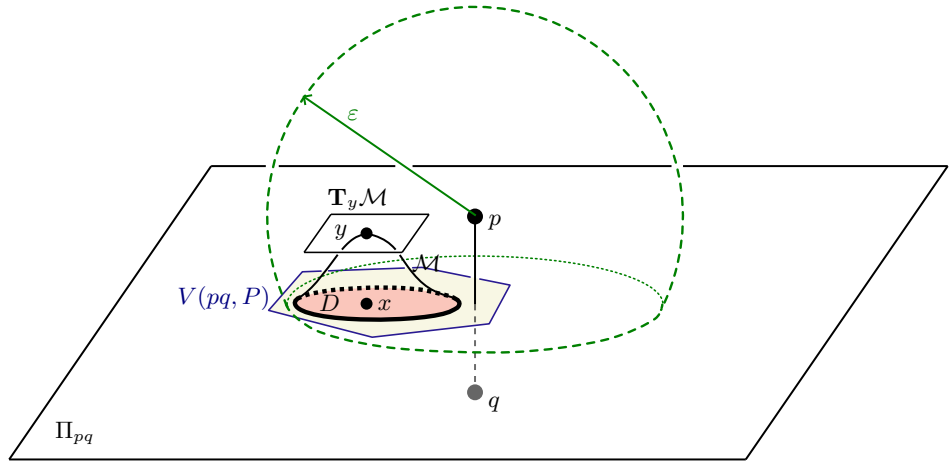
    \caption{ If $pq$ is an edge of $\DelR{P}$ with no proper coface,
      then $\M$ intersects the (yellow) Voronoi cell $V(pq,\M)$ in one
      or several loops (in bold). A key step in the proof of Lemma
      \ref{lemma:WhenRestrictedDelaunayIsPure2D} consists in showing
      that $y$ belongs to the dome (in dashed green) of
      $B(p,\varepsilon)$ above $\Pi_{pq}$.
      \label{figure:proof-pure-edges-3D}
    }
\end{figure}

By compactness, there exist a point $x \in D$ that maximizes $d(x,\M)$
and a point $y \in \M$ such that $\|x - y\| = d(x,\M)$. Furthermore,
we have
\begin{equation}\label{eq:TgtParallelToPIpq}
\angle(\Tangent y \M, \Pi_{pq}) = 0,
\end{equation}
as illustrated in Figure~\ref{figure:proof-pure-edges-3D}. Without loss of generality, we may rename
$p$ and $q$ so that $p$ and $y$ both lie on the same side of $\Pi_{pq}$;
see Figure~\ref{figure:proof-pure-edges-3D}.

We claim that $y\in B(p, \varepsilon)$.  If $y=x$, the claim is
trivial since $x \in D \subseteq B(p, \varepsilon) \cap \Pi_{pq}$. If
$y \neq x$, consider the half-line with origin at $y$ and passing
through $x$. Let $z$ be the point on this half-line whose distance to
$y$ is $\reach$; see Figure \ref{fig:EdgePQContradicts}.
By construction, the sphere centered at $z$ with radius $\reach$ is
tangent to \M at $y$. Since the reach of $\M$ is at least $\reach$,
the open ball $B^\circ(z,\reach)$ does not intersect $\M$. In
particular, it does not intersect the relative boundary of $D$ which,
by construction, is contained in \M. Hence, on one hand, $x$ belongs
to both the open disk $B^\circ(z,\reach) \cap \Pi_{pq}$ and the
topological disk $D$. On the other hand, the open disk
$B^\circ(z,\reach) \cap \Pi_{pq}$ does not meet the relative boundary
of $D$. The only possibility is that $B^\circ(z,\reach) \cap \Pi_{pq}
\subseteq D$. Combining this inclusion with the inclusion in
\eqref{eq:D-bounded}, we get that
\[
 B^\circ(z,\reach) \cap \Pi_{pq} \subseteq B(p, \varepsilon) \cap \Pi_{pq}.
\]
Denote by $\Pi_{pq}^+$ the half-space whose boundary is $\Pi_{pq}$ and
which contains $p$ (and therefore $y$). Next, we show that the above
inclusion still holds when replacing $\Pi_{pq}$ with $\Pi_{pq}^+$. In
the context of the proof of the claim, for any ball $B$ (either closed
or open), let us call the restriction of $B$ to $\Pi_{pq}^+$ the {\em
  dome} of $B$ above $\Pi_{pq}$. Because $\| x-y\| = d(x,\M) \leq \|
x-p\| \leq \varepsilon < \reach$, the dome of $B^\circ(z,\reach)$
above $\Pi_{pq}$ is less than half the ball
$B^\circ(z,\reach)$. Because $\varepsilon < \reach$, the dome of $B(p,
\varepsilon)$ above $\Pi_{pq}$ must contain the dome of
$B^\circ(z,\reach)$ above $\Pi_{pq}$:
\[
 B^\circ(z,\reach) \cap \Pi_{pq}^+ \subseteq B(p, \varepsilon) \cap \Pi_{pq}^+.
\]
Since $y$ lies on the boundary of the dome on the left side, it also
belongs to the dome on the right side.  This establishes our claim
that $y \in B(p,\varepsilon)$.

\begin{figure}[htb]
  \centering
  \scalebox{0.85}{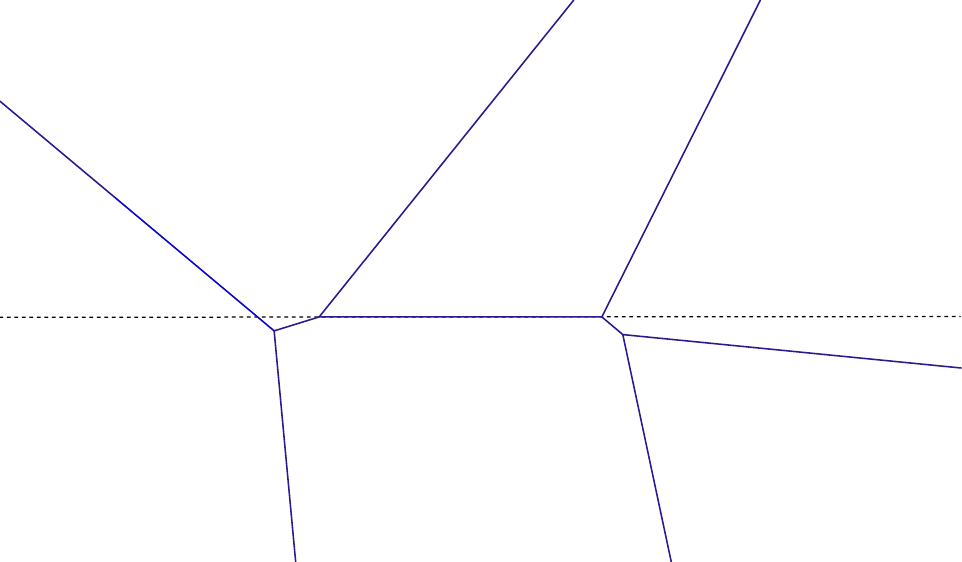}
  \caption{Notation for the proof of the claim in Lemma
    \ref{lemma:WhenRestrictedDelaunayIsPure2D}.  Note that, although
    the proof assumes $\varepsilon < \reach$, the ball
    $B(p,\varepsilon)$ is depicted as significantly larger than
    $B(z,\reach)$, which might be misleading.
	This is because the assumptions made within the proof are designed to lead to a contradiction, making the situation impossible to depict while satisfying all of those assumptions.
}
  \label{fig:EdgePQContradicts}
\end{figure}

We are now ready to reach a
contradiction.  By \cite[Corollary 3]{boissonnat2019reach}:
 \begin{equation}\label{eq:AngleTgtpy}
\angle (\Tangent y \M, \Tangent p \M) \leq 2  \arcsin \frac{\varepsilon}{2\reach}.
\end{equation}
Since $x\in B(p, \varepsilon) \cap B(q, \varepsilon) \neq \emptyset$,
we get that $||p-q|| \leq 2 \varepsilon$ and by Lemma
\ref{lemma:angleBetweenEdgeAndTangentPlane},
\begin{equation}\label{eq:AngleTgtpq}
\angle (q-p , \, \Tangent p \M) \leq \arcsin
\frac{\varepsilon}{\reach}.
\end{equation}
Since $q-p$ is orthogonal to $\Pi_{pq}$, we get that
\eqref{eq:TgtParallelToPIpq} is equivalent to $\angle (q-p, \,
\Tangent y \M) = \frac{\pi}{2}$.  This with \eqref{eq:AngleTgtpy} and
\eqref{eq:AngleTgtpq} gives:
\begin{equation*}
2  \arcsin \frac{\varepsilon}{2\reach} +   \arcsin \frac{\varepsilon}{\reach} \geq \frac{\pi}{2}.
\end{equation*}
But, this contradicts our assumption that
\[
2  \arcsin \frac{\varepsilon}{2R} +   \arcsin \frac{\varepsilon}{R} < \frac{\pi}{2}.
\]  
Hence, the assumed situation cannot occur. In other words, if an edge
belongs to $\DelR{P}$, it must have at least one triangle as a
coface.
\end{proof}

\end{document}